\documentclass[preprint]{iacrtrans}

\usepackage[utf8]{inputenc}
\usepackage[bookmarksdepth=3]{hyperref}
\usepackage{xspace}
\usepackage{enumitem}
\usepackage{booktabs}
\usepackage{multirow}
\usepackage{cleveref}
\usepackage{doi}
\usepackage{orcidlink}

\setlist[enumerate,1]{label={(\arabic*)}}

\allowdisplaybreaks

\theoremstyle{plain}
\newtheorem{thm}{Theorem}[section]
\newtheorem{lem}[thm]{Lemma}
\newtheorem{cor}[thm]{Corollary}
\newtheorem{prop}[thm]{Proposition}

\newtheorem{defn}[thm]{Definition}
\newtheorem{conj}[thm]{Conjecture}
\theoremstyle{definition}
\newtheorem{rem}[thm]{Remark}
\newtheorem{ex}[thm]{Example}
\theoremstyle{plain}

\crefname{thm}{theorem}{theorems}
\Crefname{thm}{Theorem}{Theorems}
\crefname{defn}{definition}{definitions}
\Crefname{defn}{Definition}{Definitions}
\crefname{prop}{proposition}{propositions}
\Crefname{prop}{Proposition}{Propositions}
\crefname{lem}{lemma}{lemmas}
\Crefname{lem}{Lemma}{Lemmas}
\crefname{cor}{corollary}{corollaries}
\Crefname{cor}{Corollary}{Corollaries}
\crefname{ex}{example}{examples}
\Crefname{ex}{Example}{Examples}
\crefname{rem}{remark}{remarks}
\Crefname{rem}{Remark}{Remarks}
\crefname{hyp}{hypothesis}{hypotheses}
\Crefname{hyp}{Hypothesis}{Hypotheses}
\crefname{conj}{conjecture}{conjectures}
\Crefname{conj}{Conjecture}{Conjectures}

\newcommand{\F}{\mathbb{F}}                                         
\newcommand{\Fp}{\F_{p}}                                            
\newcommand{\Fq}{\F_{q}}                                            
\newcommand{\Fqn}{\Fq^{n}}                                          

\newcommand{\abs}[1]{\left\vert #1 \right\vert}                     
\newcommand{\ceil}[1]{\left\lceil #1 \right\rceil}                  
\DeclareMathOperator{\lcm}{lcm}                                     

\newcommand{\AffSp}[2]{\mathbb{A}^{#1}_{#2}}                        
\newcommand{\ProjSp}[2]{\mathbb{P}^{#1}_{#2}}                       

\newcommand{\dehom}{\text{\normalfont deh}}                         
\newcommand{\homog}{\text{\normalfont hom}}                         
\DeclareMathOperator{\inid}{in}                                     
\DeclareMathOperator{\maxGBdeg}{max.GB.deg}                         
\DeclareMathOperator{\sat}{sat}                                     
\DeclareMathOperator{\solvdeg}{sd}                                  
\DeclareMathOperator{\reg}{reg}                                     
\newcommand{\topcomp}{\text{\normalfont top}}                       

\DeclareMathOperator{\circulant}{circ}                              
\DeclareMathOperator{\rank}{rank}                                   

\newcommand*{\degree}[1]{\deg \left( #1 \right)}                    
\DeclareMathOperator{\LC}{LC}                                       
\DeclareMathOperator{\LM}{LM}                                       
\DeclareMathOperator{\LT}{LT}                                       

\DeclareMathOperator{\gl}{GL}                                       
\newcommand{\GL}[2]{\gl_{#1} \left( #2 \right)}                     

\newcommand{\Anemoi}{\texttt{Anemoi}\xspace}

\newcommand{\Arion}{\textsf{Arion}\xspace}

\newcommand{\Ciminion}{\texttt{Ciminion}\xspace}

\newcommand{\GMiMC}{\texttt{GMiMC}\xspace}
\newcommand{\GMiMCcrf}{\texttt{GMiMC}\textsubscript{crf}\xspace}
\newcommand{\GMiMCerf}{\texttt{GMiMC}\textsubscript{erf}\xspace}

\newcommand{\Griffin}{\textsc{Griffin}\xspace}

\newcommand{\Hades}{\textsc{Hades}\xspace}
\newcommand{\Hydra}{\textsf{Hydra}\xspace}

\newcommand{\Jarvis}{\texttt{Jarvis}\xspace}

\newcommand{\LowMC}{\texttt{LowMC}\xspace}

\newcommand{\MiMC}{\texttt{MiMC}\xspace}

\newcommand{\Poseidon}{\textsc{Poseidon}\xspace}
\newcommand{\Poseidontwo}{\textsc{Poseidon}2\xspace}

\newcommand{\Rescue}{\texttt{Rescue}\xspace}

\newcommand{\ReinforcedConcrete}{\texttt{Reinforced Concrete}\xspace}

\newcommand{\Vision}{\texttt{Vision}\xspace}

\usepackage{tabularx}
\newcolumntype{P}[1]{>{\centering\arraybackslash}p{#1}}
\newcolumntype{M}[1]{>{\centering\arraybackslash}m{#1}}

\usetikzlibrary{arrows.meta,fit}

\newcommand{\MiMCfe}{
            q + 2 \cdot r - 2 \leq \reg \bigg( \mathcal{F}^\homog_\MiMC + \Big( y^q - y \cdot x_0^{q - 1} \Big) \bigg) \leq q + 2 \cdot r
}

\newcommand{\MiMCtwo}{
        4 \cdot r - 3 \leq \reg \left( \mathcal{F}_{\MiMC, 1}^\homog + \mathcal{F}_{\MiMC, 2}^\homog \right) \leq 4 \cdot r + 1
}

\newcommand{\FeistelMiMC}{
        2 \cdot r - 1 \leq \reg \left( \mathcal{F}_{\MiMC\text{-}2n/n}^\homog \right) \leq 2 \cdot r + 1
}

\newcommand{\Hashpre}{
        q + 2 \cdot r - 6 \leq \reg \bigg( \mathcal{F}_\text{preimage}^\homog + \Big( x_2^q - x_2 \cdot x_0^{q - 1} \Big) \bigg) \leq q + 2 \cdot r - 2
}

\title{Solving Degree Bounds for Iterated Polynomial Systems}
\author{Matthias Johann Steiner \orcidlink{0000-0001-5206-6579}}
\authorrunning{M.~J.~Steiner}
\institute{Alpen-Adria-Universit\"at Klagenfurt, Klagenfurt am W\"orthersee, Austria \\ \email{matthias.steiner@aau.at}}

\begin{document}

    \maketitle

    \begin{abstract}
        For Arithmetization-Oriented ciphers and hash functions Gr\"obner basis attacks are generally considered as the most competitive attack vector.
        Unfortunately, the complexity of Gr\"obner basis algorithms is only understood for special cases, and it is needless to say that these cases do not apply to most cryptographic polynomial systems.
        Therefore, cryptographers have to resort to experiments, extrapolations and hypotheses to assess the security of their designs.
        One established measure to quantify the complexity of linear algebra-based Gr\"obner basis algorithms is the so-called solving degree.
        Caminata \& Gorla revealed that under a certain genericity condition on a polynomial system the solving degree is always upper bounded by the Castelnuovo-Mumford regularity and henceforth by the Macaulay bound, which only takes the degrees and number of variables of the input polynomials into account.
        In this paper we extend their framework to iterated polynomial systems, the standard polynomial model for symmetric ciphers and hash functions.
        In particular, we prove solving degree bounds for various attacks on \texttt{MiMC}, Feistel-\texttt{MiMC}, Feistel-\texttt{MiMC}-Hash, \textsc{Hades} and \texttt{GMiMC}.
        Our bounds fall in line with the hypothesized complexity of Gr\"obner basis attacks on these designs, and to the best of our knowledge this is the first time that a mathematical proof for these complexities is provided.

        Moreover, by studying polynomials with degree falls we can prove lower bounds on the Castelnuovo-Mumford regularity for attacks on \texttt{MiMC}, Feistel-\texttt{MiMC} and Feistel-\texttt{MiMC}-Hash provided that only a few solutions of the corresponding iterated polynomial system originate from the base field.
        Hence, regularity-based solving degree estimations can never surpass a certain threshold, a desirable property for cryptographic polynomial systems.
    \end{abstract}

    \section{Introduction}\label{Sec: introduction}
With the increasing adaption of Multi-Party Computation (MPC) and Zero-Knowledge (ZK) proof systems new ciphers and hash functions are needed to implement these constructions efficiently without compromising security.
These new cryptographic primitives are commonly referred to as \emph{Arithmetization-Oriented} (AO) designs.
The main objective of AO is to minimize multiplicative complexity, the minimum number of multiplications needed to evaluate a function.
However, this comes at a cost: a very simple algebraic representation.
Examples of recently proposed AO ciphers and hash functions are \LowMC \cite{EC:ARSTZ15}, \MiMC \cite{AC:AGRRT16}, \GMiMC \cite{ESORICS:AGPRRRRS19}, \Jarvis \cite{EPRINT:AshDho18}, \Hades \cite{EC:GLRRS20}, \Poseidon \cite{USENIX:GKRRS21} and \Poseidontwo \cite{AFRICACRYPT:GraKhoSch23}, \Vision and \Rescue \cite{ToSC:AABDS20}, \Ciminion \cite{EC:DGGK21}, \ReinforcedConcrete \cite{CCS:GKLRSW22}, \Anemoi \cite{C:BBCPSVW23}, \Griffin \cite{C:GHRSWW23}, \Hydra \cite{EC:GOSW23} and \Arion \cite{Arion}.
Unfortunately, with AO an often-neglected threat reemerged in cryptography: \emph{Gr\"obner bases}.
While being a minor concern for well-established ciphers like the Advanced Encryption Standard (AES) \cite{FSE:BucPysWei06,Daemen-AES}, certain proposed AO designs have already been broken with off-the-shelf computing hardware and standard implementations of Gr\"obner bases, see for example \cite{AC:ACGKLRS19,ToSC:GKRS22}.
Therefore, to ensure computational security against Gr\"obner basis attacks cryptographers ask for tight complexity bounds of Gr\"obner basis computations \cite{ToSC:AABDS20,EPRINT:SauSze21}.

Unfortunately, the Gr\"obner basis cryptanalysis of the aforementioned AO designs is lacking mathematical rigor.
Broadly speaking, the Gr\"obner basis analysis of AO designs usually falls into two categories:
\begin{enumerate}[label=(\Roman*)]
    \item It is assumed that the polynomial system satisfies some genericity condition for which Gr\"obner basis complexity estimates are known.
    E.g., being regular or semi-regular.

    \item Empirical complexities from small scale experiments are extrapolated.
\end{enumerate}
In this paper on the other hand, we present a rigor mathematical formalism to derive provable complexity estimates for cryptographic polynomial systems.
In particular, we rigorously obtain Gr\"obner basis complexity estimates for various attacks on \MiMC, Feistel-\MiMC, Feistel-\MiMC-Hash, \Hades and \GMiMC.
We note that our bounds fall in line with the hypothesized cost of Gr\"obner basis attacks on these designs (see \cite[\S 4.3]{EC:GLRRS20} and \cite[\S 4.1.1]{ESORICS:AGPRRRRS19}).
To the best of our knowledge these are the first rigor mathematical proofs for the Gr\"obner basis cryptanalysis of these designs.
Moreover, for \MiMC, Feistel-\MiMC and Feistel-\MiMC-Hash we prove limitations of our complexity estimations, i.e., we derive lower bounds which can never be surpassed by our estimation method.

The cryptographic constructions of our interest all follow the same design principle.
Let $\Fq$ be a finite field with $q$ elements and let $n \geq 1$ be an integer, one chooses a round function $\mathcal{R}: \Fqn \times \Fqn \to \Fqn$, which depends on the input variable $\mathbf{x}$ and the key variable $\mathbf{y}$, and then iterates it $r$ times with respect to the input variable.
Such a design admits a very simple model of keyed iterated polynomials
\begin{equation}
    F_\mathcal{R} (\mathbf{x}_{i - 1}, \mathbf{y}) - \mathbf{x}_i = \mathbf{0},
\end{equation}
where $F_\mathcal{R}$ denotes the polynomial vector representing the round function $\mathcal{R}$, the $\mathbf{x}_i$'s intermediate state variables, $\mathbf{y}$ the key variable and $\mathbf{x}_0, \mathbf{x}_r \in \Fqn$ a plain/ciphertext pair given by the encryption function.
This leads us to standard Gr\"obner basis attacks on ciphers which proceed in four steps:
\begin{enumerate}
    \item Model the cipher function with an iterated system of polynomials.

    \item\label{Step 2} Compute a Gr\"obner basis with respect to an efficient term order, e.g., the degree reverse lexicographic order.

    \item\label{Step 3} Perform a term order conversion to an elimination order, e.g., the lexicographic order.

    \item\label{Step 4} Solve the univariate equation.
\end{enumerate}
Let us for the moment assume that a Gr\"obner basis has already been found and focus on the complexity of the remaining steps.
Let $I \subset \Fq [x_1, \dots, x_n]$ be a zero-dimensional ideal modeling a cipher, and denote with $d = \dim_{\Fq} \left( \Fq [x_1, \dots, x_n] / I \right)$ the $\Fq$-vector space dimension of the quotient space.
With the original FGLM algorithm \cite{Faugere-FGLM} the complexity of term order conversion is $\mathcal{O} \left( n \cdot d^3 \right)$, but improved versions with probabilistic methods achieve $\mathcal{O} \left( n \cdot d^\omega \right)$ \cite{Faugere-SubCubic}, where $2 \leq \omega < 2.37286$ \cite{SODA:AlmWil21}, and sparse linear algebra algorithms \cite{Faugere-SparseFGLM} achieve $\mathcal{O} \left( \sqrt{n} \cdot d^{2 + \frac{n - 1}{n}} \right)$.
To extract the $\Fq$-valued roots of the univariate polynomial most efficiently we compute its greatest common divisor with the field equation $x^q - x$ via the algorithm of Bariant et al.\ \cite[\S3.1]{ToSC:BBLP22}.
The complexity of this step is then
\begin{equation}\label{Equ: efficient factoring}
    \mathcal{O} \Big( d \cdot \log (q) \cdot \log (d) \cdot \log \big( \log (d) \big) +  d \cdot \log (d)^2 \cdot \log \big( \log (d) \big) \Big),
\end{equation}
provided that $d \leq q$ otherwise one has to replace the roles of $d$ and $q$ in the complexity estimate.

Furthermore, in \cite{Faugere-Stark} it was proven that one can also use $d$ to upper bound the complexity of linear algebra-based Gr\"obner basis algorithms.
Since $d$ is in general not known one has to estimate $d$ via the B\'ezout bound.

To the best of our knowledge, the aforementioned AO designs all admit a very high quotient space dimension.
Hence, to improve the capabilities of Gr\"obner basis attacks one must reduce this dimension.
For this problem we have two generic approaches:
\begin{enumerate}[label=(\roman*)]
    \item Alter the standard representation, e.g., choose polynomials in the model which approximate the round function with high probability.
    This approach was successfully deployed in \cite{AC:ACGKLRS19,ToSC:GKRS22}.

    \item Add polynomials to the system to remove parasitic solutions that lie in algebraic closure.
    E.g., the polynomial system for an additional plain/ciphertext pair or the field equations.
    This approach is the concern of this paper.
\end{enumerate}
If one successfully filters all solutions from the algebraic closure, then one expects that steps \ref{Step 3} and \ref{Step 4} are not a major concern anymore.
Therefore, we need tight estimates for the complexity of Gr\"obner basis computations.

\subsection{Contributions \& Related Work}\label{Sec: contributions}
Our main tool to bound the complexity of Gr\"obner basis computations will be the \emph{solving degree} of \emph{linear algebra-based Gr\"obner basis algorithms} which was first formalized in \cite{Ding-SolvingDegree}.
Linear algebra-based Gr\"obner basis algorithms perform Gaussian elimination on matrices associated to a polynomial system.
Given the number of equations, the number of variables and the solving degree one can then estimate the maximal size of these matrices and henceforth also the cost of Gaussian elimination.
In \cite{Caminata-SolvingPolySystems} the solving degree was upper bounded via the \emph{Castelnuovo-Mumford regularity} if the polynomial system is in \emph{generic coordinates}.
This genericity notion can be traced back to the influential work of Bayer \& Stillman \cite{BayerStillman}.
In essence, a polynomial system $\mathcal{F} = \{ f_1, \dots, f_m \} \subset P = K [x_1, \dots, x_n]$ is in generic coordinates if its homogenization $\mathcal{F}^\homog = \left\{ f_1^\homog, \dots, f_m^\homog \right\} \subset P [x_0]$ does not admit a solution with $x_0 = 0$ in the projective space $\ProjSp{n}{K}$, where $x_0$ denotes the homogenization variable.
Moreover, the Castelnuovo-Mumford regularity is always upper bounded by the \emph{Macaulay bound} \cite[Theorem~1.12.4]{Chardin-Regularity}.
Hence, if a polynomial system is in generic coordinates, then we can estimate the complexity of a Gr\"obner basis computation via the degrees of the input polynomials.

Our paper is divided into two parts.
In the first part (\Cref{Sec: Preliminaries,Sec: characterization generic coordinates,Sec: lex Groebner Basis,Sec: applications}), we develop a rigor framework for complexity estimates of Gr\"obner attacks on \MiMC, Feistel-\MiMC, Feistel-\MiMC-Hash, \Hades and \GMiMC.
To streamline the application of the technique developed by Caminata \& Gorla, we prove in \Cref{Th: generic coordinates and highest degree components} that a polynomial system is in generic coordinates if and only if it admits a finite \emph{degree of regularity} \cite{Bardet-Complexity}.
This in turn permits efficient proofs that the keyed iterated polynomial systems of \MiMC, Feistel-\MiMC, Feistel-\MiMC-Hash, \Hades and \GMiMC are in generic coordinates.

In the second part (\Cref{Sec: satiety and polynomials with degree falls,Sec: lower bounds}), we study \emph{polynomials with degree falls}.
For an inhomogeneous polynomial system $\mathcal{F} = \{ f_1, \dots, f_m \} \subset K [x_1, \dots, x_m]$, we say that a polynomial $f \in (\mathcal{F})$ has a degree fall in $d > \degree{f}$, if it cannot be constructed below degree $d$ via $\mathcal{F}$, i.e.\ there does not exist a sum $f = \sum_{i = 1}^{m} g_i \cdot f_i$ such that $\degree{g_i \cdot f_i} < d$ for all $i$.
We define the \emph{last fall degree} as the largest integer $d$ for which there exists a polynomial $f \in (\mathcal{F})$ with a degree fall in $d$.
For polynomial systems in generic coordinates we prove that the last fall degree is equal to the \emph{satiety} of $\mathcal{F}^\homog$ (\Cref{Th: last fall degree finite in generic coordinates}).
Moreover, it is well-known that the satiety of $\mathcal{F}^\homog$ is always upper bounded by the Castelnuovo-Mumford regularity of $\mathcal{F}^\homog$.
Therefore, if we find a polynomial with a degree fall in $(\mathcal{F})$ then we immediately have a lower bound for the Castelnuovo-Mumford regularity of $\mathcal{F}^\homog$.
As consequence one then has a limit on the capabilities of Castelnuovo-Mumford regularity-based complexity estimates.

We note that a different notion of last fall degree was already introduced by Huang et al.\ \cite{C:HuaKosYeo15,Huang-LastFallDegree}.
Therefore, in \Cref{Rem: Huang's last fall degree} we discuss the difference between Huang et al.'s and our notion of last fall degree.

Let \MiMC with $r$ rounds be defined over $\Fq$ and assume that the \MiMC polynomial systems have fewer than three solutions in $\Fq$, applying our bounds we obtain the following ranges on the Castelnuovo-Mumford regularity.
For \MiMC and the field equation for the key variable we have, see \Cref{Ex: MiMC solving degree I,Ex: MiMC solving degree II},
\begin{equation}
    \MiMCfe.
\end{equation}
For the two plain/ciphertext attack on \MiMC we have, see \Cref{Ex: MiMC two plaintext attack I,Ex: MiMC two plaintext solving degree II},
\begin{equation}
    \MiMCtwo.
\end{equation}
For a Feistel-$2n/n$ network based on the \MiMC round function we have, see \Cref{Ex: Feistel MiMC I,Ex: Feistel MiMC solving degree II},
\begin{equation}
    \FeistelMiMC.
\end{equation}
For a Feistel-$2n/n$ network operated in sponge mode \cite{EC:BDPV08} based on the \MiMC round function we have for the preimage attack, see \Cref{Ex: Feistel-MiMC-Hash preimage I,Ex: Feistel-MiMC-Hash preimage II},
\begin{equation}
    \Hashpre.
\end{equation}
Arguably, the bounds that include the size of the finite field $q$ do not have direct cryptographic significance.
We note that these bounds can be significantly improved by an auxiliary division by remainder computation, see the discussions after \Cref{Ex: MiMC solving degree I,Ex: MiMC solving degree II,Ex: Feistel-MiMC-Hash preimage I,Ex: Feistel-MiMC-Hash preimage II}.
We restricted our analysis to the field equation due to generic treatment as well as simple algebraic representations.
Moreover, we point out that our analysis of \MiMC polynomial system serves as role model to showcase that tight complexity estimates for cryptographic polynomial systems are achievable without the evasion to unproven hypotheses.

\subsubsection{Comparison With Existing Cryptanalysis}
In this paper we derive various proven Gr\"obner basis complexity estimates for the \MiMC family, \GMiMC and \Hades.
Let us now shortly discuss how these estimates relate to established cryptanalysis of these designs.
In \Cref{Tab: complexity estimates} we collect our complexity estimates, see \Cref{Tab: MiMC sample values,Tab: MiMC two plaintext sample values,Tab: Feistel MiMC sample values,Tab: Hades values,Tab: GMiMC complexity}, next to the estimates of established attacks that are closely related to our Gr\"obner basis attacks.

The attack on \MiMC with a field equation (first three rows in the \MiMC row in \Cref{Tab: complexity estimates}) can be considered as sparse low degree representation of the greatest common divisor (GCD) attack on \MiMC \cite[\S 4.2]{AC:AGRRT16}.
In the GCD attack with a known plain/ciphertext attack one represents the \MiMC encryption function as univariate polynomial in the key variable $y$ and then computes the GCD with the field equation $y^q - y$.
The number of \MiMC rounds is chosen so that $r \geq \log_3 \left( q \right)$, where $q$ is the size of the underlying finite field, to avoid an interpolation attack \cite{SAC:LiPre19}.
So the complexity of the GCD computation can be estimated as $\mathcal{O} \Big( d \cdot \log (d)^2 \cdot \log \big( \log (d) \big) \Big)$ with $d = 3^r$ (or $d = q$ if one considers the first division by remainder computation in the GCD algorithm to be for free).
If we do not consider the construction of the univariate polynomial to be for free, we can refine this estimate.
The keyed iterated \MiMC polynomial system is already a Gr\"obner basis, so the univariate polynomial can be constructed via the probabilistic FGLM algorithm \cite{Faugere-SubCubic} which has complexity $\mathcal{O} \left( n \cdot d^\omega \right)$, and for key extraction we can use the efficient factoring algorithm of Bariant et al.\ whose complexity is given in \Cref{Equ: efficient factoring}.
In the Gr\"obner basis attack on the other hand, the univariate \MiMC encryption function is decomposed into its $r$ round functions of degree $3$ and together with the field equation the Gr\"obner basis is computed.
As \Cref{Tab: complexity estimates} shows, \MiMC achieves a security level of at least $128$ bits for various field sizes when the sparse low degree representation is used to mount a key recovery attack with the field equation.

Alternatively to the GCD with the field equation, one can consider two plain/cipher\-texts to set up two univariate encryption polynomials and compute their GCD.
As before, we can represent the encryption functions with $r$ sparse polynomials of degree $3$ respectively which share the key variable.
A similar two plain/ciphertext attack was investigated by Albrecht et al.\ \cite[\S 6.1]{AC:ACGKLRS19}.
Since an iterated \MiMC polynomial system is already Gr\"obner basis, they proposed to run the FGLM algorithm twice to construct two univariate polynomials in the key variable and then compute their GCD.
This approach is obviously equivalent to the standard two plain/ciphertext GCD attack on \MiMC, only difference is that Albrecht et al.\ did not consider the univariate polynomial construction to be for free.
Note that Albrecht et al.'s estimate can be refined by again utilizing the probabilistic FGLM algorithm as well as Bariant et al.'s factoring technique.\footnote{\label{Note: GCD}
    In \Cref{Tab: complexity estimates} we still use the standard GCD complexity estimate, since the \MiMC two plain/ciphertext and the \MiMC-$2n/n$ Gr\"obner basis attacks do not depend on the underlying field while Bariant et al.'s method does.
    }
On the other hand, we will discuss in \Cref{Ex: MiMC two plaintext attack I} that the joint polynomial system removes almost all superfluous solutions coming from the algebraic closure of $\Fq$.
Hence, the complexity of running FGLM on the joint system can be neglected after a Gr\"obner basis has been found.
As \Cref{Tab: complexity estimates} shows, \MiMC achieves a security level of at least $128$ bits for the two plain/ciphertexts Gr\"obner basis computation already for $50$ rounds.

For \MiMC-$2n/n$ one utilizes a two branch Feistel network to encrypt two field elements with one field element.
As consequence, one can represent the left and the right branch as univariate polynomials in the key variable of degrees $3^r$ and $3^{r - 1}$  respectively.
So we can again utilize the GCD to recover the key.
In \Cref{Prop: Feistel Groebner bases} we find a DRL Gr\"obner basis for \MiMC-$2n/n$ when the output of the right branch is ignored.
Moreover, the univariate polynomials that represent the left and the right branch are again present in the LEX Gr\"obner basis.
So once again, we can refine the complexity of this attack via the probabilistic FGLM algorithm and Bariant et al.'s factoring algorithm.\textsuperscript{\ref{Note: GCD}}
On the other hand, similar to the two plain/ciphertext attack on \MiMC the Gr\"obner basis computation on \MiMC-$2n/n$ removes almost all superfluous solutions coming from the algebraic closure of $\Fq$, see \Cref{Ex: Feistel MiMC I}.
So the complexity of term order conversion via FGLM can again be ignored.
As \Cref{Tab: complexity estimates} shows, \MiMC-$2n/n$ achieves a security level of at least $128$ bits against Gr\"obner basis computations already for $50$ rounds.

For Feistel-\MiMC-Hash one utilizes the \MiMC-$2n/n$ permutation in sponge mode, though for the hash function we have only one generic choice for the second polynomial to mount a GCD attack: the field equation.
Again, the complexity estimate of the GCD attack can be refined via the probabilistic FGLM algorithm and Bariant et al.'s factoring method.
As \Cref{Tab: complexity estimates} shows, Feistel-\MiMC-Hash also achieves a security level of at least $128$ bits for various field sizes with respect to the Gr\"obner basis computations with the field equation.

\Hades is a family of Substitution-Permutation Network (SPN) ciphers targeted for MPC applications.
In the \Hades proposal the designers analyze the keyed iterated polynomial system for the resistance against Gr\"obner basis attacks \cite[\S E.3]{EPRINT:GLRRS19}.\footnote{The keyed iterated polynomial model is called \emph{second strategy} in the \Hades proposal.}
We revisit this modeling, in particular we prove in \Cref{Th: SPN generic coordinates} that for a single plain/ciphertext pair one can produce a \Hades DRL Gr\"obner basis via affine transformations.
Moreover, any DRL Gr\"obner basis immediately implies being in generic coordinates (\Cref{Cor: DRL Groebner basis in generic coordinates}), so after the affine transformations we have proven complexity estimates for \emph{any} Gr\"obner basis computation on \Hades.
The \Hades designers on the other hand had to assume that the polynomial systems are generic in the sense of Fr\"oberg's conjecture \cite{Froeberg-Conjecture,Pardue-Generic} to derive complexity estimates.
Moreover, with the property of being in generic coordinates we can reproduce the complexity estimate of the designers as minimal baseline for all DRL Gr\"obner basis computations on the iterated polynomial model of \Hades.
Therefore, our Gr\"obner basis complexity estimates coincide with the cryptanalysis of the \Hades designers.
In \cite[Table~1]{EC:GLRRS20} round numbers for \Hades proposed, the \Hades parameters in \Cref{Tab: complexity estimates} are chosen so that every instance in \cite[Table~1]{EC:GLRRS20} exceeds at least one instance in \Cref{Tab: complexity estimates}.
As \Cref{Tab: complexity estimates} shows, all proposed \Hades instances achieve at least $128$ bits of security with respect to Gr\"obner basis computations.

Finally, to the best of our knowledge the keyed iterated polynomial system has not been considered for \GMiMC in the literature before.
The \GMiMC designers only considered models where the encryption function is represented in $n$ key variables for known plain/ciphertext pairs.
Moreover, they assumed that \GMiMC polynomial systems behave like generic polynomial systems in the sense of Fr\"oberg's conjecture \cite{Froeberg-Conjecture,Pardue-Generic} to derive complexity estimates.
For \GMiMC with contracting round function (crf) they derived the estimate $\binom{n + 3^{r - 2 \cdot n + 2}}{3^{r - 2 \cdot n + 2}}^\omega$ \cite[\S 4.1.2]{ESORICS:AGPRRRRS19}, and for \GMiMC with expanding round function (erf) they derived the estimate $\binom{n + 3^{r - n}}{3^{r - n}}^\omega$ \cite[\S C.3]{EPRINT:AGPRRRRS19} for Gr\"obner basis computations.
On the other hand, in \Cref{Ex: GMiMC} we will see that \GMiMCcrf and \GMiMCerf share the same complexity estimate for the keyed iterated polynomial system, provided that they are in generic coordinates.
In particular, the complexity estimate does not depend on the number of branches $n$.
Moreover, being in generic coordinates for \GMiMC can be verified by computing the rank of a linear equation system, see \Cref{Th: Feistel generic generators criteria}.
As \Cref{Tab: complexity estimates} shows, $50$ rounds are sufficient to achieve at least $128$ bits of security for \GMiMC.
\begin{table}[H]
    \centering
    \caption{Comparison of Gr\"obner basis complexity estimates for \MiMC, \MiMC-$2n/n$, Feistel-\MiMC-Hash, \Hades and \GMiMC with established cryptanalysis.
             With $r$ we denote the number of rounds of a primitive, with $n$ the number of blocks, with $d$ the degree of a power permutation and with $m$ the number of samples for an attack.
             The total number of $\Hades$ rounds is given by $r = 2 \cdot r_f + r_p$.
             For all complexities the linear algebra constant $\omega = 2$ has been used.}
    \label{Tab: complexity estimates}
    \resizebox{\linewidth}{!}{
        \begin{tabular}{ c | c | M{25mm} || c | c }
            \toprule
            & & & \multicolumn{2}{ c }{Established Cryptanalysis} \\
            \midrule
            Primitive & Parameters & Gr\"obner Basis Complexity (bits) & Complexity (bits) & Attack Strategy \\
            \midrule

            \multirow{5}{*}{\MiMC} & $\log_2 \left( q \right) = 64$, $r = 50$  & $337.5$ & $164.1$ & Probabilistic FGLM + Efficient factoring \\
            & $\log_2 \left( q \right) = 128$, $r = 81$  & $527.4$ & $263.1$ & Probabilistic FGLM + Efficient factoring \\
            & $\log_2 \left( q \right) = 256$, $r = 162$ & $1156.2$ & $520.9$ & Probabilistic FGLM + Efficient factoring\\
            & $r = 10$, $m = 2$ & $99.4$ & $36.0$ & Probabilistic FGLM + GCD \\
            & $r = 50$, $m = 2$ & $538.1$ & $165.1$ & Probabilistic FGLM + GCD \\
            \midrule

            \multirow{2}{*}{\MiMC-$2n/n$} & $r = 10$ & $48.6$ & $35.0$ & Probabilistic FGLM + GCD \\
            & $r = 50$ & $266.7$ & $164.1$ & Probabilistic FGLM + GCD \\
            \midrule

            \multirow{3}{*}{Feistel-\MiMC-Hash} & $\log_2 \left( q \right) = 64$, $r = 51$ & $337.5$ & $167.3$ & Probabilistic FGLM + Efficient factoring \\
            & $\log_2 \left( q \right) = 128$, $r = 82$  & $527.4$ & $266.2$ & Probabilistic FGLM + Efficient factoring \\
            & $\log_2 \left( q \right) = 256$, $r = 163$ & $1156.2$ & $524.0$ & Probabilistic FGLM + Efficient factoring \\
            \midrule

            \multirow{6}{*}{\Hades} & $r_f = 3$, $r_p = 13$, $n = 2$, $d = 3$ & $130.0$ & $130.0$ & Gr\"obner basis computation \\
            & $r_f = 4$, $r_p = 10$, $n = 2$, $d = 3$ & $135.4$ & $135.4$ & Gr\"obner basis computation \\
            & $r_f = 5$, $r_p = 5$, $n = 2$, $d = 3$  & $130.0$ & $130.0$ & Gr\"obner basis computation \\
            & $r_f = 3$, $r_p = 10$, $n = 2$, $d = 5$ & $149.0$ & $149.0$ & Gr\"obner basis computation \\
            & $r_f = 4$, $r_p = 10$, $n = 2$, $d = 5$ & $177.5$ & $177.5$ & Gr\"obner basis computation \\
            & $r_f = 5$, $r_p = 4$, $n = 2$, $d = 5$  & $163.3$ & $163.3$ & Gr\"obner basis computation \\
            \midrule

            \multirow{6}{*}{\GMiMC} & $r = 10$, $n = 3$, $d = 3$ & $48.6$ & crf: $51.9$, erf: $61.4$ & Gr\"obner basis computation  \\
            & $r = 25$, $n = 3$, $d = 3$ & $130.0$ & crf: $194.5$, erf: $204.0$ & Gr\"obner basis computation  \\
            & $r = 50$, $n = 3$, $d = 3$ & $266.7$ & crf: $432.3$, erf: $441.8$ & Gr\"obner basis computation  \\
            & $r = 10$, $n = 3$, $d = 5$ & $63.5$ & crf: $78.4$, erf: $92.4$ & Gr\"obner basis computation  \\
            & $r = 25$, $n = 3$, $d = 5$ & $170.5$ & crf: $287.4$, erf: $301.3$ & Gr\"obner basis computation  \\
            & $r = 50$, $n = 3$, $d = 5$ & $350.0$ & crf: $635.7$, erf: $649.6$ & Gr\"obner basis computation  \\

            \bottomrule
        \end{tabular}
    }
\end{table}

\subsubsection{Organization of the Paper}
In \Cref{Sec: Preliminaries} we will formally introduce univariate keyed iterated polynomial systems (\Cref{Sec: keyed iterated polynomial systems}), the \MiMC cipher, Feistel-$2n/n$ networks, and recall required definitions and results for the solving degree (\Cref{Sec: solving degree}) and generic coordinates (\Cref{Sec: solving degree and Castelnuovo-Mumford regularity}).
In \Cref{Sec: characterization generic coordinates} we prove that being in generic coordinates is equivalent for the ideal of the highest degree components to be zero-dimensional (\Cref{Th: generic coordinates and highest degree components}).
Moreover, we prove that a large class of univariate keyed iterated polynomial systems, including \MiMC polynomial systems, is already in generic coordinates (\Cref{Th: iterated system generic coordinates}).
As preparation for our bounds on the solving degree we study in \Cref{Sec: lex Groebner Basis} properties of the lexicographic Gr\"obner basis of the univariate keyed iterated polynomial system and Feistel-$2n/n$.
In \Cref{Sec: applications} we finally provide upper bounds for the solving degree of various attacks on \MiMC and \MiMC-$2n/n$.
In \Cref{Sec: multivariate ciphers} we extend our framework to multivariate ciphers, in particular we investigate when the keyed iterated polynomial systems for Substitution-Permutation and generalized Feistel Networks are in generic coordinates.
With our formalism we can then demonstrate that the security analysis of \Hades and \GMiMC against Gr\"obner basis attacks is indeed mathematically sound.
In \Cref{Fig: relation graph first part} we provide a directed graph to illustrate the derivation of the main results of the first part of the paper.
\begin{figure}[H]
    \centering
    \caption{Graphical overview for the development of solving degree upper bounds.}
    \label{Fig: relation graph first part}
    \resizebox{\textwidth}{!}{
        \begin{tikzpicture}
            \coordinate (a) at (0,0);
            \coordinate (b) at (-3,0);
            \coordinate (c) at (-1.5,-1.5);
            \coordinate (d) at (-1.5,-3);
            \coordinate (e) at (-6,-1.5);
            \coordinate (f) at (-6,-5);
            \coordinate (g) at (-8,-7.5);
            \coordinate (h) at (-8,-9.5);
            \coordinate (i) at (2.5,-7.5);
            \coordinate (j) at (-2.5,-7.5);
            \coordinate (k) at (2.5,-9.5);
            \coordinate (l) at (-2.5,-9.5);

            \draw (a) node[draw, align=left] (A) {\Cref{Def: solving degree}: \\ Solving degree};
            \draw (b) node[draw, align=left] (B) {\Cref{Def: generic coordinates}: \\ Generic coordinates};
            \draw (c) node[draw, align=left] (C) {\Cref{Th: solvdeg and CM-regularity}: \\ Solving degree \& regularity};
            \draw (d) node[draw, align=left] (D) {\Cref{Cor: Macauly bound}: \\ Macaulay bound};
            \draw (e) node[draw, align=left] (E) {\Cref{Sec: Caminata Gorla technique}: \\ Caminata-Gorla technique};
            \draw (f) node[draw, align=left] (F) {\Cref{Sec: characterization generic coordinates}: Characterization \\ of generic coordinates};
            \draw (g) node[draw, align=left] (G) {\Cref{Sec: applications}: \\upper bounds for attacks \\ on \MiMC};
            \draw (h) node[draw, align=left] (H) {\Cref{Sec: lex Groebner Basis}: LEX \& DRL Gr\"obner \\ bases of keyed iterated polynomial \\ systems};
            \draw (i) node[draw, align=left] (I) {\Cref{Sec: Feistel}: Feistel cipher \\ in generic coordinates};
            \draw (j) node[draw, align=left] (J) {\Cref{Sec: SPN}: SPN cipher \\ in generic coordinates};
            \draw (k) node[draw, align=left] (K) {\Cref{Ex: GMiMC}: \GMiMC};
            \draw (l) node[draw, align=left] (L) {\Cref{Ex: Hades}: \Hades};

            \node[draw, dashed, inner sep=3mm, label={[align=left]above:\Cref{Sec: Preliminaries}: Solving degree \& generic coordinates}, fit=(A) (E) (D) (E)] {};
            \node[draw, dashed, inner sep=3mm, label={[align=left]below:\Cref{Sec: multivariate ciphers}: Multivariate ciphers \\ in generic coordinates}, fit=(I) (L) (I) (J)] {};

            \draw [fill=black] (F) ++(0,-1) circle[radius=2pt];
            \draw [fill=black] (F) ++(0,-1) -- ++(3.5,0) circle[radius=2pt];
            \draw [fill=black] (D) ++(0,-1.5) ++(1.5,0) ++(0,-5) circle[radius=2pt];

            \draw[arrows = {-Stealth[]}] (A) -- ++(0,-1);
            \draw[arrows = {-Stealth[]}] (B) -- ++(0,-1);
            \draw[arrows = {-Stealth[]}] (C) -- (D);
            \draw[arrows = {-Stealth[]}] (B) -- ++(-3,0) -- (E);
            \draw[arrows = {-Stealth[]}] (E) -- (F);
            \draw[arrows = {-Stealth[]}] (F) -- ++(0,-1) -- ++(-2,0) -- (G);
            \draw[arrows = {-Stealth[]}] (H) -- (G);
            \draw[arrows = {-Stealth[]}] (F) -- ++(0,-1) -- ++(3.5,0) -- (J);
            \draw[arrows = {-Stealth[]}] (F) -- ++(0,-1) -- ++(3.5,0) -- ++(5,0) -- (I);
            \draw[arrows = {-Stealth[]}] (J) -- (L);
            \draw[arrows = {-Stealth[]}] (I) -- (K);
            \draw[arrows = {-Stealth[]}] (D) -- ++(0,-1.5) -- ++(1.5,0) -- ++(0,-5) -- (K);
            \draw[arrows = {-Stealth[]}] (D) -- ++(0,-1.5) -- ++(1.5,0) -- ++(0,-5) -- (L);
            \draw[arrows = {-Stealth[]}] (D) -- ++(-9,0) -- ++(0,-4.5) -- (G);
        \end{tikzpicture}
    }
\end{figure}
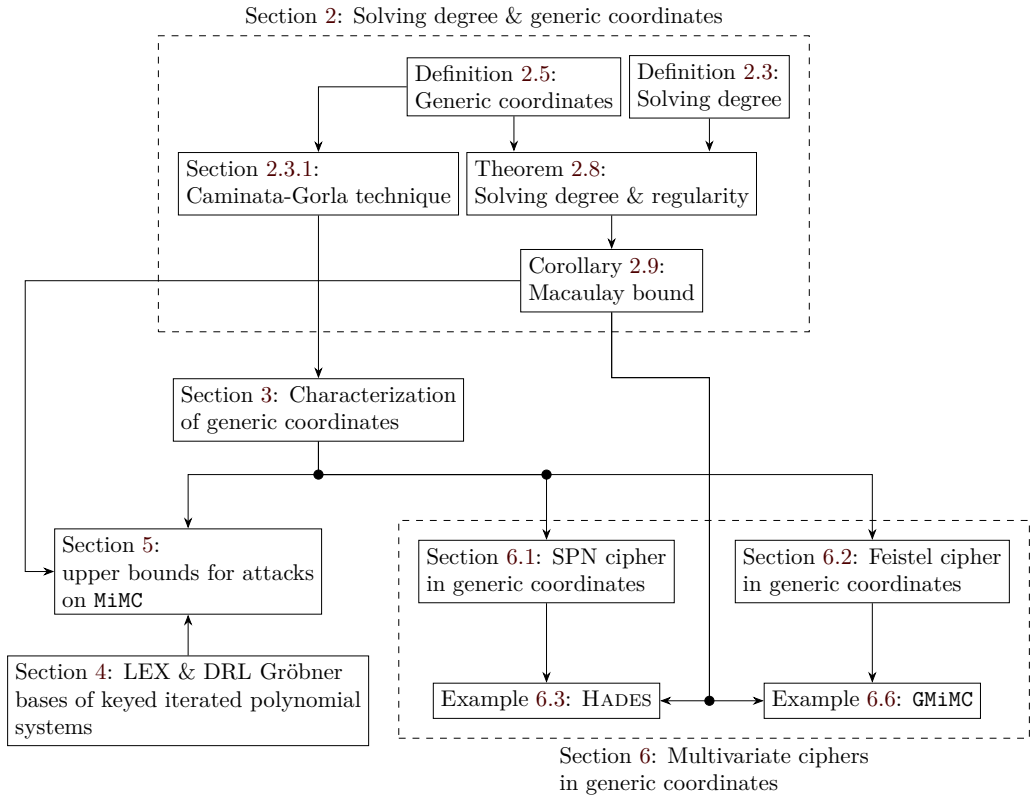

In \Cref{Sec: satiety and polynomials with degree falls} we investigate polynomials with degree falls and the last fall degree.
In particular, we establish that for a polynomial system in generic coordinates the last fall degree is equal to the satiety (\Cref{Th: last fall degree finite in generic coordinates}).
In \Cref{Sec: lower bounds} we construct polynomials with degree falls for the keyed iterated polynomial systems for univariate ciphers and Feistel-$2n/n$.
Finally, this yields regularity lower bounds for various attacks on \MiMC, Feistel-\MiMC and Feistel-\MiMC-Hash.
In \Cref{Fig: relation graph second part} we provide a directed graph to illustrate the derivation of the main results of the second part of the paper.
\begin{figure}[H]
    \centering
    \caption{Graphical overview for the development of satiety lower bounds.}
    \label{Fig: relation graph second part}
    \resizebox{0.85\textwidth}{!}{
        \begin{tikzpicture}
            \coordinate (a) at (0,0);
            \coordinate (b) at (-3,0);
            \coordinate (c) at (-1.5,-1.5);
            \coordinate (d) at (-7,-1.5);
            \coordinate (e) at (-1.5,-3.5);
            \coordinate (f) at (-8,-3.5);

            \draw (a) node[draw, align=left] (A) {\Cref{Def: last fall degree}: \\ Last fall degree};
            \draw (b) node[draw, align=left] (B) {\Cref{Def: satiety}: \\ Satiety};
            \draw (c) node[draw, align=left] (C) {\Cref{Th: last fall degree finite in generic coordinates}: \\ Last fall degree \& satiety};
            \draw (d) node[draw, align=left] (D) {\Cref{Sec: solving degree and Castelnuovo-Mumford regularity}: \\ Generic coordinates};
            \draw (e) node[draw, align=left] (E) {\Cref{Sec: lower bounds}: Last fall degree lower \\ bounds for attacks on \MiMC};
            \draw (f) node[draw, align=left] (F) {\Cref{Sec: lex Groebner Basis}: LEX \& DRL Gr\"obner \\ bases of keyed iterated polynomial \\ systems};

            \node[draw, dashed, inner sep=3mm, label={[align=left]above:\Cref{Sec: satiety and polynomials with degree falls}: Satiety \& polynomials \\ with degree falls}, fit=(A) (C) (A) (B)] {};

            \draw[arrows = {-Stealth[]}] (A) -- ++(0,-1);
            \draw[arrows = {-Stealth[]}] (B) -- ++(0,-1);
            \draw[arrows = {-Stealth[]}] (D) -- (C);
            \draw[arrows = {-Stealth[]}] (C) -- (E);
            \draw[arrows = {-Stealth[]}] (F) -- (E);
        \end{tikzpicture}
    }
\end{figure}
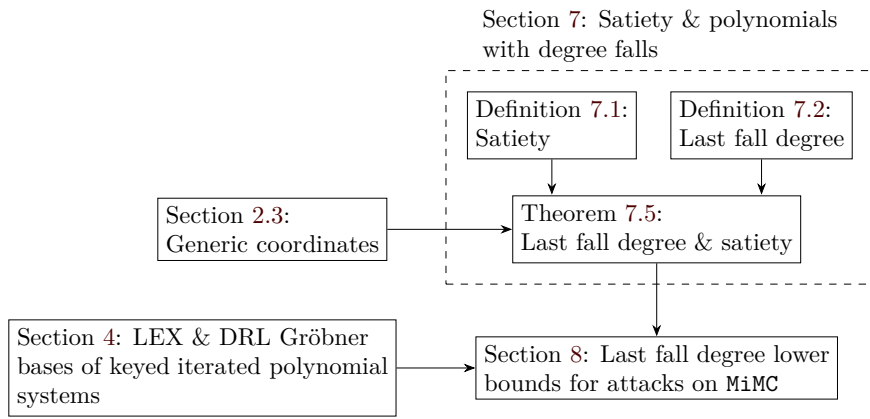

Finally, we finish with a short discussion in \Cref{Sec: discussion}.

    \section{Preliminaries}\label{Sec: Preliminaries}
By $K$ we will always denote a field, by $\bar{K}$ its algebraic closure, and we abbreviate the polynomial ring $P = K [x_1, \dots, x_n]$ if the base field and the number of variables are clear from context.
If $I \subset K [x_1, \dots, x_n]$ is an ideal, then we denote the zero locus of $I$ over $\bar{K}$ as
\begin{equation}
    \mathcal{Z} \left( I \right) = \left\{ \mathbf{p} \in \bar{K}^n \mid f (\mathbf{p}) = 0,\  \forall f \in I \right\} \subset \AffSp{n}{\bar{K}}.
\end{equation}
If moreover $I$ is homogeneous, then we denote the projective zero locus over $\bar{K}$ by $\mathcal{Z}_+ \left( I \right) \subset \ProjSp{n - 1}{\bar{K}}$.

Let $f \in K [x_1, \dots, x_n]$ be a polynomial, and let $x_0$ be an additional variable, we call
\begin{equation}
    f^\homog (x_0, \dots, x_n) = x_0^{\degree{f}} \cdot f \left( \frac{x_1}{x_0}, \dots, \frac{x_n}{x_0} \right) \in K [x_0, \dots, x_n]
\end{equation}
the homogenization of $f$ with respect to $x_0$, and analog for the homogenization of ideals $I^\homog = \left\{ f^\homog \mid f \in I \right\}$ and finite systems of polynomials $\mathcal{F}^\homog = \left\{ f_1^\homog, \dots, f_m^\homog \right\}$.
Let $F \in K [x_0, \dots, x_n]$ be a homogeneous polynomial, we call
\begin{equation}
    F^\dehom (x_1, \dots, x_n) = F (1, x_1, \dots, x_n) \in K [x_1, \dots, x_n]
\end{equation}
the dehomogenization of $F$ with respect to $x_0$, and analog for the dehomogenization of homogeneous ideals $I^\dehom = \left\{ f^\dehom \mid f \in I \right\}$.
Further, we will always assume that we can extend a term order on $K [x_1, \dots, x_n]$ to a term order on $K [x_0, \dots, x_n]$ according to \cite[Definition~8]{Caminata-SolvingPolySystems}.

For a homogeneous ideal $I \subset P$ and an integer $d \geq 0$ we denote
\begin{equation}
    I_d = \left\{ f \in I \mid \degree{f} = d,\ f \text{ homogeneous} \right\},
\end{equation}
and for inhomogeneous ideals $I \subset P$ we denote
\begin{equation}
    I_{\leq d} = \left\{ f \in I \mid \degree{f} \leq d \right\}.
\end{equation}

For a term order $>$ and an ideal $I \subset P$ we denote with
\begin{equation}
    \inid_> (I) = \{ \LT_> (f) \mid f \in I \}
\end{equation}
the initial ideal of $I$, i.e.\ the ideal of leading terms of $I$, with respect to $>$.

Every polynomial $f \in [x_1, \dots, x_n]$ can be written as $f = f_d + f_{d - 1} + \ldots + f_0$, where $f_i$ is homogeneous of degree $i$.
We denote the highest degree component $f_d$ of $f$ with $f^\topcomp$, and analog we denote $\mathcal{F}^\topcomp = \left\{ f_1^\topcomp, \dots, f_m^\topcomp \right\}$.

Let $I, J \subset K [x_1, \dots, x_n]$ be ideals, then we denote with
\begin{equation}
    I : J = \left\{ f \in K [x_1, \dots, x_n] \mid \forall g \in J \colon f \cdot g \in I \right\}
\end{equation}
the usual ideal quotient, and with
\begin{equation}
    I : J^\infty = \bigcup_{i \geq 1} I : J^i
\end{equation}
the saturation of $I$ with respect to $J$.

Let $I, \mathfrak{m} \in K [x_0, \dots, x_n]$ be homogeneous ideals where $\mathfrak{m} = (x_0, \dots, x_n)$, then we call $I^{\sat} = I : \mathfrak{m}^\infty$ the saturation of $I$.

Let $>$ be a term order on $P$, we recall the definition of Buchberger's S-polynomial of $f, g \in P$ with respect to $>$ (cf.\ \cite[Chapter~2~\S6~Definition~4]{Cox-Ideals}).
Denote with $x^\gamma = \lcm \big( \LT_> (f), \LT_> (g) \big)$, then the S-polynomial is defined as
\begin{equation}
    S_> (f, g) = \frac{x^\gamma}{\LT_> (f)} \cdot f - \frac{x^\gamma}{\LT_> (g)} \cdot g.
\end{equation}

We will often encounter the lexicographic and the degree reverse lexicographic term order which we will abbreviate as LEX and DRL respectively.

\subsection{Keyed Iterated Polynomial Systems}\label{Sec: keyed iterated polynomial systems}
A natural description of a univariate keyed function over a finite field is to write the function as composition of low degree polynomials.
This idea leads us to the general notion of keyed iterated polynomial systems.
\begin{defn}[Univariate keyed iterated polynomial system]\label{Def: keyed iterated polynomial system}
    Let $K$ be a field, let $g_1, \dots, g_n \in K [x, y]$ be non-constant polynomials, and let $p, c \in K$ be field elements which will commonly be called plain/ciphertext pair. We say that $f_1, \dots, f_n \in K [x_1, \dots, x_{n - 1}, y]$ is a univariate keyed iterated polynomial system, if the polynomials are of the form
    \begin{align}
        f_1 &= g_1(p, y) - x_1, \nonumber \\
        f_2 &= g_2(x_1, y) - x_2, \nonumber \\
        &\dots \nonumber \\
        f_n &= g_n(x_{n - 1}, y) - c. \nonumber
    \end{align}
    Moreover, we require that
    \[
    \mathcal{Z} (f_1, \dots, f_n) \cap K^n \neq  \emptyset.
    \]
\end{defn}
Before we continue we discuss why the zero locus must contain $K$-valued points.
Let us for the moment replace $p$ with the symbolic variable $x$ and ignore $c$.
Iteratively we can now substitute $f_1, \dots, f_{n - 1}$ into $g_n (x_{n - 1}, y)$, then we obtain a polynomial $f$ in the variables $x$ and $y$.
We can view $f: K \times K \to K$ as a keyed function, where $y$ is the key variable.
The intersection condition states that if $f(p, y) = c$, then there must exist $y \in K$ that satisfies the equation.
I.e., all computations involving a Gr\"obner basis for $f_1, \dots, f_n$ are non-trivial, that is $1 \notin (f_1, \dots, f_n)$.

\subsubsection{\MiMC}\label{Sec: MiMC}
Our main example of a univariate keyed iterated polynomial system is \MiMC, an AO cipher proposed in \cite[\S 2.1]{AC:AGRRT16}.
It is based on the cubing map $x \mapsto x^3$ over finite fields.
If $\Fq$ is a field with $q$ elements, then cubing induces a permutation if $\gcd \left( 3, q - 1 \right) = 1$, see \cite[7.8.~Theorem]{Niederreiter-FiniteFields}.
Let $k \in \Fq$ denote the key, let $r \in \mathbb{N}$ be the number of rounds, and let $c_1, \dots, c_r \in \Fq$ be round constants.
Then the round function of \MiMC is defined as
\begin{equation}
    F_{i, k} (x) =
    \begin{dcases}
        (x + k + c_i)^3,           & 1 \leq  i \leq r - 1, \\
        (x + k + c_r)^3 + k,  &   i = r.
    \end{dcases}
\end{equation}
The \MiMC cipher function is now defined as
\begin{equation}
    F (x, k) = F_{r, k} \circ \cdots \circ F_{1, k} (x),
\end{equation}
which is a permutation for every fixed key $k$.
Given a plain/ciphertext pair $(p, c) \in \Fq^2$ it is straight-forward to describe the univariate keyed iterated polynomial system $I_\MiMC \subset \Fq[x_1, \dots, x_{r - 1}, y]$ for \MiMC
\begin{align}
    \left( p + y + c_1 \right)^3 - x_1 &= 0, \nonumber \\
    \left( x_{i - 1} + y + c_i \right)^3 - x_i &= 0, \quad 1 \leq i \leq r - 1, \\
    \left( x_{r - 1} + y + c_r \right)^3 + y - c &= 0. \nonumber
\end{align}
It was first observed in \cite{AC:ACGKLRS19} that for the DRL term order this system is already a Gr\"obner basis.
It is now straight-forward to compute that
\begin{equation}
    \dim_{\Fq} \big( \Fq[x_1, \dots, x_{r - 1}, y] / I_\MiMC \big) = 3^r.
\end{equation}
For all proposals of \MiMC one has that at least $r \geq 60$.
Hence, using this Gr\"obner basis we do not expect a successful key recovery with today's computational capabilities.

\subsubsection{Feistel-\MiMC}
With the Feistel network we can construct block ciphers with cubing as round function.
Note that a Feistel network induces a permutation irrespective of the size or characteristic of the finite field $\Fq$.
A very special case is the Feistel-$2n/n$ network which encrypts two message blocks of size $n$ with a key of size $n$.
As previously, let $\Fq$ be a finite field, let $r$ be the number of rounds, let $k \in \Fq$ denote the key, and let $c_1, \dots, c_r \in \Fq$ be round constants.
Then the \MiMC-$2n/n$ \cite[\S 2.1]{AC:AGRRT16} round function is defined as
\begin{equation}\label{Equ: Feistel-MiMC}
    F_{i, k}
    \begin{pmatrix}
        x_L \\ x_R
    \end{pmatrix}
    =
    \begin{dcases}
        \begin{pmatrix}
            x_R + \left( x_L + k + c_i \right)^3 \\ x_L
        \end{pmatrix}
        ,   &   1 \leq i \leq r - 1, \\
        \begin{pmatrix}
            x_R + \left( x_L + k + c_r \right)^3  + k \\ x_L
        \end{pmatrix}
        ,   & i = r.
    \end{dcases}
\end{equation}
Again the cipher is defined as iteration of the round functions with respect to the plaintext variables
\begin{equation}
    F_k \left( x_L, x_R \right) = F_{r, k} \circ \cdots \circ F_{0, k} \left( x_L, x_R \right).
\end{equation}

Analog to \MiMC we can model Feistel-\MiMC with a ``multivariate'' system of keyed iterated polynomials.
\begin{defn}[Keyed iterated polynomial system for Feistel-$2n/n$]\label{Def: Feistel keyed iterated polynomial system}
    Let $K$ be a field, let $g_1, \dots, g_n \in K [x, y]$ be non-constant polynomials, and let $(p_L, p_R), (c_L, c_R) \in K^2$ be field elements which will commonly be called plain/ciphertext pair.
    We say that $f_{L, 1}, f_{R, 1}, \dots, f_{L, n}, f_{R, n} \allowbreak \in K [x_{L, 1}, x_{R, 1}, \dots, x_{L, n - 1}, x_{R, n - 1}, y]$ is keyed iterated polynomial system for Feistel-$2n/n$, if the polynomials are of the form
    \[
    \begin{pmatrix}
        f_{L, i} \\ f_{R, i}
    \end{pmatrix}
    =
    \begin{dcases}
        \begin{pmatrix}
            p_R + g_1 (p_L, y) - x_{L, 1} \\
            p_L - x_{R, 1}
        \end{pmatrix}
        , & i = 1, \\
        \begin{pmatrix}
            x_{R, i - 1} + g_i (x_{L, i - 1}, y) - x_{L, i} \\
            x_{L, i - 1} - x_{R, i}
        \end{pmatrix}
        , & 2 \leq i \leq n - 1 \\
        \begin{pmatrix}
            x_{R, n - 1} + g_n ( x_{L, n - 1}, y) - c_L \\
            x_{L, n - 1} - c_R
        \end{pmatrix}
        , & i = n.
    \end{dcases}
    \]
    Moreover, we require that
    \[
    \mathcal{Z} \left( f_{L, 1}, f_{R, 1}, \dots, f_{L, n}, f_{R, n} \right) \cap K^{2n - 1} \neq \emptyset.
    \]
\end{defn}

\subsection{Linear Algebra-Based Gr\"obner Basis Algorithms \& the Solving Degree}\label{Sec: solving degree}
Let $I \subset P = K [x_1, \dots, x_n]$ be an ideal, and let $>$ be a term order on $P$.
A finite basis $\mathcal{G} = \{ g_1, \dots, g_m \}$ of $I$ is said to be a $>$\emph{-Gr\"obner basis} \cite{Buchberger} if $\inid_> (I) = \big( \LT_> (g_1), \dots, \LT_> (g_m) \big)$.
For any term order $>$ on $P$ and any non-trivial ideal $I$ a finite $>$-Gr\"obner basis exists.
For a general introduction to Gr\"obner bases we refer to \cite{Cox-Ideals}.

Today two classes of Gr\"obner basis algorithms are known: \emph{Buchberger's algorithm} and \emph{linear algebra-based algorithms}.
In this paper we are only concerned with the latter.
These algorithms perform Gaussian elimination on the \emph{Macaulay matrices} which under certain conditions produces a Gr\"obner basis.
This idea can be traced back to \cite{Lazard-Groebner}, examples for modern linear algebra-based algorithms are F4 \cite{Faugere-F4} and Matrix-F5 \cite{Faugere-F5}.

The Macaulay matrices are defined as follows, let $\mathcal{F} = \{ f_1, \dots, f_m \} \subset P $ be a system of homogeneous polynomials and fix a term order $>$.
The \textit{homogeneous Macaulay matrix} $M_d$ has columns indexed by monomials in $P_d$ sorted from left to right with respect to $>$, and the rows of $M_d$ are indexed by polynomials $s \cdot f_i$, where $s \in P$ is a monomial such that $\deg \left( s \cdot f_i \right) = d$.
The entry of the row $s\cdot f_i$ at the column $t$ is then simply the coefficient of the polynomial $s \cdot f_i$ at the monomial $t$.
For an inhomogeneous system we replace $M_d$ with $M_{\leq d}$ and similar the degree equality with an inequality.
By performing Gaussian elimination on $M_0, \dots, M_d$ respectively $M_{\leq d}$ for a large enough value of $d$ one produces a $>$-Gr\"obner basis for $\mathcal{F}$.

Obviously, the sizes of the Macaulay matrices $M_{ d}$ and $M_{\leq d}$ depend on $d$, therefore following the idea of \cite{Ding-SolvingDegree} we define the solving degree as follows.
\begin{defn}[{Solving degree, \cite[Definition~6]{Caminata-SolvingPolySystems}}]\label{Def: solving degree}
    Let $\mathcal{F} = \{ f_1, \dots, f_m \} \subset K [x_1, \dots, \allowbreak x_n]$ and let $>$ be a term order.
    The solving degree of $\mathcal{F}$ is the least degree $d$ such that Gaussian elimination on the Macaulay matrix $M_{\leq d}$ produces a Gr\"obner basis of $\mathcal{F}$ with respect to $>$. We denote it by $\solvdeg_> (\mathcal{F})$.

    If $\mathcal{F}$ is homogeneous, we consider the homogeneous Macaulay matrix $M_d$ and let the solving degree of $\mathcal{F}$ be the least degree $d$ such that Gaussian elimination on $M_0, \dots, M_d$ produces a Gr\"obner basis of $\mathcal{F}$ with respect to $>$.
\end{defn}

Algorithms like F4/5 perform Gaussian elimination on the Macaulay matrix for increasing values of $d$, such an algorithm needs a stopping criterion to decide whether a Gr\"obner basis has already been found. Algorithms like the method we described perform Gaussian elimination on a single matrix $M_{\leq d}$ for a large enough value of $d$.
For this class of algorithms one would like to find sharp bounds on $d$ via the solving degree to keep the Macaulay matrix as small as possible.
Nevertheless, for both classes of algorithms one may choose to artificially stop a computation in the degree corresponding to the solving degree.
Due to this reason we consider the solving degree as a complexity measure of Gr\"obner basis computations and do not discuss termination criteria further.

Let $\mathcal{F} = \{ f_1, \dots, f_m \} \subset P$ be a system of polynomials, and let $\mathcal{F}^\homog$ be its homogenization in $P [x_0]$.
One has that $\left( \mathcal{F}^\homog \right) \subseteq \left( \mathcal{F} \right)^\homog$, and it is easy to construct examples for which the inclusion is strict.
Nevertheless, it was demonstrated in \cite[Theorem~7]{Caminata-SolvingPolySystems} that for the DRL term order one still has that
\begin{equation}
    \solvdeg_{DRL} \left( \mathcal{F} \right) \leq \solvdeg_{DRL} \left( \mathcal{F}^\homog \right).
\end{equation}

\subsubsection{Complexity Estimates via the Solving Degree}\label{Sec: complexity estimate}
Storjohann \cite[\S 2.2]{Storjohann-Matrix} has shown that a reduced row echelon form of a matrix $\mathbf{A} \in K^{M \times N}$, where $K$ is a field and $r = \rank \left( \mathbf{A} \right)$, can be computed in $\mathcal{O} \left( M \cdot N \cdot r^{\omega - 2} \right)$ field operations, where $2 \leq \omega < 2.37286$ is a linear algebra constant \cite{SODA:AlmWil21}.

Let $\mathcal{F} = \{ f_1, \dots, f_m \} \subset P = K [x_1, \dots, x_n]$ be a system of homogeneous polynomials.
It is well-known that the number of monomials in $P$ of degree $d$ is given by the binomial coefficient
\begin{equation}
    N (n, d) = \binom{n + d - 1}{d}.
\end{equation}
So the Macaulay matrix $M_d$ has $N (n, d)$ many columns and $N \big( n, d - \deg (f_1) \big) + \ldots + N \big( n, d - \deg (f_m) \big)$ many rows, hence we can upper bound the size of $M_{d}$ by $m \cdot N (n, d) \times N (n, d)$.
Overall we can estimate the complexity of Gaussian elimination on the Macaulay matrices $M_0, \dots, M_d$ by
\begin{equation}\label{Equ: Groebner basis complexity}
    \mathcal{O} \left( m \cdot d \cdot \binom{n + d - 1}{d}^\omega \right).
\end{equation}
Now let $\mathcal{F} \subset P$ be an inhomogeneous polynomial system and let $\mathcal{F}^\homog \subset P [x_0]$ be its homogenization.
If $\mathcal{G}$ is a DRL Gr\"obner basis of $\mathcal{F}^\homog$, then $\mathcal{G}^\dehom$ is a DRL Gr\"obner basis of $\mathcal{F}$, see \cite[Proposition~4.3.18]{Kreuzer-CompAlg2}.
Therefore, we can also consider \Cref{Equ: Groebner basis complexity} as complexity estimate for inhomogeneous Gr\"obner basis computations.

For ease of numerical computation we approximate the binomial coefficient with
\begin{equation}
    \binom{n}{k} \approx \sqrt{\frac{n}{\pi \cdot k \cdot (n - k)}} \cdot 2^{n \cdot H_2 (k / n)},
\end{equation}
where $H_2 (p) = -p \cdot \log_2 \left( p \right) - \left( 1 - p \right) \cdot \log_2 \left( 1 - p \right)$ denotes the binary entropy (cf.\ \cite[Lemma~17.5.1]{Cover-InformationTheory}).
Moreover, since in general $N (n, d) \gg m \cdot d$ we absorb the factor $m \cdot d$ into the implied constant.
Therefore, for solving degree $d$ and number of variables $n$, we estimate the bit complexity $\kappa$ of a Gr\"obner basis attack via
\begin{equation}\label{Equ: bit complexity estimate}
    \kappa \approx \omega \cdot \left( \frac{1}{2} \cdot \log_2 \left( \frac{n + d - 1}{\pi \cdot d \cdot (n - 1)} \right) + (n + d - 1) \cdot H_2 \left( \frac{d}{n + d - 1} \right) \right).
\end{equation}

\subsection{Solving Degree \& Castelnuovo-Mumford Regularity}\label{Sec: solving degree and Castelnuovo-Mumford regularity}
The mathematical foundation to estimate the solving degree via the Macaulay bound draws heavily from commutative and homological algebra.
For readers unfamiliar with the latter subject we point out that \Cref{Def: generic coordinates}, the notion of generic coordinates, is the key mathematical technique in this paper.
Although this notion dates at least back to the influential work of Bayer \& Stillman \cite{BayerStillman}, it was just recently revealed by Caminata \& Gorla \cite{Caminata-SolvingPolySystems} that for the DRL term order the solving degree of a polynomial system in generic coordinates can always be upper bounded by the Macaulay bound.
Although the theory requires heavy mathematical machinery, we will discuss in \Cref{Sec: Caminata Gorla technique} that being in generic coordinates can be verified with rather simple arithmetic operations.
For a concise treatment and as reference point for interested readers we now introduce the mathematical details that serve as foundation of our theory.

The Castelnuovo-Mumford regularity is a well-established invariant from commutative algebra and algebraic geometry.
We recap the definition from \cite[Chapter~4]{Eisenbud-Syzygies}.
Let $P = K [x_0, \dots, x_n]$ be the polynomial ring and let
\begin{equation}
    \mathbf{F}: \cdots \rightarrow F_i \rightarrow F_{i - 1} \rightarrow \cdots
\end{equation}
be a graded complex of free $P$-modules, where $F_i = \sum_j P(-a_{i, j})$.
\begin{defn}
    The Castelnuovo-Mumford regularity of $\mathbf{F}$ is defined as
    \[
    \reg \left( \mathbf{F} \right) = \sup_i a_{i,j} - i.
    \]
\end{defn}
By Hilbert's Syzygy theorem \cite[Theorem~1.1]{Eisenbud-Syzygies} any finitely graded $P$-module has a finite free graded resolution.
I.e., for every homogeneous ideal $I \subset P$ the regularity of $I$ is computable.

Before we can introduce the connection between Castelnuovo-Mumford regularity and solving degree we must introduce the notion of generic coordinates from \cite{BayerStillman}.
Let $I \subset P$ be an ideal, and let $f \in P$.
We use the shorthand notation ``$f \nmid 0 \mod I$'' for expressing that $f$ is not a zero-divisor on $P / I$.
\begin{defn}[{\cite[Definition~5]{Caminata-SolvingPolySystems,Caminata-SolvingPolySystemsPreprint}}]\label{Def: generic coordinates}
    Let $K$ be an infinite field.
    Let $I \subset K [x_0, \dots, \allowbreak x_n]$ be a homogeneous ideal with $| \mathcal{Z}_+ (I) | < \infty$.
    We say that $I$ is in generic coordinates if either $| \mathcal{Z}_+ (I) | = 0$ or $x_0 \nmid 0 \mod I^{{\sat}}$.

    Let $K$ be any field, and let $K \subset L$ be an infinite field extension.
    $I$ is in generic coordinates over $K$ if $I \otimes_K L [x_0, \dots, x_n] \subset L [x_0, \dots, x_n]$ is in generic coordinates.
\end{defn}

In general, computing the saturation of an ideal is a difficult problem on its own, but if a homogeneous ideal is in generic coordinates, then the saturation is exactly the homogenization of its dehomogenization.
\begin{lem}\label{Lem: saturation equal homogenization}
    Let $K$ be an infinite field, and let $P = K [x_1, \dots, x_n]$.
    Let $I \subset P[x_0]$ be a homogeneous ideal with $\left| \mathcal{Z}_+ (I) \right| \neq 0$.
    Then $I$ is in generic coordinates if and only if
    \[
    I^{\sat} = \left( I^\dehom \right)^\homog.
    \]
\end{lem}
\begin{proof}
    ``$\Rightarrow$'':
    Let $F \in I^{\sat} = I : \mathfrak{m}^\infty$, then there exists an $N \geq 0$ such that $x_0^N \cdot F \in I$.
    On the other hand by \cite[Proposition~4.3.5]{Kreuzer-CompAlg2} we have that $\left( I^\dehom \right)^\homog = I : x_0^\infty$, so also $F \in \left( I^\dehom \right)^\homog$.

    By our assumption $\left| \mathcal{Z}_+ (I) \right| \neq 0$ and contraposition of the projective weak Nullstellensatz \cite[Chapter~8~\S3~Theorem~8]{Cox-Ideals}, we have that $I^\dehom \neq (1)$.
    Now let $F \in \left( I^\dehom \right)^\homog$, since $F^\dehom \notin K$ then also by \cite[Proposition~4.3.5]{Kreuzer-CompAlg2} there must exist an $N \geq 0 $ such that $x_0^N \cdot F \in I$.
    By definition $I \subset I^{\sat}$ so also $x_0^N \cdot F \in I^{\sat}$.
    By assumption $x_0 \nmid 0 \mod I^{\sat}$, hence we must already have that $x_0^{N - 1} \cdot F \in J^{\sat}$.
    Iterating this argument we conclude that $F \in J^{\sat}$.

    ``$\Leftarrow$'':
    We have the ideal equality $I^{\sat} = \left( I^\dehom \right)^\homog = I : x_0^\infty$, so
    \[
    I^{\sat} : x_0 = \left( I : x_0^\infty \right) : x_0 = I : x_0^\infty = I^{\sat}.
    \]
    So if $x_0 \cdot f \in I^{\sat}$, then already $f \in I^{\sat}$ which implies $x_0 \nmid 0 \mod I^{\sat}$.
\end{proof}

We provide a simple counterexample to the ideal equality when the ideal is not in generic coordinates.
\begin{ex}
    Let $K$ be a field and let $I = \left( x^2, y \cdot z \right) \subset K [x, y, z]$ be an ideal where we consider $z$ as the homogenization variable.
    Then $I^{\sat} = I$ but $I : z^\infty = (x^2, y)$.
\end{ex}

Let us now present the connection between the solving degree and the Castelnuovo-Mumford regularity.
\begin{thm}[{\cite[Theorem~9, 10]{Caminata-SolvingPolySystems}}]\label{Th: solvdeg and CM-regularity}
    Let $K$ be an algebraically closed field, and let $\mathcal{F} = \{ f_1, \dots, f_m \} \subset K [x_1, \dots, x_n]$ be an inhomogeneous polynomial system such that $\left( \mathcal{F}^\homog \right)$ is in generic coordinates.
    Then
    \[
        \solvdeg_{DRL} \left( \mathcal{F} \right) \leq \reg \left( \mathcal{F}^\homog \right).
    \]
\end{thm}

By a classical result one can always bound the regularity of an ideal via the Macaulay bound (see \cite[Theorem~1.12.4]{Chardin-Regularity}).
\begin{cor}[{Macaulay bound, \cite[Theorem~2]{Lazard-Groebner}, \cite[Corollary~2]{Caminata-SolvingPolySystems}}]\label{Cor: Macauly bound}
    Consider a system of equations $\mathcal{F} = \{ f_1, \dots, f_m \} \subset K [x_1, \dots, x_n]$ with $d_i = \deg \left( f_i \right)$ and $d_1 \geq \ldots \geq d_m$. Set $l = \min \{ n + 1, m \}$.
    Assume that $\left| \mathcal{Z}_+ \left( \mathcal{F}^\homog \right) \right| < \infty$ and that $\left( F^\homog \right)$ is in generic coordinates over $\bar{K}$.
    Then
    \[
    \solvdeg_{DRL} \left( \mathcal{F} \right) \leq \reg \left( \mathcal{F}^\homog \right) \leq d_1 + \ldots + d_l - l + 1.
    \]
    In particular, if $m > n$ and $d = d_1$, then
    \[
    \solvdeg_{DRL} \left( \mathcal{F} \right) \leq (n + 1) \cdot (d - 1) + 1.
    \]
\end{cor}

A sufficient condition for a polynomial system to be in generic coordinates is that the system contains the field equations or their fake Weil descent \cite[Theorem~11]{Caminata-SolvingPolySystems}.

Via inclusion of the field equations we obtain the following solving degree bound for \MiMC.
\begin{ex}[\MiMC and all field equations I]\label{Ex: MiMC all field equations I}
    Let \MiMC be defined over $\Fq$, and let $r$ be the number of rounds.
    Denote the ideal of all field equations by $F$, and the \MiMC ideal with $I_\MiMC$.
    Then by \cite[Theorem~11]{Caminata-SolvingPolySystems} the solving degree is bounded by
    \[
    \solvdeg_{DRL} \left( I_\MiMC + F \right) \leq r \cdot (q - 1) + 3.
    \]
\end{ex}

However, this bound is very unsatisfying, because it only takes the field equations into account except for one summand.
On the other hand, it suffices to add only the field equation for the key variable to $I_\MiMC$ to restrict all solutions to $\Fq^r$. However, this modification is not covered by \cite[Theorem~11]{Caminata-SolvingPolySystems}.

\subsubsection{The Caminata-Gorla Technique}\label{Sec: Caminata Gorla technique}
Since we are going to emulate the proof of \cite[Theorem~11]{Caminata-SolvingPolySystems} several times in this paper, we recapitulate its main argument.
By \cite[Theorem~2.4]{BayerStillman} a homogeneous ideal $I \subset P = \bar{K} [x_0, \dots, x_n]$ with $\dim \left( P / I \right) = 1$ and $\left| \mathcal{Z}_+ (I) \right| < \infty$ is in generic coordinates if and only if $\inid_{DRL} (I)$ is in generic coordinates.
Assume that
\begin{equation}\label{Equ: projective variety intersection condition}
    \mathcal{Z}_+ \big( \inid_{DRL} (I) \big) \cap \mathcal{Z}_+ \left( x_0 \right) = \mathcal{Z}_+ \Big( \big( \inid_{DRL} (I), x_0 \big) \Big) = \emptyset,
\end{equation}
then by the projective weak Nullstellensatz \cite[Chapter~8~\S3~Theorem~8]{Cox-Ideals} there exists some $r \geq 1$ such that $\mathfrak{m}^r = (x_0, \dots, x_n)^r \subset \big( \inid_{DRL} (I), x_0 \big)$.
This also implies that for every $1 \leq i \leq n$ there exists some $r_i \geq 1$ such that $x_i^{r_i} \in \inid_{DRL} (I)$.\footnote{
    Let $B$ be a basis of $\inid_{DRL} (I)$ and $B'$ be basis of $\big( \inid_{DRL} (I), x_0 \big)$.
    If $\mathfrak{m}^r \subset \big( \inid_{DRL} (I), x_0 \big)$ for some $r \geq 1$, then for all $0 \leq i \leq n$ there exists a smallest integer $r_i \in \mathbb{Z}$ such that $x_i^{r_i} \in B'$.
    Observe that a monomial $m \in B$ is also an element of $B'$ if $x_0 \nmid m$.
    Conversely, any basis element from $B'$ different to $x_0$ must come from $B$.}
Now suppose that $x_0 \cdot f \in \inid_{DRL} (I)^{\sat}$, then for every $g \in \mathfrak{m}$ there exists $N \geq 1$ such that $g^N \cdot (x_0 \cdot f) \in \inid_{DRL} (I)$.
Let $g$ be a monomial, we do a case distinction.
\begin{itemize}
    \item For $\gcd \left( g, x_0 \right) = 1$, we increase the power of $g$ until $x_i^{r_i} \mid g^M \cdot f$, for some $1 \leq i \leq n$ and $M \geq 1$, hence $g^M \cdot f \in \inid_{DRL} (I)$.

    \item For $\gcd \left( g, x_0 \right) \neq 1$,  we use the factorization $g^{N + 1} = \frac{g}{x_0} \cdot g^N \cdot x_0$, hence $g^{N + 1} \cdot f \in \inid_{DRL} (I)$.
\end{itemize}
Now let $g \in \mathfrak{m}$ be a polynomial, then we can find $N \geq 0$ big enough so that for every monomial present in $g$ one of the two previous cases applies.
So if $x_0 \cdot f \in \inid_{DRL} (I)^{\sat}$ we also have that $f \in \inid_{DRL} (I)^{\sat}$.
Hence, $x_0 \nmid 0 \mod \inid_{DRL} (I)^{\sat}$ and by \cite[Theorem~2.4]{BayerStillman} also $x_0 \nmid 0 \mod I^{\sat}$.

Finally, in practice \Cref{Equ: projective variety intersection condition} can efficiently be checked with the following ideal equality \cite[Lemma~2.2]{BayerStillman}
\begin{equation}
    \inid_{DRL} \left( I, x_0 \right) = \big( \inid_{DRL} (I), x_0 \big).
\end{equation}

    \section{Characterization of Polynomial Systems in Generic Coordinates}\label{Sec: characterization generic coordinates}
Let $\mathcal{F}$ be a polynomial system which contains equations $x_i^{d_i} - p_i (x_1, \dots, x_n)$, where $\degree{p_i} < d_i$, for all $i$, then the Caminata-Gorla technique implies that $\left( \mathcal{F}^\homog \right)$ is in generic coordinates, see \cite[Remark~13]{Caminata-SolvingPolySystems}.
Though, the polynomial systems of our interest are not of this form in general, e.g.\ the keyed iterated polynomial system for \MiMC.
However, it is already implicit in the Caminata-Gorla technique that for a homogenized polynomial system to be in generic coordinates the associated ideal of the highest degree components has to be zero-dimensional.
If this is the case, then we can indeed find equations $x_i^{d_i} - p_i (x_1, \dots, x_n)$ in $\left( \mathcal{F} \right)$ that lift to $x_i^{d_i} - x_0^{d_i - \degree{p_i}} \cdot p_i (x_1, \dots, x_n)$ in $\left( \mathcal{F}^\homog \right)$ which implies genericity.

To formally prove this observation we need a lemma.
\begin{lem}\label{Lem: radical monomial ideal between maximal ideals}
    Let $K$ be a field, and let $I \subset K [x_0, \dots, x_n]$ be a radical monomial ideal such that $(x_1, \dots, x_n) \subset I \subset (x_0, \dots, x_n)$.
    Then either $I = (x_1, \dots, x_n)$ or $I = (x_0, \dots, x_n)$.
\end{lem}
\begin{proof}
    Let $P = K [x_1, \dots, x_n]$, by the isomorphism theorems for rings we have that
    \[
    P [x_0] / I \cong \big( P [x_0] / (x_1, \dots, x_n) \big) / \big( I / (x_1, \dots, x_n) \big) \cong K [x_0] / \big( I / (x_1, \dots, x_n) \big).
    \]
    Moreover, if $I / (x_1, \dots, x_n) \neq (0)$, then $I / (x_1, \dots, x_n) = (f)$, where $f \in K [x_0]$.
    $I$ is radical, so $f$ has to be reduced.
    Since $I$ is also a monomial ideal this implies that $f = x_0$.
\end{proof}

Now we can prove the following characterization of generic coordinates.
\begin{thm}\label{Th: generic coordinates and highest degree components}
    Let $K$ be an algebraically closed field, and let $\mathcal{F} = \{ f_1, \dots, f_m \} \subset K [x_1, \dots, x_n]$ be a polynomial system such that
    \begin{enumerate}[label=(\roman*)]
        \item $(\mathcal{F}) \neq (1)$, and

        \item $\dim \left( \mathcal{F} \right) = 0$.
    \end{enumerate}
    Then the following are equivalent.
    \begin{enumerate}
        \item $\left( \mathcal{F}^\homog \right)$ is in generic coordinates.

        \item\label{Item: radical} $\sqrt{\mathcal{F}^\topcomp} = \left( x_1, \dots, x_n \right)$.

        \item $\left( \mathcal{F}^\topcomp \right)$ is zero-dimensional in $K [x_1, \dots, x_n]$.

        \item For every $1 \leq i \leq n$ there exists $d_i \in \mathbb{Z}_{\geq 1}$ such that $x_i^{d_i} \in \inid_{DRL} \left( \mathcal{F}^\homog \right)$.
    \end{enumerate}
\end{thm}
\begin{proof}
    ``(1) $\Rightarrow$ (4)'':
    Let $\left( \mathcal{F}^\homog \right)$ be in generic coordinates and suppose that $\mathcal{Z}_+ \left( \mathcal{F}^\homog \right) = \emptyset$.
    Then by the projective weak Nullstellensatz \cite[Chapter~8~\S 3~Theorem~8]{Cox-Ideals} $x_0^k \in \left( \mathcal{F}^\homog \right)$, where $k \geq 1$.
    In particular, this implies that $1 \in \left( \mathcal{F}^\homog \right)^\dehom = (\mathcal{F})$, a contradiction to $(\mathcal{F}) \neq (1)$.
    So $\abs{\mathcal{Z}_+ \left( \mathcal{F}^\homog \right)} \neq 0$, then by \Cref{Lem: saturation equal homogenization} we have that
    \[
    \left( \mathcal{F}^\homog \right)^{\sat} = \left( \left( \mathcal{F}^\homog \right)^\dehom \right)^\homog = \left( \mathcal{F} \right)^\homog.
    \]
    By assumption $\left( \mathcal{F} \right)$ is zero-dimensional, so for every $1 \leq i \leq n$ there exists $f \in \left( \mathcal{F} \right)$ such that $\LM_{DRL} (f) = x_i^d$, where $d > 0 $, see \cite[Theorem~5.11]{Kemper-CommAlg} and \cite[Chapter~5~\S 3~Theorem~6]{Cox-Ideals}.
    Therefore, $f^\homog \in \left( \mathcal{F}^\homog \right)^{\sat}$.
    By definition of the saturation, for every $s \in \mathfrak{m}$ there exists an integer $N \geq 0$ such that $s^N \cdot f^\homog \in \left( \mathcal{F}^\homog \right)$, thus for $s = x_i$ also $x_i^N \cdot f^\homog \in \left( \mathcal{F}^\homog \right)$.
    Obviously, we have that $\LM_{DRL} \left( x_i^N \cdot f^\homog \right) = x_i^{N + d}$.

    ``(4) $\Rightarrow$ (3)'':
    By assumption, for every $1 \leq i \leq n$ there exists $f \in \left( \mathcal{F}^\homog \right)$ such that $\LM_{DRL} (f) = x_i^{d_i}$, where $d_i > 0$.
    Without loss of generality we can assume that $f$ is homogeneous, so we can represent it as
    \[
    f = \sum_{j = 1}^{m} g_i \cdot f_i^\homog,
    \]
    where $g_i \in K [x_0, \dots, x_n]$ is homogeneous for all $i$.
    Now we split the $g_i$'s and $f_i$'s as
    \begin{align}
        f_i^\homog &= f_i^\topcomp + x_0 \cdot \tilde{f}_i, \nonumber \\
        g_i &= g_i^\topcomp + x_0 \cdot \tilde{g}_i, \nonumber
    \end{align}
    where $\tilde{f}_i, \tilde{g}_i \in K [x_0, \dots, x_n]$ are homogeneous and if $\tilde{f}_i, g_i, \tilde{g}_i \neq 0$, then
    \begin{align}
        \degree{f_i^\homog} &= \degree{f_i^\topcomp} = \degree{x_0 \cdot \tilde{f}_i}, \nonumber \\
        \degree{g_i} &= \degree{g_i^\topcomp} = \degree{x_0 \cdot \tilde{g}_i}, \nonumber \\
        \degree{g_i} &= \degree{f} - \degree{f_i}. \nonumber
    \end{align}
    We can now further decompose
    \begin{equation}\label{Equ: highest degree decomposition}
        f
        = \sum_{j = 1}^{m} \left( g_i^\topcomp + x_0 \cdot \tilde{g}_i \right) \cdot \left( f_i^\topcomp + x_0 \cdot \tilde{f}_i \right)
        = \sum_{j = 1}^{m} g_i^\topcomp \cdot f_i^\topcomp + x_0 \cdot \tilde{f},
    \end{equation}
    where $\degree{\tilde{f}} = \degree{f} - 1$.
    Since $\sum_{j = 1}^{m} g_i^\topcomp \cdot f_i^\topcomp >_{DRL} x_0 \cdot \tilde{f}$ we must have that $\LM_{DRL} \left( f \right) = \LM_{DRL} \left( \sum_{j = 1}^{m} g_i^\topcomp \cdot f_i^\topcomp \right)$.
    We can also decompose the left-hand side of the last equation $f = f^\topcomp + x_0 \cdot \hat{f}$, and by rearranging we yield that
    \[
    \underbrace{\left( f^\topcomp - \sum_{j = 1}^{m} g_i^\topcomp \cdot f_i^\topcomp \right)}_{\in K [x_1, \dots, x_n]} = x_0 \cdot \left( \tilde{f} - \hat{f} \right).
    \]
    The only element in $K [x_1, \dots, x_n]$ divisible by $x_0$ is $0$, so we have constructed an element in $\left( \mathcal{F}^\topcomp \right)$ with leading monomial $x_i^{d_i}$.
    Again by \cite[Theorem~5.11]{Kemper-CommAlg} and \cite[Chapter~5~\S 3~Theorem~6]{Cox-Ideals} this implies zero-dimensionality of $\left( \mathcal{F}^\topcomp \right)$.

    ``(3) $\Rightarrow$ (4)'': Suppose $\left( \mathcal{F}^\topcomp \right)$ is zero-dimensional.
    For the claim we can work through the arguments of the previous claim in a backwards manner.
    Since $\left( \mathcal{F}^\topcomp \right)$ is homogeneous and zero-dimensional we can find $f^\topcomp = \sum_{j = 1}^{m} g_i^\topcomp \cdot f_i^\topcomp$, where $g_i^\topcomp$ homogeneous, such that $\LM_{DRL} \left( f^\topcomp \right) = x_i^{d_i}$, where $d_i > 0$.
    With \Cref{Equ: highest degree decomposition} we can lift this decomposition (with $\tilde{g}_i = 0$) to $\left( \mathcal{F}^\homog \right)$.
    Now let $s, t \in K [x_0, \dots, x_n]$ be monomials such that $\degree{s} = \degree{t}$ and $x_0 \nmid s$ and $x_0 \mid t$.
    For compatibility with homogenization we have set $x_0$ as least variable with respect to DRL, this immediately implies that $s >_{DRL} t$ and the claim follows.

    ``(2) $\Leftrightarrow$ (3)'':
    This is just a reformulation of the projective weak Nullstellensatz \cite[Chapter~8~\S 3~Theorem~8]{Cox-Ideals}, \cite[Chapter~5~\S 3~Theorem~6]{Cox-Ideals} and \cite[Theorem~5.11]{Kemper-CommAlg}.

    ``(2) $\Rightarrow$ (1)'':
    Assume that $\sqrt{\mathcal{F}^\topcomp} = (x_1, \dots, x_n)$ in $K [x_1, \dots, x_n]$.
    To apply \cite[Theorem~2.4]{BayerStillman} in the Caminata-Gorla technique (\Cref{Sec: Caminata Gorla technique}) we have to show that $\dim \left( \mathcal{F}^\homog \right) = 1$.
    By the equivalence of (2) and (4) we know that
    \[
    (x_1, \dots, x_n) \subset \sqrt{\inid_{DRL} \left( \mathcal{F}^\homog \right)}.
    \]
    So by \Cref{Lem: radical monomial ideal between maximal ideals} either
    $\sqrt{\inid_{DRL} \left( \mathcal{F}^\homog \right)} = (x_1, \dots, x_n)$ or $\sqrt{\inid_{DRL} \left( \mathcal{F}^\homog \right)} = (x_0, \dots, x_n)$.
    Assume the latter, then there exists a homogeneous $f \in \left( \mathcal{F}^\homog \right)$ such that $\LM_{DRL} (f) = x_0^d$, where $d > 0$.
    Since $x_0$ is the least variable with respect to DRL this already implies that $f = x_0^d$.
    Thus, $1 \in \left( \mathcal{F}^\homog \right)^\dehom = \left( \mathcal{F} \right)$, a contradiction to the non-triviality of $\mathcal{F}$.
    So $\sqrt{\inid_{DRL} \left( \mathcal{F}^\homog \right)} = (x_1, \dots, x_n)$.
    Note that this also implies that $\mathcal{Z}_+ \left( \mathcal{F}^\homog \right) \neq \emptyset$ by a contraposition of the equivalence in the projective weak Nullstellensatz \cite[Chapter~8~\S 3~Theorem~3]{Cox-Ideals}.
    It is well-known that $\left( \mathcal{F}^\homog \right)$ and $\inid_{DRL} \left( \mathcal{F}^\homog \right)$ have the same affine Hilbert function, see \cite[Chapter~9~\S 3~Proposition 4]{Cox-Ideals}.
    Moreover, for any ideal $I \subset K [x_0, \dots, x_n]$ the affine Hilbert polynomials of $I$ and $\sqrt{I}$ have the same degree, see \cite[Chapter~9~\S 3~Proposition~6]{Cox-Ideals}.
    Since dimension of an affine ideal is equal to the degree of the affine Hilbert polynomial, see \cite[Theorem~11.13]{Kemper-CommAlg}, the two previous observations imply that
    \[
    \dim \left( \mathcal{F}^\homog \right)
    = \dim \Big( \inid_{DRL} \big( \mathcal{F}^\homog \big) \Big)
    = \dim \Big( \sqrt{\inid_{DRL} \big( \mathcal{F}^\homog \big)} \Big)
    = \dim \left( x_1, \dots, x_n \right) = 1
    \]
    in $K [x_0, \dots, x_n]$.
    Also, the dimension of an affine variety $\mathcal{Z} (I)$, where $I \subset K [x_0, \dots, x_n]$, is defined as the degree of the affine Hilbert polynomial of $I$, see \cite[Chapter~9~\S 3]{Cox-Ideals}.
    If $I$ is in addition homogeneous and $\mathcal{Z}_+ (I) \neq \emptyset$, then by \cite[Chapter~9~\S 3~Theorem~12]{Cox-Ideals} we have for the dimension of the projective variety $\mathcal{Z}_+ (I)$ that
    \[
    \dim \big( \mathcal{Z}_+ (I) \big) = \dim \big( \mathcal{Z} (I) \big) - 1 = \dim (I) - 1.
    \]
    Combining, all our previous observations we yield that $\dim \Big( \mathcal{Z}_+ \big( \mathcal{F}^\homog \big) \Big) = 0$, and it is well-known that zero-dimensional projective varieties have only finitely many points, i.e.\ $\left| \mathcal{Z}_+ \left( \mathcal{F}^\homog \right) \right| < \infty$, see \cite[Chapter~9~\S 4~Proposition~6]{Cox-Ideals}.
    To apply the Caminata-Gorla technique (\Cref{Sec: Caminata Gorla technique}) it is left to show that $\mathcal{Z}_+ \Big( \inid_{DRL} \big( \mathcal{F}^\homog \big), x_0 \Big) = \emptyset$.
    Note that
    \[
    \left( \mathcal{F}^\homog, x_0 \right) = \left( \mathcal{F}^\homog, x_0 \right) = \left( \mathcal{F}^\topcomp, x_0 \right),
    \]
    so by \cite[Lemma~2.2]{BayerStillman}
    \[
    \Big( \inid_{DRL} \big( \mathcal{F}^\homog \big), x_0 \Big)
    = \inid_{DRL} \left( \mathcal{F}^\homog, x_0 \right)
    = \inid_{DRL} \left( \mathcal{F}^\topcomp, x_0 \right)
    = \Big( \inid_{DRL} \big( \mathcal{F}^\topcomp \big), x_0 \Big).
    \]
    Finally, by our initial assumption and the projective weak Nullstellensatz \cite[Chapter~8~\S 3~Theorem~3]{Cox-Ideals} we have
    \[
    \mathcal{Z}_+ \Big( \inid_{DRL} \big( \mathcal{F}^\homog \big), x_0 \Big)
    = \mathcal{Z}_+ \Big( \inid_{DRL} \big( \mathcal{F}^\topcomp \big), x_0 \Big)
    = \emptyset.
    \]
    So we can apply the Caminata-Gorla technique (\Cref{Sec: Caminata Gorla technique}) to deduce that $x_0 \nmid 0 \mod \left( \mathcal{F}^\homog \right)^{\sat}$.
\end{proof}

As consequence, we can conclude that every zero-dimensional affine polynomial system has a set of generators that is in generic coordinates.
\begin{cor}\label{Cor: DRL Groebner basis in generic coordinates}
    Let $K$ be an algebraically closed field, and let $\mathcal{F} = \{ f_1, \dots, f_m \} \subset K [x_1, \dots, x_n]$ be a polynomial system such that
    \begin{enumerate}[label=(\roman*)]
        \item $(\mathcal{F}) \neq (1)$, and

        \item $\dim \left( \mathcal{F} \right) = 0$.
    \end{enumerate}
    For every DRL Gr\"obner basis $\mathcal{G} \subset (\mathcal{F})$ the ideal $(\mathcal{G}^\homog)$ is in generic coordinates.
\end{cor}

Another quantity that is often studied in the Gr\"obner basis complexity literature is the so-called degree of regularity of a polynomial system.
\begin{defn}[{Degree of regularity, \cite[Definition~4]{Bardet-Complexity}}]
    Let $K$ be a field, and let $\mathcal{F} \subset P = K [x_1, \dots, x_n]$.
    Assume that $\left( \mathcal{F}^\topcomp \right)_d = P_d$ for some integer $d \geq 0$.
    The degree of regularity is defined as
    \begin{equation*}
        d_{\reg} \left( \mathcal{F} \right) = \min \left\{ d \geq 0 \; \middle\vert \; \left( \mathcal{F}^\topcomp \right)_d = P_d \right\}.
    \end{equation*}
\end{defn}

It follows from the projective weak Nullstellensatz \cite[Chapter~8~\S 3~Theorem~8]{Cox-Ideals} and \cite[Theorem~5.11]{Kemper-CommAlg} that $d_{\reg} \left( \mathcal{F} \right) < \infty$ is equivalent to $\dim \left( \mathcal{F}^\topcomp \right) = 0$.
\begin{cor}\label{Cor: degree of regularity finite}
    Let $K$ be an algebraically closed field, and let $\mathcal{F} = \{ f_1, \dots, f_m \} \subset K [x_1, \dots, x_n]$ be a polynomial system such that
    \begin{enumerate}[label=(\roman*)]
        \item $(\mathcal{F}) \neq (1)$, and

        \item $\dim \left( \mathcal{F} \right) = 0$.
    \end{enumerate}
    Then $\left( \mathcal{F}^\homog \right)$ is in generic coordinates if and only if $d_{\reg} \left( \mathcal{F} \right) < \infty$.
\end{cor}

\Cref{Th: generic coordinates and highest degree components} also significantly simplifies application of the Caminata-Gorla technique.
For an inhomogeneous polynomial system $\mathcal{F} \subset K [x_1, \dots, x_n]$ we can verify \Cref{Th: generic coordinates and highest degree components} \ref{Item: radical} as follows.
\begin{enumerate}
    \item Homogenize $\mathcal{F}$.

    \item Extract the highest degree components via $\mathcal{F}^\topcomp = \mathcal{F}^\homog \mod (x_0)$.

    \item For $x_1, \dots, x_n$, construct a polynomial $f \in \sqrt{\mathcal{F}^\topcomp}$ such that $f = x_i^d$, where $d > 0$.
    Then replace $\mathcal{F}^\topcomp$ by $\mathcal{F}^\topcomp \mod (x_i)$.
\end{enumerate}

\begin{rem}\label{Rem: genericity notions}
    The notion of generic coordinates is not the only genericity notion for polynomial respectively monomial ideals.
    Other worthwhile mentioning notions are being in quasi-stable position \cite[Definition~3.1]{Hashemi-Deterministic} and Noether position \cite[Definition~4.1]{Hashemi-Deterministic}.
    Let $\mathcal{F} \subset K [x_1, \dots, x_n]$ be such that $\mathcal{F}^\homog$ is in generic coordinates.
    Then, by the proof of \Cref{Th: generic coordinates and highest degree components} we have that $\dim \left( \mathcal{F}^\homog \right) = 1$.
    It follows from \Cref{Lem: saturation equal homogenization} that being in generic coordinates coincides with being in quasi-stable position \cite[Proposition~3.2]{Hashemi-Deterministic} and Noether position \cite[Lemma~4.1]{Bermejo-Computing} in dimension $1$.
    For a survey of different genericity notions and their relations we refer to \cite{Hashemi-Deterministic}.
\end{rem}

Utilizing \Cref{Th: generic coordinates and highest degree components} we can finally provide an elementary proof that a keyed iterated polynomial system is in generic coordinates.
\begin{thm}\label{Th: iterated system generic coordinates}
    Let $K$ be an algebraically closed field, and let $P = K [x_1, \dots, \allowbreak x_{n - 1}, y]$.
    Let $\mathcal{F} = \{ f_1, \dots, \allowbreak f_n \} \subset P$ be a univariate keyed iterated system of polynomials such that
    \begin{enumerate}[label=(\roman*)]
        \item $d_i = \deg \left( f_i \right) \geq 2$ for all $1 \leq i \leq n$, and

        \item $f_i$ has the monomial $x_{i - 1}^{d_i}$ for all $2 \leq i \leq n$.
    \end{enumerate}
    Then every non-trivial homogeneous ideal $I \subset P[x_0]$ with $\mathcal{Z}_+ \left( I \right) \neq \emptyset$ and $\mathcal{F}^\homog \subset I$ is in generic coordinates.
\end{thm}
\begin{proof}
    Let us substitute $x_0 = 0$ into the equations $f_i^\homog = 0$.
    For $f_1$ we then have $y^{d_1} = 0$ and hence also $y = 0$.
    Substituting $x_0 = y = 0$ into $f_2$ then yields $x_1 = 0$, hence by iteration we obtain that $x_0 = y = x_1 = \ldots = x_{n - 1} = 0$.
    Therefore, $\sqrt{I^\topcomp} = (y, x_1, \dots, x_n)$ and the claim follows from \Cref{Th: generic coordinates and highest degree components}.
\end{proof}

    \section{DRL \& LEX Gr\"obner Bases of Keyed Iterated Polynomial Systems}\label{Sec: lex Groebner Basis}
In this section we investigate the DRL \& LEX Gr\"obner basis of univariate keyed iterated polynomial systems and Feistel-$2n/n$ polynomial systems.
Consequently, we will see that the solving degree of \MiMC and all field equations can be upper bounded by \MiMC and the field equation for the key variable, and that under a mild assumption also Feistel-$2n/n$ polynomial systems are in generic coordinates.
Moreover, understanding the degrees of polynomials in the lexicographic Gr\"obner basis will be a key ingredient in the proofs of the Castelnuovo-Mumford regularity lower bounds.

The following lemma certainly has been proven by many students of computer algebra.
\begin{lem}[{\cite[Chapter~4~\S5~Exercise~13]{Cox-Ideals}}]\label{Lem: lex shape lemma}
    Let $K$ be a field, let $f_1, \dots, f_n \in K [x_1]$ be polynomials in one variable such that $\deg \left( f_1 \right) > 0$, and let
    \[
    I = \big( f_1 (x_1), x_2 - f_2 (x_1), \dots, x_n - f_n (x_1) \big) \subset K [x_1, \dots, x_n]
    \]
    be an ideal.
    \begin{enumerate}
        \item\label{Item: unique lex representation} Every $f \in K [x_1, \dots, x_n]$ can be written uniquely as $f = q + r$ where $q \in I$ and $r \in K [x_1]$ with either $r = 0$ or $\deg \left( r \right) < \deg \left( f_1 \right)$.

        \item\label{Item: univariate LEX ideal membership} Let $f \in K [x_1]$, then $f \in I$ if and only if $f$ is divisible by $f_1 \in K [x_1]$.

        \item $I$ is a prime ideal if and only if $f_1 \in K [x_1]$ is irreducible.

        \item $I$ is a radical ideal if and only if $f_1 \in K [x_1]$ is square-free.

        \item Let $f_{1, \text{red}} \in K [x_1]$ be the generator of the radical ideal $(f_{1, \text{red}}) = \sqrt{(f_1)}$, then $\sqrt{I} = (f_{1, \text{red}}) + I$.
    \end{enumerate}
\end{lem}

If we use the LEX term order $x_2 > \ldots > x_n > x_1$, then it's easy to see that the generators of $I$ are already a LEX Gr\"obner basis.
Now we establish that the LEX Gr\"obner basis of a univariate keyed iterated polynomial system has exactly the shape of \Cref{Lem: lex shape lemma}.
\begin{lem}[Keyed Iterated Shape Lemma I]\label{Lem: keyed iterated shape lemma I}
    Let $K$ be a field, let $f_1, \dots, f_n \in K [x_1, \dots, \allowbreak x_{n - 1}, \allowbreak y]$ be a univariate keyed iterated polynomial system together with the LEX term order $x_1 > \ldots > x_{n - 1} > y$.
    Let $\hat{f}_1, \dots, \hat{f}_n \in K [x_1, \dots, x_{n - 1}, y]$ be constructed via the following iteration:
    \begin{enumerate}[label=(\roman*)]
        \item For $i = 1$, set $\hat{f}_1 = -f_1$.

        \item For $2 \leq i \leq n$, let $\hat{f}_{i} = \left( -f_i \mod \hat{f}_{i - 1} \right)$ where the modulo operation is computed with respect to the LEX term order.
    \end{enumerate}
    Then
    \begin{enumerate}
        \item\label{Item: lex polynomial shape} For $1 \leq i < n$, we have that $\hat{f}_i = x_i - \hat{g}_i (y)$ for some $\hat{g}_i \in K [y]$ and $\hat{f}_n \in K [y]$.

        \item\label{Item: lex Groebner basis} $I = \left( f_1, \dots, f_n \right) = \left( \hat{f}_1, \dots, \hat{f}_n \right)$, in particular $\hat{f}_1, \dots, \hat{f}_n$ is a LEX Gr\"obner basis of $I$.

        \item\label{Item: upper bound for all field equations} If $\left| K \right| = q$, then $I + \left( y^q - y \right) = \left( \hat{f}_1, \dots, \hat{f}_{n - 1}, \gcd \left( \hat{f}_n, y^q - y \right) \right)$, and this ideal is radical.
        In particular,
        \[
        I + \left( y^q - y \right) = I + \left( x_1^q - x_1, \dots, y^q - y \right),
        \]
        and
        \[
        \solvdeg_{DRL} \left( f_1, \dots, f_n, x_1^q - x_1, \dots, y^q - y \right) \leq \solvdeg_{DRL} \left( f_1, \dots, f_n, y^q - y \right).
        \]
    \end{enumerate}
\end{lem}
\begin{proof}
    For (1), follows from the construction of the $\hat{f}_i$'s.

    For (2), if we record the ``quotients'' which we drop in the modulo operation in the construction of the $\hat{f}_i$'s, then we can reconstruct the $f_i$'s with the $\hat{f}_i$'s.
    So the $\hat{f}_i$'s are indeed an ideal basis.
    Moreover, they have coprime leading monomials under LEX, so by \cite[Chapter~2~\S9~Theorem~3, Proposition~4]{Cox-Ideals} they are a LEX Gr\"obner basis of $I$.

    For (3), let $d = \gcd \left( \hat{f}_n, y^q - y \right)$.
    Clearly, $\left( \hat{f}_1, \dots, \hat{f}_{n - 1}, d \right)$ is an ideal basis of $I + \left( y^q - y \right)$, and again the leading monomials are pairwise coprime under LEX, so they are a Gr\"obner basis of $I + \left( y^q - y \right)$.
    Since $y^q - y$ is square-free also $d$ must be square-free, so by \Cref{Lem: lex shape lemma} $I + \left( y^q - y \right)$ is a radical ideal.
    It is obvious from the shape of the $\hat{f}_i$'s that already $\mathcal{Z} \big( I + (y^q - y) \big) \subset \Fqn$.
    Now we can conclude from Hilbert's Nullstellensatz and \cite[Theorem~3.1.2]{Gao-Counting} that $I + \left( y^q - y \right) = I + \left( x_1^q - x_1, \dots, y^q - y \right)$.
    For the inequality observe that the Macaulay matrix of the polynomial system with one field equation is a submatrix of the Macaulay matrix of the polynomial system with all field equations.
    So the claim follows.
\end{proof}

With an additional assumption on the leading monomials of a univariate keyed iterated polynomial system we can compute the degrees in the LEX Gr\"obner basis.
\begin{cor}\label{Cor: degrees in lexicographic Groebner basis}
    Let $K$ field, and let $f_1, \dots, f_n \in K [x_1, \dots, x_{n - 1}, y]$ be a univariate keyed iterated polynomial system such that
    \begin{enumerate}[label=(\roman*)]
        \item $d_i = \deg \left( f_i \right) \geq 2$ for all $1 \leq i \leq n$, and

        \item $f_i$ has the monomial $x_{i - 1}^{d_i}$ for all $2 \leq i \leq n$.
    \end{enumerate}
    Let $\hat{f}_1, \dots, \hat{f}_n$ be the LEX Gr\"obner basis of $f_1, \dots, f_n$.
    Then
    \[
    \deg \left( \hat{f}_i \right) = \prod_{k = 1}^{i} d_k.
    \]
\end{cor}
\begin{proof}
    The assertion follows straight-forward from the monomial assumption and the LEX Gr\"obner basis construction procedure.
\end{proof}

Conversely, we can transform any lexicographic Gr\"obner basis with the shape of \Cref{Lem: lex shape lemma} into a univariate keyed iterated polynomial system.
\begin{lem}[Keyed Iterated Shape Lemma II]\label{Lem: keyed iterated shape lemma II}
    Let $K$ be a field, and assume that the ideal $I \subset K [x_1, \dots, y]$ has a LEX Gr\"obner basis of the form
    \[
    x_1 - g_1 (y), \dots, x_{n - 1} - g_{n - 1} (y), g_n (y)
    \]
    such that $1 \leq \deg \left( g_1 \right) \leq \ldots \leq \deg \left( g_n \right)$.
    Then $I$ has an ideal basis of the form
    \[
    \hat{g}_1 (y) - x_1, \hat{g}_2 (x_1, y) - x_2, \dots, \hat{g}_{n - 1} (x_{n - 2}, y) - x_{n - 1}, \hat{g}_n (x_{n - 1}, y).
    \]
    I.e., the ideal is generated by a univariate keyed iterated polynomial system.
\end{lem}
\begin{proof}
    For the proof we work with the DRL term order $x_1 > \ldots > x_{n - 1} > y$.
    Let $f_1, \dots, f_n$ denote the polynomials in the LEX Gr\"obner basis, and let $\hat{f}_1, \dots, \hat{f}_n$ denote the polynomials that we claim are the univariate keyed iterated basis.
    We set $\hat{f}_1 = - \big( x_1 - f_1 (y) \big)$. For $2 \leq i \leq n$ we now compute $\hat{f}_i = - f_i \mod \hat{f}_{i - 1}$ with respect to DRL.
    Since we assumed that $1 \leq \deg \left( f_1 \right) \leq \ldots \leq \deg \left( f_n \right)$ the modulo operation indeed constructs non-trivial polynomials $\hat{g}_i (x_{i - 1}, y)$.
\end{proof}

Note that the keyed iterated system from \Cref{Lem: keyed iterated shape lemma II} is in general not a DRL Gr\"obner basis.
We present a simple counterexample.
\begin{ex}
    Let $K$ be a field, and let
    \[
    I = (x_1 - y^3, x_2 - y^5, y^7) \in K [x_1, x_2, y].
    \]
    The respective keyed iterated polynomial system of $I$ is then given by
    \[
    y^3 - x_1,\ x_1 \cdot y^2 - x_2,\ x_2 \cdot y^2,
    \]
    but the DRL Gr\"obner basis of $I$ is given by
    \[
    x_1 \cdot y^2 - x_2,\ x_2 \cdot y^2,\ y^3 - x_1,\ x_1^2 - x_2 \cdot y,\ x_1 \cdot x_2,\ x_2^2.
    \]
\end{ex}

With \Cref{Lem: lex shape lemma} \ref{Item: unique lex representation} and \Cref{Lem: keyed iterated shape lemma I} we can transform every polynomial $f \in K [x_1, \dots, \allowbreak x_n, y]$ into a univariate polynomial $\hat{f} \in K [y]$ using only ideal operations, i.e.\ by performing division by remainder with respect to the LEX Gr\"obner basis.
Understanding the degree of these univariate polynomials will be our main ingredient in proving lower bounds on the regularity.
\begin{prop}\label{Prop: lex substitution properties}
    Let $K$ be a field, and let $I = \left( f_1, \dots, f_n \right) \subset P = K [x_1, \dots, x_{n - 1}, \allowbreak y]$ be an ideal generated by a univariate keyed iterated polynomial system such that
    \begin{enumerate}[label=(\roman*)]
        \item $d_i = \deg \left( f_i \right) \geq 2$ for all $1 \leq i \leq n$, and

        \item $f_i$ has the monomial $x_{i - 1}^{d_i}$ for all $2 \leq i \leq n$.
    \end{enumerate}
    Let $f \in P$ be a polynomial, then we denote with $\hat{f} \in K [y]$ the unique univariate polynomial obtained via division by remainder of $f$ by $I$ with respect to LEX.
    Then
    \begin{enumerate}
        \item Let $a \in P \setminus \inid_{DRL} (I)$ be a monomial, in the computation of $\hat{a}$ via division by remainder there is never a reduction modulo the univariate LEX polynomial.

        \item Let $a, b \in P \setminus \inid_{DRL} (I)$ be monomials such that $a \vert b$, then $\hat{a} \vert \hat{b}$ and $\degree{\hat{a}} \leq \degree{\hat{b}}$.

        \item Let $a, b, c \in P \setminus \inid_{DRL} (I)$ be monomials such that $a \cdot c, b \cdot c \in P \setminus \inid_{DRL} (I)$.
        If $\deg \left( \hat{a} \right) \leq  \deg \left( \hat{b} \right)$, then $\deg \left( \hat{a} \cdot \hat{c} \right) \leq  \deg \left( \hat{b} \cdot \hat{c} \right)$.
    \end{enumerate}
    Let $s_i = \prod_{j = i}^{n - 1} x_i^{d_{i + 1} - 1}$.
    Then
    \begin{enumerate}[resume]
        \item\label{Item: degree} The degree of $\hat{s}_i$ is given by
        \[
        \deg \left( \hat{s}_i \right) = \prod_{k = 1}^{n} d_k - \prod_{k = 1}^{i} d_k.
        \]

        \item\label{Item: degree after substitution} Let $t \in P \setminus \inid_{DRL} (I)$ be a monomial such that $\deg \left( t \right) \leq \deg \left( s_i \right)$.
        Then we also have that $\deg \left( \hat{t} \right) \leq \deg \left( \hat{s}_i \right)$, and the inequality is strict if $t \neq s$.
    \end{enumerate}
\end{prop}
\begin{proof}
    Let $I = \big(x_1 - \tilde{f}_1 (y), \dots, x_{n - 1} - \tilde{f}_{n - 1} (y), \tilde{f}_n (y) \big)$ be the LEX Gr\"obner basis of $I$, see \Cref{Lem: keyed iterated shape lemma I} \ref{Item: lex polynomial shape}.

    For (1), let $m = y^{d_1 - 1} \cdot \prod_{i = 1}^{n - 1} x_i^{d_{i + 1} - 1}$, then any monomial $a \in P \setminus \inid_{DRL} (I)$ divides $m$.
    So if there is a reduction modulo $\tilde{f}_n$ in the construction of $\hat{a}$, then there also must be a reduction in the construction of $\hat{m}$.
    Via \Cref{Cor: degrees in lexicographic Groebner basis} let us compute
    \begin{align}
        \deg \Big( m \big(\tilde{f}_1, \dots, \tilde{f}_{n - 1}, y \big) \Big) \nonumber
        &= d_1 - 1 + \sum_{i = 1}^{n - 1} (d_{i + 1} - 1) \cdot \prod_{k = 1}^{i} d_k \nonumber \\
        &= d_1 - 1 + \sum_{i = 1}^{n - 1} \left( \prod_{k = 1}^{i + 1} d_k - \prod_{k = 1}^{i} d_k \right) \nonumber \\
        &= d_1 - 1 + \prod_{k = 1}^{n} d_k - d_1 = \prod_{k = 1}^{n} d_k - 1. \nonumber
    \end{align}
    Since $\degree{\tilde{f}_n} = \prod_{k = 1}^{n} d_k$, there cannot be a reduction modulo $\tilde{f}_n$ in the construction of $\hat{m}$ anymore.
    So we have already computed $\degree{\hat{m}}$.
    By contraposition the claim follows.

    For (2) and (3), by (1) there is no reduction modulo $\tilde{f}_n$ in the construction of $\hat{a}$, $\hat{b}$ and $\hat{c}$, so the claims follow from standard polynomial arithmetic.

    For (4), the computation is analog to the degree computation in (1)
    \[
    \degree{\hat{s}_i}
    = \sum_{j = i}^{n - 1} (d_{j + 1} - 1) \cdot \prod_{k = 1}^{j} d_k
    = \sum_{j = i}^{n - 1} \left( \prod_{k = 1}^{j + 1} d_k - \prod_{k = 1}^{j} d_k \right)
    = \prod_{k = 1}^{n} d_k - \prod_{k = 1}^{i} d_k.
    \]

    For (5), we do a downwards induction.
    Assume that there is a monomial $t \in P \setminus \inid_{DRL} (I)$ such that $t \neq s_i$, $\degree{t} \leq \degree{s_i}$ and $\degree{\hat{t}} > \degree{\hat{s}_i}$.
    The monomial $t$ must differ from $s_i$ in at least one variable.
    Assume that the difference is in the variable $x_{n - 1}$, then $t$ must divide the monomial
    \[
    u_{n - 1} = y^{d_1 - 1} \cdot x_{n - 1}^{d_n - 2} \cdot \prod_{i = 1}^{n - 2} x_i^{d_{i + 1} - 1}.
    \]
    Let us compute the degree of the LEX remainder degree analog to (1) and (4)
    \begin{align}
        \degree{\hat{u}_{n - 1}}
        &= d_1 - 1 + \left( d_n - 2 \right) \cdot \prod_{k = 1}^{n - 1} d_k + \sum_{j = 1}^{n - 2} \left(d_{j + 1} - 1 \right) \cdot \prod_{k = 1}^{j} d_k \nonumber \\
        &= d_1 - 1 - \prod_{k = 1}^{n - 1} d_k + \sum_{j = 1}^{n - 1} (d_{j + 1} - 1) \cdot \prod_{k = 1}^{j} d_k \nonumber \\
        &= \prod_{k = 1}^{n} d_k - \prod_{k = 1}^{n - 1} d_k - 1 < \degree{\hat{s}_i}. \nonumber
    \end{align}
    On the other hand, by (2) we have that $\degree{\hat{t}} \leq \degree{\hat{u}_{n - 1}}$.
    Therefore, $t$ has to coincide with $s_i$ on $x_{n - 1}$, else we already have $\degree{\hat{t}} < \degree{\hat{s}_i}$.
    Now we replace $s_i$ and $t$ by $s_i / x_{n - 1}^{d_n - 1}$ and $t / x_{n - 1}^{d_n - 1}$ respectively, then we perform the same argument for $x_{n - 2}$.
    Inductively we now conclude that either $t = s_i$ or $\degree{\hat{t}} < \degree{\hat{s}_i}$.
\end{proof}

\subsection{DRL \& LEX Gr\"obner Bases for Feistel-2n/n}
Having studied the LEX Gr\"obner basis of univariate keyed iterated polynomial systems we now describe LEX and DRL Gr\"obner bases of Feistel-$2n/n$ polynomial systems, see \Cref{Def: Feistel keyed iterated polynomial system}.
\begin{prop}\label{Prop: Feistel Groebner bases}
    Let $K$ be a field, and let $\mathcal{F} = \left\{ f_{L, 1}, f_{R, 1}, \dots, f_{L, n}, f_{R, n} \right\} \subset K [x_{L, 1}, x_{R, 1}, \allowbreak \dots, \allowbreak x_{L, n - 1}, \allowbreak x_{R, n - 1}, y]$ be a keyed iterated polynomial system for Feistel-$2n/n$ such that
    \begin{enumerate}[label=(\roman*)]
        \item $d_i = \deg \left( f_{L, i} \right) \geq 2$ for all $1 \leq i \leq n$, and

        \item\label{Item: Feistel monomial assumption} $f_i$ has the monomial $x_{L, i - 1}^{d_i}$ for all $2 \leq i \leq n$.
    \end{enumerate}
    Then
    \begin{enumerate}
        \item\label{Item: Feistel Groebner basis} For the DRL term order $x_{L, 1} > x_{R, 1} > \ldots > x_{L, n - 1} > x_{R, n - 1} > y$, a DRL Gr\"obner basis $\mathcal{G}$ of $\mathcal{F} \setminus \{ f_{R, n} \}$ is given by
        \[
        \begin{pmatrix}
            \tilde{f}_{L, i} \\ \tilde{f}_{R, i}
        \end{pmatrix}
        =
        \begin{dcases}
            \begin{pmatrix}
                p_R + g_1 (p_L, y) - x_{R, 2} \\
                p_L - x_{R, 1} \\
            \end{pmatrix}
            , & i = 1, \\
            \begin{pmatrix}
                p_L + g_2 (x_{R, 2}, y) - x_{R, 3} \\
                x_{L, 1} - x_{R, 2}
            \end{pmatrix}
            , & i = 2, \\
            \begin{pmatrix}
                x_{R, i - 1} + g_i (x_{R, i}, y) - x_{R, i + 1} \\
                x_{L, i - 1} - x_{R, i}
            \end{pmatrix}
            , & 3 \leq i \leq n - 2, \\
            \begin{pmatrix}
                x_{R, n - 2} + g_{n - 1} (x_{R, n - 1}, y) - x_{L, n - 1} \\
                x_{L, n - 2} - x_{R, n - 1}
            \end{pmatrix}
            , & i = n - 1, \\
            \begin{pmatrix}
                x_{R, n - 1} + g_n (x_{L, n - 1}, y) - c_R \\
                0
            \end{pmatrix}
            , & i = n.
        \end{dcases}
        \]

        \item\label{Item: generic generators} If we remove the linear polynomials from the  DRL Gr\"obner basis $\mathcal{G}$, then this downsized polynomial system $\mathcal{H} \subset P = K [x_{R, 2}, \dots, x_{R, n - 1}, \allowbreak x_{L, n - 1}, \allowbreak y]$ is already a zero-dimensional Gr\"obner basis.
        Moreover, $\left( \mathcal{H}^\homog, f_{R, n}^\homog \right)$ is in generic coordinates over $\bar{K}$.

        \item\label{Item: Feistel LEX Groebner Basis} For the LEX term order $x_{R, 2} > \ldots > x_{R, n - 1} > x_{L, n - 1} > y$ the Gr\"obner basis of $(\mathcal{H})$ is of the form
        \[
        x_{L, 1} - \hat{f}_1, x_{R, 2} - \hat{f}_2, \dots, \allowbreak x_{R, n - 1} - \hat{f}_{n - 1}, \hat{f}_n
        \]
        where the $\hat{f}_i \in K [y]$ are constructed analog to the LEX Gr\"obner basis in \Cref{Lem: keyed iterated shape lemma I}.

        \item\label{Item: Feistel degrees in LEX Groebner basis} The degree of $\hat{f}_i$ is given by
        \[
        \degree{\hat{f}_i} = \prod_{k = 1}^{i} d_k.
        \]
    \end{enumerate}
    Let $f \in P$ be a polynomial, then we denote with $\hat{f} \in K [y]$ the unique univariate polynomial obtained via division by remainder of $f$ by $(\mathcal{H})$ with respect to LEX, and for $2 \leq i \leq n - 1$ let $s_i = x_{L, n - 1}^{d_n - 1} \cdot \prod_{j = i}^{n - 2} x_{R, i}^{d_i - 1}$ where $s_{n - 1} = x_{L, n - 1}^{d_n - 1}$.
    Then
    \begin{enumerate}[resume]
        \item\label{Item: Feistel degree} The degree of $\hat{s}_i$ is given by
        \[
        \deg \left( \hat{s}_i \right) = \prod_{k = 1}^{n} d_k - \prod_{k = 1}^{i} d_k.
        \]

        \item\label{Item: Feistel degree inequality} Let $t \in P \setminus \inid_{DRL} (I)$ be a monomial such that $\deg \left( t \right) \leq \deg \left( s_i \right)$.
        Then we also have that $\deg \left( \hat{t} \right) \leq \deg \left( \hat{s}_i \right)$, and the inequality is strict if $t \neq s$.
    \end{enumerate}
\end{prop}
\begin{proof}
    For (1), the polynomials $\tilde{f}_{L, i}$ are constructed by substituting the linear polynomials of $\mathcal{F} \setminus \{ f_{R, n} \}$ into the non-linear ones.
    After the substitution all leading monomials are coprime so by \cite[Chapter 2 §9 Theorem 3, Proposition 4]{Cox-Ideals} we have constructed a Gr\"obner basis.

    For (2), it is easy to see that $\mathcal{H} = \{ \tilde{f}_{L, 1}, \dots, \tilde{f}_{L, n} \} \subset P = K [x_{R, 2}, \dots, x_{R, n - 1}, \allowbreak x_{L, n - 1}, \allowbreak y]$ and that
    $\inid_{DRL} (\mathcal{H}) = \left( y^{d_1}, x_{R, 2}^{d_2}, \dots, x_{R, n - 1}^{d_{n - 1}}, x_{L, n - 1}^{d_n} \right)$.
    Observe that only a finite number of monomials of $P$ is not contained in $\inid_{DRL} (\mathcal{H})$.
    I.e., $\dim_K \big( P / \inid_{DRL} (\mathcal{H}) \big) < \infty$ as $K$-vector space and by a well-known equivalence from commutative algebra (see \cite[Theorem~5.11]{Kemper-CommAlg}) this implies zero-dimensionality.
    Lastly, being in generic coordinates is proven analog to \Cref{Th: iterated system generic coordinates}.

    For (3), given the DRL Gr\"obner basis from (1) and (2) the LEX Gr\"obner basis can be constructed via iterated substitutions.

    For (4), this follows analog to \Cref{Cor: degrees in lexicographic Groebner basis}.

    For (5) and (6), the proofs are identical to \Cref{Prop: lex substitution properties} \ref{Item: degree} and \ref{Item: degree after substitution}.
\end{proof}
We provide a counterexample that in general the generators of the DRL Gr\"obner basis of Feistel-$2n/n$ cannot be transformed into a univariate keyed iterated polynomial system.
\begin{ex}
    Consider \MiMC-$2n/n$ over $\F_{13}$ with the round constants and plain/ciphertext pair
    \[
    c_1 = 0, \quad c_2 = 0, \quad c_3 = 0, \quad c_4 = 0, \quad
    \begin{pmatrix}
        p_L \\ p_R
    \end{pmatrix}
    =
    \begin{pmatrix}
        0 \\ 0
    \end{pmatrix}
    , \quad
    \begin{pmatrix}
        c_L \\ c_R
    \end{pmatrix}
    =
    \begin{pmatrix}
        0 \\ 0
    \end{pmatrix}
    .
    \]
    The downsized DRL Gr\"obner basis is
    \begin{align}
        &y^{3} - x_{R, 2}, \nonumber \\
        &x_{R, 2}^3 - 2 \cdot x_{R, 2}^2 \cdot y - 2 \cdot x_{R, 2} \cdot y^2 + y^3 - x_{R, 3}, \nonumber \\
        &x_{R, 3}^3 - 2 \cdot x_{R, 3}^2 \cdot y - 2 \cdot x_{R, 3} \cdot y^2 + 2 \cdot y^3 - x_{L, 3}, \nonumber \\
        &x_{L, 3}^3 - 2 \cdot x_{L, 3}^2 \cdot y - 2 \cdot x_{L, 3} \cdot y^2 + y^3 + y + x_{R, 3}. \nonumber
    \end{align}
    But the univariate keyed iterated generators of this system are
    \begin{align}
        &y^{3} - x_{R, 2}, \nonumber \\
        &x_{R, 2}^3 - 2 \cdot x_{R, 2}^2 \cdot y - 2 \cdot x_{R, 2} \cdot y^2 + x_{R, 2} - x_{R, 3}, \nonumber \\
        &x_{R, 3}^3 - 2 \cdot x_{R, 3}^2 \cdot y - 2 \cdot x_{R, 3} \cdot y^2 + 2 \cdot y^3 - x_{L, 3}, \nonumber \\
        &y^9 - 2 \cdot y^7 - 2 \cdot y^5 + x_{L, 3}^3 - 2 \cdot x_{L, 3}^2 \cdot y - 2 \cdot x_{L, 3} \cdot y^2 + 2 \cdot y^3 + y. \nonumber
    \end{align}
\end{ex}

\subsection{DRL \& LEX Gr\"obner Bases for Univariate Keyed Iterated Polynomial Systems With Two Plain/ciphertexts}
If one has multiple plain/ciphertext samples for a cipher, then one can combine the respective iterated polynomial systems into a joint system and compute its Gr\"obner basis.
Analog to \Cref{Lem: keyed iterated shape lemma I} and \Cref{Prop: lex substitution properties} we now describe DRL Gr\"obner bases for a two plain/ciphertext attack on a univariate cipher.
With the same assumptions as in \Cref{Th: iterated system generic coordinates} we can also prove that the polynomial system of a two plain/ciphertext attack is in generic coordinates.
\begin{prop}\label{Prop: variety intersection substituion properties}
    Let $K$ be a field, and let
    \[
    \begin{split}
        f_1, \dots, f_n &\in K [u_1, \dots, u_{n - 1}, y], \text{ and} \\
        h_1, \dots, h_n &\in K [v_1, \dots, v_{n - 1}, y]
    \end{split}
    \]
    be two univariate keyed iterated polynomial systems which are constructed with the same $g_1, \dots, \allowbreak g_n \in K [x, y]$ but with different plain/ciphertext pairs $(p_1, c_1), (p_2, c_2) \allowbreak \in K^2$.
    Let $\mathcal{F} = \{ f_1, \dots, f_n, h_1, \dots, h_n \}$, and assume that
    \begin{enumerate}[label=(\roman*)]
        \item $d_i = \deg \left( g_i \right) \geq 2$ for all $1 \leq i \leq n$, and

        \item $g_i$ has the monomial $x^{d_i}$ for all $2 \leq i \leq n$.
    \end{enumerate}
    Then
    \begin{enumerate}
        \item\label{Item: Groebner bases} The sets
        \[
        \{ f_1, \dots, f_n \}, \quad \{ h_1, \dots, h_n \}, \quad \mathcal{F} \setminus \{ h_1 \}, \quad \mathcal{F} \setminus \{ f_1 \},
        \]
        are DRL Gr\"obner bases.

        \item\label{Item: variety intersection generic coordinates} If in addition $K$ is algebraically closed, then $\left( \mathcal{F}^\homog \right)$ is in generic coordinates.
    \end{enumerate}
\end{prop}
\begin{proof}
    (1) follows from \cite[Chapter~2~\S9~Theorem~3,~Proposition~4]{Cox-Ideals}, and the proof of (2) is analog to \Cref{Th: iterated system generic coordinates}.
\end{proof}

Note that it is also straight-forward to generalize \ref{Item: variety intersection generic coordinates} for any number of plain/ciphertext pairs.

    \section{Solving Degree Upper Bounds For Attacks On \MiMC}\label{Sec: applications}
Combining our results from \Cref{Sec: characterization generic coordinates,Sec: lex Groebner Basis} we can now derive upper bounds for the solving degree of various attacks on \MiMC, Feistel-\MiMC and Feistel-\MiMC-Hash.
To illustrate our bounds in practice we also compute the bit complexity of \Cref{Equ: Groebner basis complexity} for sample values.

\subsection{Adding a Minimal Number of Field Equations}\label{Sec: field equation attack}
In the original bound for \MiMC, see \Cref{Ex: MiMC all field equations I}, we had to include all field equations into the system, but as we saw in \Cref{Lem: keyed iterated shape lemma I} it suffices to include a single field equation to limit all solutions to the base field.
\begin{ex}[\MiMC and one field equation I]\label{Ex: MiMC solving degree I}
    Let \MiMC be defined over $\Fq$, and let $r$ be the number of rounds.
    We denote with $I_\MiMC$ the \MiMC ideal.
    It follows from \Cref{Lem: keyed iterated shape lemma I} \ref{Item: upper bound for all field equations} that one only needs to include the field equation for the key variable $y$ to limit all solutions to $\Fq$. Hence, by applying \Cref{Cor: Macauly bound} and \Cref{Th: iterated system generic coordinates} to this system we yield
    \[
        \solvdeg_{DRL} \big( I_\MiMC + (y^q - y) \big) \leq q + 2 \cdot r.
    \]
\end{ex}

As an immediate consequence we can also improve the bound of the attack with all field equations.
\begin{ex}[\MiMC and all field equations II]\label{Ex: MiMC all field equations II}
    Let \MiMC be defined over $\Fq$, and let $r$ be the number of rounds.
    Denote the ideal of all field equations by $F$ and the \MiMC ideal with $I_\MiMC$.
    Then by \Cref{Lem: keyed iterated shape lemma I} \ref{Item: upper bound for all field equations} and \Cref{Ex: MiMC solving degree I} the solving degree is bounded by
    \[
        \solvdeg_{DRL} (I_\MiMC + F) \leq q + 2 \cdot  r.
    \]
    Moreover, small scale experiments indicate that the solving degree of this attack is always less than or equal to $q + r - 1$.
\end{ex}

Since the \MiMC polynomials are already a DRL Gr\"obner basis, we can also replace the field equation $y^q - y$ by its remainder $r_y$ modulo $I_\MiMC$ with respect to DRL.
Then the solving degree bound becomes
\begin{equation}
    \solvdeg_{DRL} \big( I_\MiMC + (y^q - y) \big) \leq \degree{r_y} + 2 \cdot r.
\end{equation}
Let $r \geq \ceil{\log_3 \left( q \right)}$, then experimentally we observed that
\begin{equation}\label{Equ: observed degree bound}
    \degree{r_y} \leq 2 \cdot \ceil{\log_3 \left( q \right)}.
\end{equation}

In \Cref{Tab: MiMC sample values} we provide complexity estimates in bits for a Gr\"obner basis computation of \MiMC and the field equation for the key for an optimal adversary with $\omega = 2$.
For ease of computation we estimated the logarithm of the binomial coefficient with \Cref{Equ: bit complexity estimate}.
For the chosen field sizes and the least round number such that $r \geq \ceil{\log_3 \left( q \right)}$ \MiMC achieves a security level of at least $128$ bits against the field equation attack.
\begin{table}[H]
    \centering
    \caption{Complexity estimation of Gr\"obner basis computation for \MiMC and the field equation for the key variable with $\omega = 2$.
             $\kappa_\text{rem}$ denotes the complexity estimate for the remainder of the field equation with the empirically observed degree bound (\Cref{Equ: observed degree bound}).
             The number of rounds $r$ is the least integer such that $r \geq \ceil{\log_3 \left( q \right)}$.}
    \label{Tab: MiMC sample values}
    \begin{tabular}{ c | c | c | c }
        \toprule

        $\log_2 \left( q \right)$  & $64$    & $128$   & $256$ \\
        \midrule

        $r$                        & $50$    & $81$    & $162$    \\
        $\kappa_\text{rem}$ (bits) & $337.5$ & $572.4$ & $1156.2$ \\

        \bottomrule
    \end{tabular}
\end{table}

\subsection{The Two Plain/Ciphertext Attack}
Intuitively, with a single plain/ciphertext pair one can construct a fully determined polynomial system a cipher.
By adding more plain/ciphertext pairs one constructs an overdetermined system, and it is expected that the additional information reduces the difficulty of solving the system. Let $I, J \subset P$ be ideals representing a cipher for different plain/ciphertext pairs.
Combining these two systems into a single system geometrically corresponds to the intersection of two varieties, i.e.,
\begin{equation}
    \mathcal{Z} (I + J) = \mathcal{Z} (I) \cap \mathcal{Z} (J).
\end{equation}
Let us now apply these considerations to \MiMC.
\begin{ex}[\MiMC and two plain/ciphertext pairs I]\label{Ex: MiMC two plaintext attack I}
    Let \MiMC be defined over $\Fq$, let $r$ be the number of rounds, and let $(p_1, c_1), (p_2, c_2) \in \Fq^2$ be two distinct plain/ciphertext pairs generated with the same key by a \MiMC encryption function.
    For these pairs we can construct the univariate polynomials $f_1, f_2 \in \Fq [y]$ in the respective LEX Gr\"obner basis of degree $3^r$.
    These two polynomials must have at least one common root, namely the key $k \in \Fq$.
    If one divides $f_1$ and $f_2$ by $y - k$ and considers them as random polynomials, then with high probability they are coprime.
    Now let $I_{\MiMC, 1} \subset \Fq [u_1, \dots, u_{r - 1}, y]$ and $I_{\MiMC, 2} \subset \Fq [v_1, \dots, v_{r - 1}, y]$ denote the ideals corresponding to the plain/ciphertext pairs. Then, with high probability $\mathcal{Z} (I_{\MiMC, 1} + I_{\MiMC, 2})$ contains only a single point.
    By \Cref{Cor: Macauly bound} and \Cref{Prop: variety intersection substituion properties} \ref{Item: variety intersection generic coordinates} we now obtain the following bound for the solving degree of $I_{\MiMC, 1} + I_{\MiMC, 2}$
    \[
    \solvdeg_{DRL} \left( I_{\MiMC, 1} + I_{\MiMC, 2} \right) \leq 4 \cdot r + 1.
    \]
\end{ex}

In \Cref{Tab: MiMC two plaintext sample values} we provide complexity estimates in bits for a Gr\"obner basis computation of \MiMC and two plain/ciphertexts for an optimal adversary with $\omega = 2$.
For ease of computation we estimated the logarithm of the binomial coefficient with \Cref{Equ: bit complexity estimate}.
Recall from \Cref{Tab: MiMC sample values} that for $q \geq 2^{64}$ we have that $r \geq 50$, hence $50$ rounds are sufficient to achieve $128$ bits security against the two plain/ciphertext attack.
\begin{table}[H]
    \centering
    \caption{Complexity estimation of Gr\"obner basis computation for \MiMC and two plain/ciphertext pairs with $\omega = 2$ over a finite field $\Fq$ with $\gcd \left( 3, q - 1 \right) = 1$.}
    \label{Tab: MiMC two plaintext sample values}
    \begin{tabular}{ c | c }
        \toprule

        $r$ & $\kappa$ (bits) \\
        \midrule

        $10$ & $99.4$  \\
        $50$ & $538.1$ \\

        \bottomrule
    \end{tabular}
\end{table}

\subsection{Feistel-\MiMC}
Interestingly, \MiMC-$2n/n$ behaves very similar to the two plaintext attack on \MiMC, in the sense that with high probability the standard polynomial model of \MiMC-$2n/n$ does not have any solutions in the algebraic closure and its Gr\"obner basis is expected to be linear.
\begin{ex}[\MiMC-$2n/n$ I]\label{Ex: Feistel MiMC I}
    Let $\Fq$ be a finite field, let $r$ be the number of rounds, and let $k \in \Fq$ denote the key.
    Suppose we are given a plain/ciphertext pair $\left( p_L, p_R \right), \left( c_L, c_R \right) \in \Fq^2$ for \MiMC-$2n/n$ generated by the key $k$.
    By substituting this pair into the cipher function we obtain two univariate polynomials $F_y (p_L, p_R) - (c_L, c_R) = (0, 0)$ in the key variable $y$.
    These polynomials have at least one common root, namely $y - k$.
    If we divide these polynomials by $y - k$ and consider them as random polynomials, then with high probability they are coprime.
    Now, to launch an efficient Gr\"obner basis attack we first compute the downsized DRL Gr\"obner basis of the Feistel-$2n/n$ polynomial system from \Cref{Prop: Feistel Groebner bases} \ref{Item: Feistel Groebner basis}.
    Then we add the missing polynomial and compute the Gr\"obner basis.
    By \Cref{Prop: Feistel Groebner bases} \ref{Item: generic generators} the polynomial system is in generic coordinates, therefore we can also apply \Cref{Cor: Macauly bound} to obtain the following bound for the solving degree
    \[
    \solvdeg_{DRL} \left( I_{\MiMC\text{-}2n/n} \right) \leq 2 \cdot r + 1.
    \]
\end{ex}

In the following table we provide complexity estimates in bits for a Gr\"obner basis computation of Feistel-\MiMC for an optimal adversary with $\omega = 2$.
For ease of computation we estimated the logarithm of the binomial coefficient with \Cref{Equ: bit complexity estimate}.
As in \Cref{Tab: MiMC sample values,Tab: MiMC two plaintext sample values}, $50$ rounds are sufficient for Feistel-\MiMC to achieve $128$ bits security.
\begin{table}[H]
    \centering
    \caption{Complexity estimation of Gr\"obner basis computation for Feistel-\MiMC with $\omega = 2$ over a finite field $\Fq$.}
    \label{Tab: Feistel MiMC sample values}
    \begin{tabular}{ c | c }
        \toprule

        $r$ & $\kappa$ (bits) \\
        \midrule

        $10$ & $48.6$  \\
        $50$ & $266.7$ \\

        \bottomrule
    \end{tabular}
\end{table}

\subsection{Feistel-\MiMC-Hash}\label{Sec: Feistel-MiMC-Hash upper bounds}
For Feistel-\MiMC-Hash the Feistel-\MiMC permutation is instantiated in the sponge framework \cite{EC:BDPV08}.
For a preimage attack on Feistel-\MiMC-Hash we have to, as the name suggest, compute a preimage to a given hash value $\alpha \in \Fq$.
We have two generic choices to do so.
First, we can guess the second permutation output value and then simply invert the permutation.
If the preimage is of the form $(\beta, 0)$, for some $\beta \in \Fq$, then the attack was successful.
Though, the success probability of this approach is $1 / q$, and $q$ is at least a $64$-bit prime number, which is too small for a practical attack.
Second, we can use an indeterminate $x_2$ for the second permutation output, then we have to find a solution for the equation
\begin{equation}
    \text{Feistel-}\MiMC
    \begin{pmatrix}
        x_1 \\ 0
    \end{pmatrix}
    =
    \begin{pmatrix}
        \alpha \\ x_2
    \end{pmatrix}
    .
\end{equation}
Further, for the preimage problem we have only one generic choice of polynomials to restrict all solutions to the base field: field equations.
\begin{ex}[Feistel-\MiMC-Hash preimage attack I]\label{Ex: Feistel-MiMC-Hash preimage I}
    Let $\Fq$ be a finite field, and let $r$ be the number of rounds.
    We can construct the polynomial system for Feistel-\MiMC-Hash from the one for the keyed permutation, see \Cref{Def: Feistel keyed iterated polynomial system}, by setting $y = 0$, $p_L = x_1$, $p_R = 0$, $c_L = \alpha$ and $c_R = x_2$, where $x_1$ and $x_2$ are indeterminates and $\alpha \in \Fq$ is the hash value.
    Moreover, we choose the DRL term order such that the intermediate state variables are naturally ordered, $x_1 > x_2$ and all intermediate state variables are bigger than $x_1$.
    Analog to \Cref{Prop: Feistel Groebner bases} we can compute the DRL Gr\"obner basis of the system by substituting the linear polynomials into the non-linear ones, but this time we do not have to remove any linear polynomial.
    Since $p_R = 0$ there are only $r - 1$ polynomials of degree $3$ in $r - 1$ variables.
    To find a solution (if it exists) we now either have to compute the LEX Gr\"obner basis and factor a polynomial of degree $3^{r - 1}$ or add the field equation for $x_2$ to the polynomial system.
    For the latter case we obtain the following bound on the solving degree
    \[
    \solvdeg_{DRL} \big( I_\text{hash} + (x_2^q - x_2) \big) \leq q + 2 \cdot r - 2.
    \]
\end{ex}

Analog to the field equation attack on \MiMC we can also compute the remainder of the field equation modulo the DRL Gr\"obner basis to further reduce the solving degree.
Note that the solving degrees of \Cref{Ex: MiMC solving degree I,Ex: Feistel-MiMC-Hash preimage I} differ only by $2$, therefore we refer to \Cref{Tab: MiMC sample values} for the complexity of Gr\"obner basis computations of Feistel-\MiMC-Hash.

    \section{Multivariate Ciphers in Generic Coordinates}\label{Sec: multivariate ciphers}
So far all our complexity estimates are only applicable to univariate ciphers and two branch Feistel networks.
Naturally, one would like to extend the theory to more advanced multivariate constructions.
Therefore, in \Cref{Sec: SPN} we derive that Substitution-Permutation Network (SPN) based ciphers are in generic coordinates, hence we can apply the Macaulay bound to estimate the solving degree.
In \Cref{Sec: Feistel} we study three classes of generalized Feistel Networks for which we derive efficient criteria to check whether the corresponding polynomial systems are in generic coordinates.

For starters, let us fix some notation.
Let $n, r \geq 1$ be integers, $n$ always denotes the number of blocks and $r$ the number of rounds of a cipher.
Throughout this section we will denote plaintext variables with $\mathbf{x} = (x_1, \dots, x_n)^\intercal$ and key variables with $\mathbf{y} = (y_1, \dots, y_n)^\intercal$.
With
\begin{equation}
    \begin{split}
        \mathcal{K}_\mathbf{y}: \Fqn \times \Fqn &\to \Fqn, \\
        (\mathbf{x}, \mathbf{y}) &\mapsto \mathbf{x} + \mathbf{y}
    \end{split}
\end{equation}
we denote the key addition function, and with
\begin{equation}
    \begin{split}
        \mathcal{A}: \Fqn &\to \Fqn, \\
        \mathbf{x} & \mapsto \mathbf{A} \mathbf{x} + \mathbf{c}
    \end{split}
\end{equation}
we denote affine permutations where $\mathbf{A} \in \GL{n}{\Fq}$ and $\mathbf{c} \in \Fq$.
For $1 \leq i \leq r$ let $\mathcal{A}^{(1)}, \dots, \mathcal{A}^{(r)}: \Fqn \to \Fqn$ be affine permutations and let $\mathcal{P}^{(1)}, \dots, \mathcal{P}^{(r)}: \Fqn \to \Fqn$ some arbitrary permutations.
Then a block cipher without key schedule is defined to be the following composition
\begin{equation}
    \begin{split}
        \mathcal{C}_{n, r}: \Fqn \times \Fqn &\to \Fqn, \\
        \left( \mathbf{x}, \mathbf{y} \right) &\mapsto \big( \mathcal{K}_\mathbf{y} \circ \mathcal{A}^{(r)} \circ \mathcal{P}^{(r)} \big) \circ \cdots \circ \big( \mathcal{K}_\mathbf{y} \circ \mathcal{A}^{(1)} \circ \mathcal{P}^{(1)} \big) (\mathbf{x} + \mathbf{y}),
    \end{split}
\end{equation}
where the composition is taken with respect to the plaintext variable.

For $1 \leq i \leq r - 1$, let $\mathbf{x}^{(i)} = \left( x_1^{(i)}, \dots, x_n^{(i)} \right)^\intercal$ denote intermediate state variables, and let $\mathbf{y} = \left( y_1, \dots, y_n \right)^\intercal$ denote the key variables.
Let $\mathbf{p}, \mathbf{c} \in \Fqn$ be a plain/ciphertext pair given by the block cipher $\mathcal{C}_{n, r}$
Since every function $\Fqn \to \Fq$ can be represented with polynomials, we define the multivariate keyed iterated polynomial system $\mathcal{F} = \left\{ \mathbf{f}^{(1)}, \dots, \mathbf{f}^{(r)} \right\} \subset \Fq \left[ \mathbf{x}^{(1)}, \dots, \mathbf{x}^{(r - 1)}, \mathbf{y} \right]$ for the cipher $\mathcal{C}_{n, r}$ as
\begin{equation}\label{Equ: multivariate keyed iterated polynomial system}
    \mathbf{f}^{(i)}
    =
    \begin{dcases}
        \mathbf{A}_1 \mathcal{P}^{(1)} \left( \mathbf{p + \mathbf{y}} \right) + \mathbf{c}_1 + \mathbf{y} - \mathbf{x}^{(1)}, & i = 1, \\
        \mathbf{A}_i \mathcal{P}^{(i)} \left( \mathbf{x}^{(i - 1)} \right) + \mathbf{c}_i + \mathbf{y} - \mathbf{x}^{(i)}, & 2 \leq i \leq r - 1, \\
        \mathbf{A}_r \mathcal{P}^{(r)} \left( \mathbf{x}^{(r - 1)} \right) + \mathbf{c}_r + \mathbf{y} - \mathbf{c}, & i = r.
    \end{dcases}
\end{equation}
If a key schedule is applied we have two options for the polynomial model.
Either we substitute the key schedule directly into \Cref{Equ: multivariate keyed iterated polynomial system} or we add additional iterated key schedule equations to $\mathcal{F}$.

\subsection{Substitution-Permutation Networks}\label{Sec: SPN}
In symmetric key cryptography the Substitution Permutation Network (SPN) is the most widely adopted strategy to construct block ciphers.
For example, the Advanced Encryption Standard (AES) \cite{AES,Daemen-AES} is an SPN.
Moreover, so-called partial SPNs have been adopted for AO designs \cite{EC:ARSTZ15,EC:GLRRS20,USENIX:GKRRS21,AFRICACRYPT:GraKhoSch23}.
We start with the formal definition of SPN-based ciphers.
\begin{defn}[SPN cipher]
    Let $\Fq$ be a finite field, and let $n, r \geq 1$ be integers.
    \begin{enumerate}
        \item Let $f_1, \dots, f_n \in \Fq [x]$ be permutation polynomials.
        Then the full Substitution Layer is defined as
        \begin{align}
            \mathcal{S}_{f_1, \dots, f_n}: \Fqn &\to \Fqn, \nonumber \\
            \left( x_1, \dots, x_n \right) &\mapsto \big( f_1 (x_1), \dots, f_n (x_n) \big). \nonumber
        \end{align}

        \item Let $f \in \Fq [x]$ be permutation polynomial.
        Then the partial Substitution Layer is defined as
        \begin{align}
            \mathcal{S}_f: \Fqn &\to \Fqn, \nonumber \\
            \left( x_1, \dots, x_n \right) &\mapsto \big( f (x_1), x_2, \dots, x_n \big). \nonumber
        \end{align}

        \item For $1 \leq n \leq r$, let $\mathcal{S}^{(i)}: \Fqn \to \Fqn$ be either a full or a partial Substitution Layer and let $\mathcal{A}_i: \Fqn \to \Fqn$ be an affine permutation.
        Then the SPN cipher is defined as
        \begin{align}
            \mathcal{C}_{n, r}: \Fqn \times \Fqn &\to \Fqn, \nonumber \\
            \left( \mathbf{x}, \mathbf{y} \right) &\mapsto \big( \mathcal{K}_\mathbf{y} \circ \mathcal{A}^{(r)} \circ \mathcal{S}^{(r)} \big) \circ \cdots \circ \big( \mathcal{K}_\mathbf{y} \circ \mathcal{A}^{(1)} \circ \mathcal{S}^{(1)} \big) (\mathbf{x} + \mathbf{y}), \nonumber
        \end{align}
        where the composition is taken with respect to the plaintext variable.
    \end{enumerate}
\end{defn}

Accordingly, a round where a full/partial Substitution Layer is applied is called a full/partial round.

Under a mild assumption on the first round of an SPN cipher $\mathcal{C}_{n, r}$ we can compute a DRL Gr\"obner basis of the multivariate keyed iterated polynomial system.
\begin{thm}\label{Th: SPN generic coordinates}
    Let $\Fq$ be a finite field, let $\overline{\Fq}$ be its algebraic closure, let $n, r \geq 1$ be integers, and let $\mathcal{C}_{n, r}: \Fqn \times \Fqn \to \Fqn$ be an SPN cipher such that $\mathcal{S}^{(1)}$ is a full SPN and every univariate permutation polynomial in $\mathcal{S}^{(1)}$ has degree greater than $1$.
    Let $\mathcal{F} = \left\{ \mathbf{f}^{(1)}, \dots, \mathbf{f}^{(r)} \right\} \subset P = \overline{\Fq} \left[ \mathbf{x}^{(1)}, \dots, \mathbf{x}^{(r - 1)}, \mathbf{y} \right]$ be the multivariate keyed iterated polynomial system for $\mathcal{C}_{n, r}$, and let
    \[
    \mathcal{G} = \left\{ \mathbf{A}_1^{-1} \mathbf{f}^{(1)}, \dots, \mathbf{A}_r^{-1} \mathbf{f}^{(r)} \right\}.
    \]
    Then
    \begin{enumerate}
        \item $\mathcal{G}$ is a DRL Gr\"obner basis.

        \item Every homogeneous ideal $I \subset P [x_0]$ such that $\mathcal{Z}_+ (I) \neq \emptyset$ and $\mathcal{G}^\homog \subset I$ is in generic coordinates.
    \end{enumerate}
\end{thm}
\begin{proof}
    For (1), we consider the DRL term order $x_1^{(1)} > \ldots > x_n^{(1)} > x_1^{(2)} > \ldots > x_n^{(n - 1)} > y_1 > \ldots > y_n$.
    Let $f_1^{(j)}, \dots, f_n^{(j)}$ denote the functions of the SPN in the $j$\textsuperscript{th} round, then
    \begin{align}
        \LM_{DRL} \left( \mathbf{A}_1^{-1} \mathbf{f}^{(1)} \right) &= \left( y_i^{\degree{f_i^{(1)}}} \right)_{1 \leq i \leq n}, \nonumber \\
        \LM_{DRL} \left( \mathbf{A}_1^{-1} \mathbf{f}^{(j)} \right) &= \left( {x_i^{(j)}}^{\degree{f_i^{(j)}}} \right)_{1 \leq i \leq n}, \quad 2 \leq j \leq r, \nonumber
    \end{align}
    since $\degree{f_i^{(1)}} > 1$ for all $i$ and $\mathbf{x}^{(1)} > \ldots > \mathbf{x}^{(r - 1)} > \mathbf{y}$.
    So the polynomials in $\mathcal{G}$ have pairwise coprime leading monomials and the claim follows from \cite[Chapter 2 §9 Theorem 3, Proposition 4]{Cox-Ideals}.

    For (2), this follows from (1) and \Cref{Cor: DRL Groebner basis in generic coordinates}.
\end{proof}

Let us now apply \Cref{Th: SPN generic coordinates} to a cipher that utilizes partial as well as full Substitution Layers: the \Hades strategy \cite{EC:GLRRS20}, a cipher family for MPC applications.
The keyed \Hades permutation starts with $r_f$ full rounds, then it applies $r_p$ partial rounds, and it finishes with another application of $r_f$ full rounds.
So, in total \Hades has $r = 2 \cdot r_f + r_p$ many rounds.
All SPNs apply the same univariate permutation $x^d$ for some appropriate $d$.
\Hades has an affine key schedule \cite[\S 3.1]{EC:GLRRS20}, and it is straight-forward to incorporate an affine key schedule into the multivariate keyed iterated polynomial system from \Cref{Equ: multivariate keyed iterated polynomial system}.
Moreover, an affine key schedule does not affect the proof of \Cref{Th: SPN generic coordinates} as long as the master key was added before application of the first Substitution Layer.
\begin{ex}[Solving degree bounds for \Hades]\label{Ex: Hades}
    Let $\Fq$ be a finite field, let $n \geq 1$ denote the number of branches, and let $d \in \mathbb{Z}_{> 1}$ be an integer such that $\gcd \left( d, q - 1 \right) = 1$.
    Let $r_f, r_p \geq 1$ denote the number full and partial rounds, and let $I_\Hades$ denote the \Hades ideal.
    Then by \Cref{Cor: Macauly bound} and \Cref{Th: SPN generic coordinates}
    \[
        \solvdeg_{DRL} \left( I_\Hades \right) \leq \left( d - 1 \right) \cdot \left( 2 \cdot n \cdot r_f + r_p \right) + 1.
    \]
    Now let $I_{\Hades, 1}$ and $I_{\Hades, 2}$ denote \Hades ideals for two different plain/ciphertext pairs.
    It is straight-forward to extend \Cref{Th: SPN generic coordinates} to $I_{\Hades, 1} + I_{\Hades, 2}$, cf.\ \Cref{Prop: variety intersection substituion properties} \ref{Item: variety intersection generic coordinates}, therefore by \Cref{Cor: Macauly bound}
    \[
        \solvdeg_{DRL} \left( I_{\Hades, 1} + I_{\Hades, 2} \right) \leq 2 \cdot (d - 1) \cdot (2 \cdot n \cdot r_f + r_p) + 1.
    \]
\end{ex}

The \Hades designers use \Cref{Equ: Groebner basis complexity} and the Macaulay bound (\Cref{Cor: Macauly bound}) to estimate the resistance of \Hades against Gr\"obner basis attacks, see \cite[\S 4.3]{EC:GLRRS20} and \cite[\S E.3]{EPRINT:GLRRS19}.
In particular, their \emph{second strategy} is the multivariate keyed iterated polynomial system from \Cref{Equ: multivariate keyed iterated polynomial system}.
To justify this approach the authors hypothesized that the \Hades polynomial system is a \emph{generic} polynomial system in the sense of Fr\"oberg's conjecture \cite{Froeberg-Conjecture,Pardue-Generic}.
With \Cref{Th: SPN generic coordinates} and \Cref{Ex: Hades} this hypothesis can be bypassed, and we have proven that the complexity estimation of the \Hades designers is indeed mathematically sound.

In \Cref{Tab: MiMC sample values} we provide complexity estimates in bits for a Gr\"obner basis computation of \Hades where we use the Macaulay bound of the keyed iterated \Hades polynomial system as minimal baseline of the solving degree for an optimal adversary with $\omega = 2$.
We assume that the key schedule equations have been substituted into \Cref{Equ: multivariate keyed iterated polynomial system}.
Recall from \Cref{Th: SPN generic coordinates} that a partial \Hades round only contributes one non-linear equation but $n - 1$ affine equations.
Therefore, we can eliminate $r_p \cdot (n - 1)$ many variables in the \Hades polynomial system in advance, i.e.\ after the elimination the number of variables and equations in the \Hades polynomial system is $2 \cdot n \cdot r_f + r_p$.
Note that this elimination does not affect the Macaulay bound from \Cref{Ex: Hades}.
We stress that after the substitution we exactly reproduce the complexity estimation of the second strategy \cite[p.~48-51]{EPRINT:GLRRS19}, though with a proof and not under a hypothesis.
For ease of computation we estimated the logarithm of the binomial coefficient with \Cref{Equ: bit complexity estimate}.

\begin{table}[H]
    \centering
    \caption{Complexity estimation of Gr\"obner basis computations for \Hades via the Macaulay bound for $I_\Hades$ with $n = 2$, and $\omega = 2$ over a finite field $\Fq$ such that $\gcd \left( d, q - 1 \right) = 1$.}
    \label{Tab: Hades values}
    \begin{tabular}{c | c c c | c c c }
        \toprule
        & \multicolumn{3}{ c | }{$d = 3$} & \multicolumn{3}{ c }{$d = 5$} \\
        \midrule

        $r_f$           & $3$     & $4$     & $5$     & $3$     & $4$     & $5$     \\
        $r_p$           & $13$    & $10$    & $5$     & $10$    & $10$    & $4$     \\
        $\kappa$ (bits) & $130.0$ & $135.4$ & $130.0$ & $149.0$ & $177.6$ & $163.3$ \\

        \bottomrule
    \end{tabular}
\end{table}
In \cite[Table~1]{EC:GLRRS20} round numbers for \Hades instantiations are proposed, as can be derived from our table all instantiations achieve a security level of $128$ bits already for $n = 2$ and $d = 3$.

We also mention that it is straight-forward to compute \Hades' quotient space dimension
\begin{equation}
    \dim_{\Fq} \left( I_\Hades \right) = d^{2 \cdot n \cdot r_f + r_p}.
\end{equation}

\subsection{Generalized Feistel Networks}\label{Sec: Feistel}
The second permutation that has been dominant in block cipher design in the past is the so-called Feistel Network, named after its inventor Horst Feistel.
For example, the predecessor of AES the Data Encryption Standard (DES) \cite{DES77} is based on the Feistel Network.
Moreover, so-called unbalanced generalized Feistel Networks have been proposed for AO designs \cite{ESORICS:AGPRRRRS19}.
We start with the formal definition of Feistel-based ciphers.
\begin{defn}[Generalized Feistel cipher]
    Let $\Fq$ be a finite field, and let $n, r \geq 1$ be integers.
    \begin{enumerate}
        \item For $1 \leq i \leq n - 1$, let $f_i \in \Fq [x_{i + 1}, \dots, x_n]$ be a polynomial.
        Then the generalized Feistel Layer is defined as
        \begin{align}
            \mathcal{F}_{f_1, \dots, f_{n - 1}}: \Fqn &\to \Fqn, \nonumber \\
            \left( x_1, \dots, x_n \right) &\mapsto \big( x_1 + f_1 (x_2, \dots, x_n), \dots, x_{n - 1} + f_{n - 1} (x_n), x_n \big). \nonumber
        \end{align}

        \item For $1 \leq n \leq r$, let $\mathcal{F}^{(i)}: \Fqn \to \Fqn$ be a generalized Feistel Layer and let $\mathcal{A}_i: \Fqn \to \Fqn$ be an affine permutation.
        Then the Feistel cipher is defined as
        \begin{align}
            \mathcal{C}_{n, r}: \Fqn \times \Fqn &\to \Fqn, \nonumber \\
            \left( \mathbf{x}, \mathbf{y} \right) &\mapsto \big( \mathcal{K}_\mathbf{y} \circ \mathcal{A}^{(r)} \circ \mathcal{F}^{(r)} \big) \circ \cdots \circ \big( \mathcal{K}_\mathbf{y} \circ \mathcal{A}^{(1)} \circ \mathcal{F}^{(1)} \big) (\mathbf{x} + \mathbf{y}), \nonumber
        \end{align}
        where the composition is taken with respect to the plaintext variable.
    \end{enumerate}
\end{defn}

For special types of Feistel ciphers we can derive efficient criteria to verify whether the corresponding multivariate keyed iterated polynomial system is in generic coordinates.
\begin{thm}\label{Th: Feistel generic generators criteria}
    Let $\Fq$ be a finite field, let $\overline{\Fq}$ be its algebraic closure, let $n, r \geq 1$ be integers, and let $\mathcal{C}_{n, r}: \Fqn \times \Fqn \to \Fqn$ be a Feistel cipher.
    Let $\mathcal{F} = \left\{ \mathbf{f}^{(1)}, \dots, \mathbf{f}^{(r)} \right\} \subset P = \overline{\Fq} \left[ \mathbf{x}^{(1)}, \dots, \mathbf{x}^{(r - 1)}, \mathbf{y} \right]$ be a multivariate keyed iterated polynomial system for $\mathcal{C}_{n, r}$.
    \begin{enumerate}
        \item\label{Item: expanding} For $1 \leq i \leq r$, let $f^{(i)} \in \Fq [x_n]$ be polynomials such that $\degree{f^{(i)}} > 1$, and assume that the $i$\textsuperscript{th} Feistel Layer of $\mathcal{C}_{n, r}$ is $\mathcal{F}_{f^{(i)}, \dots, f^{(i)}}$.
        Let the polynomial system $\mathcal{G} = \left\{ \mathbf{g}^{(1)}, \dots, \mathbf{g}^{(r)} \right\}$ be defined as follows
        \[
        \left( \mathbf{g}^{(i)} \right)_j =
        \begin{dcases}
            \left( \mathbf{A}_i^{-1} \mathbf{f}^{(i)} \right)_j, & j = 1, n, \\
            \left( \mathbf{A}_i^{-1} \mathbf{f}^{(i)} \right)_j - \left( \mathbf{A}_i^{-1} \mathbf{f}^{(i)} \right)_1, & 2 \leq j \leq n - 1.
        \end{dcases}
        \]
        Then every homogeneous ideal $I \subset P [x_0]$ such that $\mathcal{Z}_+ (I) \neq \emptyset$ and $\mathcal{G}^\homog \subset I$ is in generic coordinates if the following linear system has rank $r \cdot \left( n - 1\right)$
        \begin{alignat*}{1}
            y_j - y_1 + \bigg( \mathbf{A}_1^{-1} \Big( \boldsymbol{\hat{x}}^{(1)} - \boldsymbol{\hat{y}} \Big) \bigg)_1 + \bigg( \mathbf{A}_1^{-1} \Big( \boldsymbol{\hat{y}} - \boldsymbol{\hat{x}}^{(1)} \Big) \bigg)_j &= 0, \\
            \bigg( \mathbf{A}_1^{-1} \Big( \boldsymbol{\hat{y}} - \boldsymbol{\hat{x}}^{(1)} \Big) \bigg)_n &= 0, \\
            x_j^{(i - 1)} - x_1^{(i - 1)} + \bigg( \mathbf{A}_i^{-1} \Big( \boldsymbol{\hat{x}}^{(i)} - \boldsymbol{\hat{y}} \Big) \bigg)_1 + \bigg( \mathbf{A}_i^{-1} \Big( \boldsymbol{\hat{y}} - \boldsymbol{\hat{x}}^{(i)} \Big) \bigg)_j &= 0, \\
            \bigg( \mathbf{A}_i^{-1} \Big( \boldsymbol{\hat{y}} - \boldsymbol{\hat{x}}^{(i)} \Big) \bigg)_n &= 0, \\
            x_j^{(r - 1)} - x_1^{(r - 1)} + \bigg( \mathbf{A}_r^{-1} \Big( \boldsymbol{\hat{x}}^{(r - 1)} - \boldsymbol{\hat{y}} \Big) \bigg)_1 + \bigg( \mathbf{A}_r^{-1} \Big( \boldsymbol{\hat{y}} - \boldsymbol{\hat{x}}^{(r - 1)} \Big) \bigg)_j &= 0, \\
            \left( \mathbf{A}_r^{-1} \boldsymbol{\hat{y}} \right)_n &= 0,
        \end{alignat*}
        where $2 \leq i \leq r - 1$, $2 \leq j \leq n - 1$, $\boldsymbol{\hat{x}}^{(i)} = \left( x_1^{(i)}, \dots, x_{n - 1}^{(i)}, 0 \right)$ and $\boldsymbol{\hat{y}} = (y_1, \dots, y_{n - 1}, 0)$.

        \item\label{Item: strong contracting} For $1 \leq i \leq r$ and $1 \leq j \leq n - 1$, let $f_j^{(i)} \in \Fq [x_{i + 1}, \dots, x_n]$ be polynomials such that $\degree{f^{(i)}} > 1$ and the monomial $x_{i + 1}^{\degree{f_j^{(i)}}}$ is present in $f_j^{(i)}$, and assume that the $i$\textsuperscript{th} Feistel Layer of $\mathcal{C}_{n, r}$ is $\mathcal{F}_{f_1^{(i)}, \dots, f_{n - 1}^{(i)}}$.
        Let
        \[
        \mathcal{G} = \left\{ \mathbf{A}_1^{-1} \mathbf{f}^{(1)}, \dots, \mathbf{A}_r^{-1} \mathbf{f}^{(r)} \right\}.
        \]
        Then every homogeneous ideal $I \subset P [x_0]$ such that $\mathcal{Z}_+ (I) \neq \emptyset$ and $\mathcal{G}^\homog \subset I$ is in generic coordinates if the following linear system has rank $r$
        \begin{align*}
            \bigg( \mathbf{A}_1^{-1} \Big( \boldsymbol{\hat{y}} - \boldsymbol{\hat{x}}^{(1)} \Big) \bigg)_n &= 0, \\
            \bigg( \mathbf{A}_i^{-1} \Big( \boldsymbol{\hat{y}} - \boldsymbol{\hat{x}}^{(i)} \Big) \bigg)_n &= 0, \\
            \left( \mathbf{A}_r^{-1} \boldsymbol{\hat{y}} \right)_n &= 0,
        \end{align*}
        where $2 \leq i \leq r - 1$, $\boldsymbol{\hat{x}}^{(i)} = \left( x_1^{(i)}, 0 \dots, 0 \right)$ and $\boldsymbol{\hat{y}} = (y_1, 0, \dots, 0)$.

        \item\label{Item: contracting} For $1 \leq i \leq r$, let $f^{(i)} \in \Fq [x_2, \dots, x_n]$ be such that
        \[
        f^{(i)} (x_2, \dots, x_n) = \hat{f}^{(i)} \left( \sum_{j = 2}^{n} a_{i, j} \cdot x_j \right)
        \]
        with $\hat{f}^{(i)} \in \Fq [x]$ such that $\degree{\hat{f}^{(i)}} > 1$ and $a_{i, 2}, \dots, a_{i, n} \in \Fq$, and assume that the $i$\textsuperscript{th} Feistel Layer of $\mathcal{C}_{n, r}$ is $\mathcal{F}_{f^{(i)}, 0, \dots, 0}$.
        Let
        \[
        \mathcal{G} = \left\{ \mathbf{A}_1^{-1} \mathbf{f}^{(1)}, \dots, \mathbf{A}_r^{-1} \mathbf{f}^{(r)} \right\}.
        \]
        Then every homogeneous ideal $I \subset P [x_0]$ such that $\mathcal{Z}_+ (I) \neq \emptyset$ and $\mathcal{G}^\homog \subset I$ is in generic coordinates if the following linear system has rank $r \cdot n$
        \begin{align*}
            \sum_{k = 2}^{n} a_{1, k} \cdot y_k &= 0, \\
            \bigg( \mathbf{A}_1^{-1} \Big( \mathbf{y} - \mathbf{x}^{(1)} \Big) \bigg)_j &= 0, \\
            \sum_{k = 2}^{n} a_{i, k} \cdot x_k^{(i - 1)} &= 0, \\
            x_j^{(i - 1)} + \bigg( \mathbf{A}_1^{-1} \Big( \mathbf{y} - \mathbf{x}^{(i)} \Big) \bigg)_j &= 0, \\
            \sum_{k = 2}^{n} a_{r, k} \cdot x_k^{(r - 1)}  &= 0, \\
            x_j^{(r - 1)} + \left( \mathbf{A}^{-1} \mathbf{y} \right)_j &= 0,
        \end{align*}
        where  $2 \leq i \leq r - 1$ and $2 \leq j \leq n$.
    \end{enumerate}
\end{thm}
\begin{proof}
    For all cases we show that $\sqrt{\mathcal{G}^\topcomp} = (x_1, \dots, x_n)$.

    For (1), note that for all $1 \leq i \leq r$ we have that the degree of the first component of $\mathbf{g}^{(i)}$ is $\degree{f^{(i)}}$ and $1$ for the other components.
    Substituting $x_0 = 0$ into $\mathcal{G}^\homog$ we yield from the first components of the $\left( \mathbf{g}^{(i)} \right)^\homog$'s that $y_n^{\degree{f^{(1)}}} = {x_n^{(1)}}^{\degree{f^{(2)}}} = \ldots = {x_n^{(r - 1)}}^{\degree{f^{(r)}}} = 0$ so also $y = x_n^{(1)} = \ldots = x_n^{(r - 1)} = 0$.
    Now we substitute these coordinates into the remaining equations.
    This yields the linear system from the assertion.
    If the linear system is of rank $r \cdot \left( n - 1\right)$, then $\sqrt{\mathcal{G}^\homog} = (x_1, \dots, x_n)$.

    For (2), note that for $1 \leq i \leq r$ and $1 \leq j \leq n - 1$ we have that $\degree{\left( A_i^{-1} \mathbf{f}^{(i)} \right)_j} = \degree{f_j^{(i)}}$.
    Now we substitute $x_0 = 0$ into $\mathcal{G}^\homog$, in the $i$\textsuperscript{th} round in the $(n - 1)$\textsuperscript{th} component this yields ${x_n^{(i)}}^{\degree{f_{n - 1}^{(i)}}} = 0$.
    Inductively we now work through all higher branches in the $i$\textsuperscript{th} round and then through all rounds to obtain that
    $y_2 = \ldots = y_n = x_2^{(1)} = \ldots = x_n^{(1)} = \ldots = x_2^{(r - 1)} = \ldots = x_n^{(r - 1)} = 0$.
    Substituting these variables into the remaining equations that come from the last branch of the  $\mathbf{f}^{(i)}$'s yields the linear system from the assertion which proves the claim.

    For (3), after substituting $x_0 = 0$ into $\mathcal{G}^\homog$ we obtain for the first branch of each round
    \begin{alignat*}{4}
        \Bigg( &\sum_{k = 2}^{n} a_{1, k} \cdot y_k \Bigg)^{\degree{\hat{f}^{(1)}}}
        &&= \Bigg( &&\sum_{k = 2}^{n} a_{i, k} \cdot x_k^{(i - 1)} \Bigg)^{\degree{\hat{f}^{(i)}}}
        &&= 0 \\
        & &&\Longrightarrow \\
        &\sum_{k = 2}^{n} a_{1, k} \cdot y_k
        &&= &&\sum_{k = 2}^{n} a_{i, k} \cdot x_k^{(i - 1)}
        &&= 0,
    \end{alignat*}
    where $2 \leq i \leq r$.
    Combining these linear equations with the remaining equations from $\mathcal{G}^\homog$ we obtain the linear system from the assertion.
\end{proof}

The Feistel Networks from \Cref{Th: Feistel generic generators criteria} \ref{Item: expanding} and \ref{Item: contracting} are also known as \emph{expanding round function} (erf) and \emph{contracting round function} (crf) respectively.
An example for block ciphers with these round functions is the \GMiMC family \cite[\S 2.1]{ESORICS:AGPRRRRS19}, which is targeted for MPC applications.
Moreover, the designers of \GMiMC use \Cref{Equ: Groebner basis complexity} and the Macaulay bound (\Cref{Cor: Macauly bound}) to estimate the resistance of \GMiMC against Gr\"obner basis attacks, see \cite[\S 4.1.1]{ESORICS:AGPRRRRS19}.\footnote{
    For completeness, we mention that the \GMiMC designers did not analyze the keyed iterated polynomial system (\Cref{Equ: multivariate keyed iterated polynomial system}).
    They only studied systems where all rounds are substituted into each other, i.e.\ one has $n$ equations in $n$ variables.}
To justify this approach the authors hypothesized that the \GMiMC polynomial systems are \emph{generic} polynomial systems in the sense of Fr\"oberg's conjecture \cite{Froeberg-Conjecture,Pardue-Generic}.
With \Cref{Th: Feistel generic generators criteria} this hypothesis can be bypassed for \GMiMC without a key schedule.
For \GMiMC an affine key schedule was proposed, hence one can extend \Cref{Th: Feistel generic generators criteria} to this scenario by replacing the key variables with intermediate key variables after the first round and by adding the linear part of the affine key schedule to the linear systems.
Thus, we have derived efficient criteria to verify that the complexity estimations of the \GMiMC designers can indeed be mathematically sound.
\begin{ex}[Solving degree bounds for \GMiMC]\label{Ex: GMiMC}
    Let $\Fq$ be a finite field, let $n, r \geq 1$ denote the number of branches and rounds, and let $d \geq 1$ be the degree of the degree increasing function.
    In the proof of \Cref{Th: Feistel generic generators criteria} we saw that the \GMiMCerf polynomial system can be transformed so that there is only one non-linear polynomial in every round.
    Therefore, \GMiMCcrf and \GMiMCerf have the same Macaulay bound.
    Let $I_\GMiMC$ be a \GMiMC ideal and assume that $n$ and $r$ are such that the corresponding matrix from \Cref{Th: Feistel generic generators criteria} has full rank, i.e.\ \GMiMC is in generic coordinates.
    Therefore, by \Cref{Cor: Macauly bound}
    \[
    \solvdeg_{DRL} \left( I_\GMiMC \right) \leq (d - 1) \cdot r + 1.
    \]
    Now let $I_{\GMiMC, 1}$ and $I_{\GMiMC, 2}$ denote \GMiMC ideals for two different plain/ciphertext pairs.
    It is straight-forward to extend \Cref{Th: Feistel generic generators criteria} to $I_{\GMiMC, 1} + I_{\GMiMC, 2}$, cf. \Cref{Prop: variety intersection substituion properties} \ref{Item: variety intersection generic coordinates}.
    Provided that $n$ and $r$ are such that $I_{\GMiMC, 1} + I_{\GMiMC, 2}$ is in generic coordinates we have by \Cref{Cor: Macauly bound}
    \[
    \solvdeg_{DRL} \left( I_{\GMiMC, 1} + I_{\GMiMC, 2} \right) \leq 2 \cdot (d - 1) \cdot r + 1.
    \]
\end{ex}

For small primes we applied \Cref{Th: Feistel generic generators criteria} to \GMiMCcrf and \GMiMCerf without key schedules.
Depending on the parameters $n$ and $r$ we noticed a highly regular pattern when the matrices from the theorem have full rank.
In \Cref{Tab: GMiMC generic generators shift permutation} we record this pattern for small sample parameters.
\begin{table}[H]
    \centering
    \caption{Matrix criteria from \Cref{Th: Feistel generic generators criteria} for sample parameters for \GMiMCcrf and \GMiMCerf with the shift permutation $(x_1, \dots, x_n) \mapsto (x_n, x_1, \dots, x_{n - 1})$ in the affine layer and without key schedules.}
    \label{Tab: GMiMC generic generators shift permutation}
    \begin{tabular}{ c | c  || c | c || c | c }
        \toprule
        \multicolumn{2}{ c || }{$n = 3$} & \multicolumn{2}{ c || }{$n = 4$} & \multicolumn{2}{ c }{$n = 5$} \\
        \midrule

        $r$  & Full rank      & $r$  & Full rank      & $r$  & Full rank      \\
        \midrule

        $10$ & \texttt{True}  & $12$ & \texttt{True}  & $10$ & \texttt{True}  \\
        $11$ & \texttt{False} & $13$ & \texttt{True}  & $11$ & \texttt{False} \\
        $12$ & \texttt{True}  & $14$ & \texttt{False} & $12$ & \texttt{True}  \\
        $13$ & \texttt{False} & $15$ & \texttt{True}  & $13$ & \texttt{False} \\
             &                & $16$ & \texttt{True}  &      &                \\
             &                & $17$ & \texttt{False} &      &                \\

        \bottomrule
    \end{tabular}
\end{table}
We observed that for the shift permutation $(x_1, \dots, x_n) \mapsto (x_n, x_1, \dots, x_{n - 1})$ in the affine layer the matrix criteria for \GMiMCcrf and \GMiMCerf behave identical.
On the other hand, if we instantiate \GMiMC with the circulant matrix $\circulant (1, \dots, n)$\footnote{
    We understand circulant matrices as right shift circulant matrices, i.e.\
    \[
    \circulant (a_1, \dots, a_n) =
    \begin{pmatrix}
        a_1 & a_2 & \dots  & a_{n - 1} & a_n       \\
        a_n & a_1 & \dots  & a_{n - 2} & a_{n - 1} \\
            &     & \vdots &           &           \\
        a_2 & a_3 & \dots  & a_n       & a_1
    \end{pmatrix}
    .
    \]}, then we observed that \GMiMCerf is always in generic coordinates and for \GMiMCcrf the criterion is identical to \Cref{Tab: GMiMC generic generators shift permutation}.

In \Cref{Tab: GMiMC complexity} we provide complexity estimates in bits for a Gr\"obner basis computation of \GMiMC where we use the Macaulay bound of the keyed iterated \GMiMC polynomial system as minimal baseline of the solving degree for an optimal adversary with $\omega = 2$.
We assume that the key schedule equations have been substituted into \Cref{Equ: multivariate keyed iterated polynomial system}.
Also, recall from \Cref{Th: Feistel generic generators criteria} that we can transform a \GMiMC polynomial system so that every round contains only one non-linear polynomial.
Thus, we can use the affine equations to eliminate $r \cdot (n - 1)$ many variables in the \GMiMC polynomial system in advance, i.e.\ after the elimination the number of variables and equations in the \GMiMC polynomial system is $r$.
Note that the elimination does not affect the Macaulay bound in \Cref{Ex: GMiMC}.
For ease of computation we estimated the logarithm of the binomial coefficient with \Cref{Equ: bit complexity estimate}.
\begin{table}[H]
    \centering
    \caption{Complexity estimation of Gr\"obner basis computations for \GMiMC via the Macaulay bound with $\omega = 2$ over any finite field $\Fq$.}
    \label{Tab: GMiMC complexity}
    \begin{tabular}{c | c c c | c c c }
        \toprule
        & \multicolumn{3}{ c | }{$d = 3$} & \multicolumn{3}{ c }{$d = 5$} \\
        \midrule

        $r$              & $10$   & $25$    & $50$    & $10$   & $25$    & $50$    \\
        $\kappa$ (bits)  & $48.6$ & $130.0$ & $266.7$ & $63.5$ & $170.5$ & $350.0$ \\

        \bottomrule
    \end{tabular}
\end{table}

In \cite[Table~7]{EPRINT:AGPRRRRS19} round number for \GMiMCerf instantiations are proposed, as can be derived from our table all instantiations achieve a security level of $128$.

\subsection{The Problem With Sponge Constructions \& Generic Coordinates}
Let us return to the sponge construction \cite{Sponge,EC:BDPV08}.
Let $\mathcal{P}: \Fqn \to \Fqn$ be an arbitrary permutation which we instantiate in sponge mode with capacity $1 < c < n $ and rate $r = n - c$.
Let $\texttt{IV} \in \Fq^c$ be a fixed initial value, and let $\alpha \in \Fq$ be a hash output.
To find a preimage $\mathbf{x} \in \Fq^r$ we have to solve the equation
\begin{equation}
    \mathcal{P}
    \begin{pmatrix}
        \mathbf{x} \\ \texttt{IV}
    \end{pmatrix}
    =
    \begin{pmatrix}
        \alpha \\
        \mathbf{y}
    \end{pmatrix}
    ,
\end{equation}
where $\mathbf{y} \in \Fq^{n - 1}$ is an indeterminate variable.
First, we observe that this polynomial system is only fully determined if $c = n - 1$, else one always has $r + n - 1 > n$ many variables for $\mathbf{x}$ and $\mathbf{y}$.
Otherwise, we have to guess some entries of $\mathbf{x}$ and $\mathbf{y}$ which we expect to be successful with probability $1 / q$.
Second, if we model the sponge $\mathcal{P}$ with iterated polynomials, then the Caminata-Gorla technique (\Cref{Sec: Caminata Gorla technique}) will fail whenever the last round of $\mathcal{P}$ is non-linear in all its components.
In this case, after homogenizing the keyed iterated polynomial system and setting $x_0 = 0$ we will always remove the variables $\mathbf{y}$ from the equations.
So \Cref{Th: generic coordinates and highest degree components} \ref{Item: radical} cannot be satisfied, and the naive homogenization of a sponge polynomial system cannot be in generic coordinates.

We illustrate this property with a simple example.
\begin{ex}
    We work over the field $\F_5$.
    We consider an SPN sponge function based on the cubing map with $n = 2$ and $r = 3$ where the first and the last round are full SPNs and the middle round is a partial SPN.
    In every round the mixing matrix is $\circulant (1, 2)$ and all round constants are $\mathbf{0}$.
    The matrix is also applied before application of the first SPN.
    We illustrate this sponge function in \Cref{Fig: hash}.
    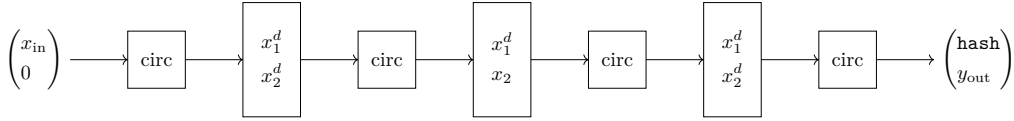
\begin{figure}[H]
        \centering
        \caption{Illustration of a simple SPN sponge function.}
        \label{Fig: hash}
        \resizebox{\textwidth}{!}{
            \begin{tikzpicture}
                \draw[->] (0, 0) node[left] {$\left(
                    \begin{aligned}
                        &x_\text{in} \\ &0
                    \end{aligned}
                    \right)$} -- ++(1, 0);

                \draw (1, -0.5) rectangle node {$\circulant$} ++(1, 1);

                \draw[->] (2, 0) -- ++(1, 0);

                \draw (3, -1) rectangle node {$
                    \begin{aligned}
                        &x_1^d \\ &x_2^d
                    \end{aligned}
                    $} ++(1, 2);

                \draw[->] (4, 0) -- ++(1, 0);

                \draw (5, -0.5) rectangle node {$\circulant$} ++(1, 1);

                \draw[->] (6, 0) -- ++(1, 0);

                \draw (7, -1) rectangle node (A) {$
                    \begin{aligned}
                        &x_1^d \\ &x_2
                    \end{aligned}
                    $} ++(1, 2);

                \draw[->] (8, 0) -- ++(1, 0);

                \draw (9, -0.5) rectangle node {$\circulant$} ++(1, 1);

                \draw[->] (10, 0) -- ++(1, 0);

                \draw (11, -1) rectangle node {$
                    \begin{aligned}
                        &x_1^d \\ &x_2^d
                    \end{aligned}
                    $} ++(1, 2);

                \draw[->] (12, 0) -- ++(1, 0);

                \draw (13, -0.5) rectangle node {$\circulant$} ++(1, 1);

                \draw[->] (14, 0) -- ++(1, 0) node[right] {$\left(
                    \begin{aligned}
                        &\texttt{hash} \\ &y_\text{out}
                    \end{aligned}
                    \right)$};
            \end{tikzpicture}
        }
    \end{figure}
    For hash value $0$ the iterated polynomial system $\mathcal{F} \subset \F_5 \left[ x_1^{(1)}, x_2^{(1)}, x_1^{(2)}, x_2^{(2)} x_\text{in}, y_\text{out} \right]$ is
    \begin{align*}
        x_\text{in}^3 + 2 \cdot x_1^{(1)} + x_2^{(1)} &= 0, \\
        -2 \cdot x_\text{in}^3 + x_1^{(1)} + 2 \cdot x_2^{(1)} &= 0, \\
        {x_1^{(1)}}^3 + 2 \cdot x_1^{(2)} + x_2^{(2)} &= 0, \\
        x_2^{(1)} + x_1^{(2)} + 2 \cdot x_2^{(2)} &= 0, \\
        {x_1^{(2)}}^3 + y_\text{out} &= 0, \\
        {x_2^{(2)}}^3 + 2 \cdot y_\text{out} &= 0.
    \end{align*}
    Note that $\left( \mathcal{F} \right)$ is zero-dimensional.
    Let $x_0$ denote the homogenization variable.
    Then $I^{\sat} = \left( \mathcal{F}^\homog \right)^{\sat}$ is generated by
    \begin{align*}
        x_\text{in}^3 + 2 \cdot x_1^{(1)} \cdot x_0 &= 0, \\
        x_2^{(1)} + x_1^{(2) + 2 \cdot x_2{(2)}} &= 0, \\
        \left( x_1^{(2)} + x_2^{(2)} \right) \cdot x_0^2 &= 0, \\
        {x_2^{(2)}}^3 + 2 \cdot y_\text{out} \cdot x_0^2 &= 0, \\
        {x_1^{(1)}}^3 + y_\text{out} \cdot x_0^2 &= 0, \\
        {x_1^{(1)}}^3 + 2 \cdot x_2^{(2)} \cdot x_0^2 &= 0, \\
        y_\text{out} \cdot x_0^4 &= 0.
    \end{align*}
    Hence, after reducing modulo $(x_0)$ we remove the variable $y_\text{out}$.
\end{ex}

To resolve this problem we have to add additional polynomials to the system.
Over finite fields we can always add the field equations for $\mathbf{y}$ though for AO designs this introduces high degree equations to a low degree polynomial system.
On the other hand, we could add the inverse of the last round of an iterated construction to the polynomial system to introduce polynomials with leading monomials in $\mathbf{y}$.
Though, in general we also expect that this trick introduces high degree equations.

\subsection{The Problem With Non-Affine Key Schedules \& Generic Coordinates}
We face a similar obstacle for the Caminata-Gorla technique if we deploy a non-affine key schedule.
For sake of example let us return to \MiMC with the key schedule
\begin{equation}
    y_i = y_{i - 1}^3,
\end{equation}
for $2 \leq i \leq r$ and $y_1 \in \Fq$ the master key.
We then add the $i$\textsuperscript{th} key in the $i$\textsuperscript{th} round.
Obviously, we then have to add the equations $y_{i - 1}^3 - y_i = 0$ to the \MiMC keyed iterated polynomial system.
Now we homogenize this new system and set $x_0 = 0$, like in \Cref{Th: iterated system generic coordinates} we can iterate through the rounds to deduce that $y_1 = \ldots = y_{r - 1} = x_1 = \ldots = x_{r - 2} = 0$.
But for the last round we obtain that $y_r + x_{r - 1} = 0$, and we do not have any more equations left to cancel one of the variables.
Again, we would have to add polynomials to the system to fix our method like the field equations, or if possible the inverse of the last key schedule equation.

    \section{Polynomials With Degree Falls \& The Satiety}\label{Sec: satiety and polynomials with degree falls}
We now return to studying \MiMC, Feistel-\MiMC and Feistel-\MiMC-Hash.
In \Cref{Sec: applications} we derived solving degree estimates for various attacks on these primitives.
A natural question for the cryptanalyst is tightness of these bounds.
To partially answer this question we derive Castelnuovo-Mumford regularity lower bounds for the attacks on these primitives.
Essentially, if we find a non-trivial lower bound for the Castelnuovo-Mumford regularity, then regularity-based complexity estimates can never improve upon the lower bound.

In this section we develop the theoretical foundation for our regularity lower bounds.
First we introduce the notion of last fall degree of $\mathcal{F} \subset P$, that is the largest $d \in \mathbb{Z} \cup \{ \infty \}$ such that the row space of the inhomogeneous Macaulay matrix $M_{\leq d}$ is unequal to $(\mathcal{F})_{\leq d}$ (as $K$-vector space).
Then we prove that in generic coordinates the last fall degree of $\mathcal{F}$ is equal to the satiety of $\mathcal{F}^\homog$, another invariant closely related to the regularity.

Let $I \subset P = K [x_0, \dots, x_n]$ be a homogeneous ideal, it is well-known that the saturation $I^{\sat} = I : \mathfrak{m}^\infty$ is the unique largest ideal $J \subset P$ such that there exists $m \geq 0$ and for all $l \geq m$ one has $I_l = J_l$.
This motivates the following definition.
\begin{defn}\label{Def: satiety}
    Let $I \subset K [x_0, \dots, x_n]$ be a homogeneous ideal.
    The satiety of $I$, denoted by $\sat \left( I \right)$ is the smallest positive integer $m$ such that $I_l = I_l^{\sat}$ for all $l \geq m$.
\end{defn}

We recall some properties of the satiety.
If $x_0 \nmid 0 \mod I^{\sat}$, then by \cite[Lemma~1.8]{BayerStillman} one has that
\begin{equation}\label{Equ: satiety less than regularity}
    \sat \left( I \right) \leq \reg \left( I \right),
\end{equation}
and by \cite[Proposition~2.2]{Hashemi-EfficientComputation} one has that
\begin{equation}
    \sat \left( I \right) = \sat \big( \inid_{DRL} (I) \big).
\end{equation}
Moreover, if $I$ is in generic coordinates and $\mathcal{Z}_+ (I) \neq \emptyset$, then by \cite[Theorem~2.30]{Green-InitialIdeals} one has that
\begin{equation}
    \reg \left( I \right) = \max \left\{ \sat \left( I \right), \reg \left( I^{\sat} \right) \right\}.
\end{equation}

Let $\mathcal{F} = \{ f_1, \dots, f_m \} \subset P = K [x_1, \dots, x_n]$ be a polynomial system, and let $M_{\leq d}$ be the inhomogeneous Macaulay matrix in degree $d$.
We denote with
\begin{equation}
    W_{\mathcal{F}, d} = \left\{ f \in P \, \middle | \, f = \sum_{i = 1}^{m} g_i \cdot f_i,\ \degree{g_i \cdot f_i} \leq d \right\}
\end{equation}
the row space of $M_{\leq d}$.
Next we define the last fall degree.
\begin{defn}\label{Def: last fall degree}
    Let $K$ be a field, and let $\mathcal{F} \subset K [x_1, \dots, x_n]$ be a polynomial system.
    \begin{enumerate}
        \item For any $f \in (\mathcal{F})$ let
        \[
        d_f = \min \{ d \in \mathbb{Z}_{\geq 0} \mid f \in W_{\mathcal{F}, d} \}.
        \]

        \item If $d_f > \deg \left( f \right)$, then we say that $f$ has a degree fall in degree $d_f$.
        We say that $\mathcal{F}$ has a degree fall if there is an $f \in (\mathcal{F})$ such that $f$ has a degree fall.
        Else we say that $\mathcal{F}$ has no degree falls.

        \item Let $W_{\mathcal{F}, \infty} = (\mathcal{F})$ and $V_{\mathcal{F}, -1} = \emptyset$.
        The last fall degree of $\mathcal{F}$ is
        \[
        d_\mathcal{F} = \min \left\{ d \in \mathbb{Z}_{\geq 0} \cup \{ \infty \} \mid f \in W_{\mathcal{F}, \max \left\{ d, \degree{f} \right\}} \text{ for all } f \in (\mathcal{F}) \right\}.
        \]
    \end{enumerate}
\end{defn}

Note that the definition implies that for all $d \geq d_\mathcal{F}$ we have that $W_{\mathcal{F}, d} = (\mathcal{F}) \cap P_{\leq d} = (\mathcal{F})_{\leq d}$.

Next we collect some alternative characterizations of the last fall degree.
\begin{prop}\label{Prop: alternatice charcterizations of last fall degree}
    Let $K$ be a field, and let $\mathcal{F} \subset P = K [x_1, \dots, x_n]$ be a polynomial system.
    \begin{enumerate}
        \item If there exists a largest $d \in \mathbb{Z}_{\geq 0}$ such that $W_{\mathcal{F}, d} \cap P_{\leq d - 1} \neq W_{\mathcal{F}, d - 1}$, then $d_\mathcal{F} = d$.
        Else $d_\mathcal{F} = \infty$.

        \item $d_\mathcal{F} \geq \sup \left\{ d_f \mid f \in (\mathcal{F}),\ d_f > \degree{f} \right\}$.

        \item If $d_\mathcal{F} < \infty$, then $d_\mathcal{F} = \max \left\{ d_f \mid f \in (\mathcal{F}),\ d_f > \degree{f} \right\}$.
    \end{enumerate}
\end{prop}
\begin{proof}
    For (1), let $d < \infty$ be as asserted.
    By definition of the last fall degree $d_\mathcal{F} \geq d$.
    If $f \in (\mathcal{F})$ is such that $\degree{f} \leq d_\mathcal{F}$, then by definition of the last fall degree $f \in W_{\mathcal{F}, d_\mathcal{F}}$.
    Therefore, we have that
    \[
    W_{\mathcal{F}, d_\mathcal{F}} \cap P_{\leq d_\mathcal{F} - 1} = (\mathcal{F}) \cap P_{\leq d_\mathcal{F} - 1} \neq W_{\mathcal{F}, d_\mathcal{F} - 1},
    \]
    which implies that $d_\mathcal{F} \leq d$.
    If such a $d \in \mathbb{Z}_{\geq 0}$ does not exist, then obviously $d_\mathcal{F} = \infty$.

    For (2), let $f \in (\mathcal{F})$ be such that $d_f > \degree{f}$.
    Then for all $d < d_f$ we have $f \notin W_{\mathcal{F}, d}$.
    Therefore, by definition of the last fall degree $d_\mathcal{F} > d_f - 1$.
    Hence, $d_\mathcal{F} \geq \sup \left\{ d_f \mid f \in (\mathcal{F}),\ d_f > \degree{f} \right\}$.

    For (3), since the last fall degree is finite by assumption the supremum from (2) is indeed a maximum.
    Now let $d = \max \left\{ d_f \mid f \in (\mathcal{F}),\ d_f > \degree{f} \right\}$ and fix $f \in (\mathcal{F})$.
    If $d_f > \degree{f}$, then $f \in W_{\mathcal{F}, d_f} \subset W_{\mathcal{F}, d}$ and $d = \max \{ d, \degree{f} \}$ since $d \geq d_f > \degree{f}$.
    If $d_f = \degree{f}$, then we always have that $W_{\mathcal{F}, \degree{f}} \subset W_{\mathcal{F}, \max \left\{ d, \degree{f} \right\}}$.
    Thus, for all $f \in (\mathcal{F})$ we have $f \in W_{\mathcal{F}, \max \left\{ d, \degree{f} \right\}}$.
    So $d \in \{ e \in \mathbb{Z}_{\geq 0} \cup \{ \infty \} \mid f \in W_{\mathcal{F}, \max \left\{ e, \degree{f} \right\}}  \text{ for all } f \in (\mathcal{F}) \}$ which implies that $d_\mathcal{F} \leq d$.
\end{proof}

Unsurprisingly, the last fall degree can also be considered as a measure of the complexity of solving polynomial systems.
Let $\maxGBdeg_> \left( \mathcal{F} \right)$ denote the maximal degree of the polynomials appearing in the reduced $>$-Gr\"obner basis of $\mathcal{F}$.
\begin{lem}
    Let $K$ be a field, and let $\mathcal{F} \subset P = K [x_1, \dots, x_n]$ be a polynomial system with $d_\mathcal{F} < \infty$.
    Then
    \[
        \solvdeg_{DRL} \left( \mathcal{F} \right) \leq \max \big\{ d_\mathcal{F}, \maxGBdeg_{DRL} (\mathcal{F}) \}.
    \]
\end{lem}
\begin{proof}
    Let $g \in \mathcal{G}$ be an element of the reduced DRL Gr\"obner basis of $\mathcal{F}$.
    If $g$ has a degree fall, then by \Cref{Prop: alternatice charcterizations of last fall degree} $g \in W_{\mathcal{F}, d_g} \subset W_{\mathcal{F}, d_\mathcal{F}}$.
    If $g$ does not have a degree fall, then $g \in W_{\mathcal{F}, \degree{g}}$.
    Thus, the upper bound follows by taking the maximum over the last fall degree and the maximal degree in the reduced DRL Gr\"obner basis.
\end{proof}

If a polynomial system is in generic coordinates, then one can guarantee that the last fall degree is finite.
\begin{thm}\label{Th: last fall degree finite in generic coordinates}
    Let $K$ be an algebraically closed field, and let $\mathcal{F} = \{ f_1, \dots, f_m \} \subset K [x_1, \dots, x_n]$ be an inhomogeneous polynomial system such that $\left( \mathcal{F}^\homog \right)$ is in generic coordinates and $\left| \mathcal{Z}_+ \left( \mathcal{F}^\homog \right) \right| \neq 0$.
    If $d \geq \sat \left( \mathcal{F}^\homog \right)$ is an integer, then
    \[
        (\mathcal{F})_{\leq d} = W_{\mathcal{F}, d}.
    \]
    In particular,
    \[
        d_\mathcal{F} = \sat \left( \mathcal{F}^\homog \right).
    \]
\end{thm}
\begin{proof}
    As always we abbreviate $P = K [x_1, \dots, x_n]$.
    $W_{\mathcal{F}, d} \subset (\mathcal{F})_{\leq d}$ is trivial, so let $f \in (\mathcal{F})_{\leq d}$.
    Recall that by \Cref{Lem: saturation equal homogenization} $(\mathcal{F})^\homog = \left( \mathcal{F}^\homog \right)^{\sat}$.
    Therefore, by definition of the satiety $x_0^{d - \degree{f}} \cdot f^\homog \in (\mathcal{F})_d^\homog = \left( \mathcal{F}^\homog \right)_d$.
    So we can construct $x_0^{d - \degree{f}} \cdot f^\homog = \sum_{i = 1}^{m} g_i \cdot f_i^\homog$, where $g_i$ homogeneous and $\degree{g_i \cdot f_i^\homog} = d$.
    Then by \cite[Proposition~4.3.2]{Kreuzer-CompAlg2}
    \[
    f = \left( x_0^{d - \degree{f}} \cdot f^\homog \right)^\dehom = \sum_{i = 1}^{m} g_i^\dehom \cdot f_i,
    \]
    where $\degree{g_i^\dehom \cdot f_i} \leq d$.
    So $f \in W_{\mathcal{F}, d}$.
    Therefore, $W_{\mathcal{F}, d} \cap P_{\leq d - 1} = (\mathcal{F})_{\leq d} \cap P_{\leq d - 1} = (\mathcal{F})_{\leq d - 1} = W_{\mathcal{F}, d - 1}$ for all $d > \sat \left( \mathcal{F}^\homog \right)$.
    So by \Cref{Prop: alternatice charcterizations of last fall degree} we can conclude that $d_\mathcal{F} \leq \sat \left( \mathcal{F}^\homog \right)$.

    For the second claim let $f \in (\mathcal{F})_{\leq d_\mathcal{F}}$, then it can be constructed as
    \[
    f = \sum_{i = 1}^{m} g_i \cdot f_i,
    \]
    where $\degree{g_i \cdot f_i} \leq d_\mathcal{F}$.
    Let $\hat{d} = \max_{1 \leq i \leq m} \degree{g_i \cdot f_i} \leq  d_\mathcal{F}$.
    Then by \cite[Proposition~4.3.2]{Kreuzer-CompAlg2}
    \[
    x_0^{\hat{d} - \degree{f}} \cdot f^\homog = \sum_{i = 1}^{m} x_0^{\hat{d} - \degree{f_i \cdot g_i}} \cdot g_i^\homog \cdot f_i^\homog.
    \]
    Multiplying this equation by $x_0^{d_\mathcal{F} - \hat{d}}$ lifts it to $\left( \mathcal{F} \right)^\homog_{d_\mathcal{F}}$.
    Since $f \in (\mathcal{F})_{\leq d_\mathcal{F}}$ was arbitrary we can then conclude that
    \[
        \left( \mathcal{F}^\homog \right)_{d_\mathcal{F}}
        = \left\{ x_0^{d_\mathcal{F} - \degree{f}} \cdot f^\homog \; \middle\vert \; f \in (\mathcal{F})_{\leq d_\mathcal{F}} \right\}
        = \left( \mathcal{F} \right)^\homog_{d_\mathcal{F}}.
    \]
    Obviously, the latter equality extends to all $d \geq d_\mathcal{F}$, so by minimality of the saturation we also have that $\sat \left( \mathcal{F}^\homog \right) \leq d_\mathcal{F}$.
\end{proof}

\begin{cor}
    In the situation of \Cref{Th: last fall degree finite in generic coordinates}, if $f \in (\mathcal{F})$ has a degree fall in $d_f$, then
    \[
    d_f \leq \sat \left( \mathcal{F}^\homog \right).
    \]
\end{cor}
\begin{proof}
    This is a consequence of \Cref{Prop: alternatice charcterizations of last fall degree} and \Cref{Th: last fall degree finite in generic coordinates}.
\end{proof}

So by \Cref{Equ: satiety less than regularity} the construction of a polynomial with a degree fall yields a lower bound on the regularity of $\mathcal{F}^\homog$.

\begin{rem}\label{Rem: Huang's last fall degree}
    We note that the first notion of ``last fall degree'' already appeared in Huang et al.\ \cite{C:HuaKosYeo15,Huang-LastFallDegree}.
    They define their last fall degree as follows: Let $\mathcal{F} \subset P = K [x_1, \dots, x_n]$ be a polynomial system the vector space of constructible polynomials $V_{\mathcal{F}, d}$ in degree $d \geq 0$ is defined via
    \begin{enumerate}[label=(\roman*)]
        \item $\{ f \in \mathcal{F} \mid \degree{f} \leq d \} \subset V_{\mathcal{F}, d}$,

        \item if $g \in V_{\mathcal{F}, d}$ and $h \in P$ with $\degree{g \cdot h} \leq d$, then $h \cdot g \in V_{\mathcal{F}, d}$.
    \end{enumerate}
    Analog to \Cref{Def: last fall degree} Huang et al.\ define the last fall degree as
    \[
    \overline{d}_\mathcal{F} = \min \left\{ d \in \mathbb{Z}_{\geq 0} \cup \{ \infty \} \mid f \in V_{\mathcal{F}, \max \left\{ d, \degree{f} \right\}} \text{ for all } f \in (\mathcal{F}) \right\}.
    \]
    For Huang et al.'s last fall degree one can also prove analog characterizations to \Cref{Prop: alternatice charcterizations of last fall degree}, see \cite[Propostion~2.6]{Huang-LastFallDegree} and \cite[Theorem~2.8]{Caminata-Degrees}.
    Moreover, Huang et al.'s last fall degree is always finite \cite[Propostion~2.6]{Huang-LastFallDegree}.

    Huang et al.'s last fall degree can also be interpreted in terms of Macaulay matrices.
    Let $>$ be a degree compatible term order on $P$, and let $M_{\leq d}$ be the inhomogeneous Macaulay matrix for $\mathcal{F}$ in degree $d$.
    First we compute the basis $\mathcal{B}$ of the row space of $M_{\leq d}$ via Gaussian elimination.
    Now we generate the Macaulay matrix $M_{\leq d}$ for $\mathcal{B}$ and again compute the row space basis $\mathcal{B}'$ via Gaussian elimination.
    We repeat this procedure until $\mathcal{B} = \mathcal{B}'$.
    We then denote with $\overline{W}_d$ the row space of the final stationary Macaulay matrix in that procedure.
    It is clear that for $d$ large enough $\mathcal{B}$ contains a Gr\"obner basis, in analogy to \Cref{Def: solving degree} one can define another notion of solving degree $\overline{\solvdeg}_> \left( \mathcal{F} \right)$ to be the minimal $d$ such that the iterated Macaulay matrix construction produces a $>$-Gr\"obner basis for $\mathcal{F}$.
    Gorla et al.\ proved that $\overline{W}_d = V_{\mathcal{F}, d}$ \cite[Theorem~1]{Gorla-StrongerBounds}, and Caminata \& Gorla proved that \cite[Theorem~~3.1]{Caminata-Degrees}
    \[
    \overline{\solvdeg}_> \left( \mathcal{F} \right) = \max \left\{ \overline{d}_\mathcal{F}, \maxGBdeg_> \left( \mathcal{F} \right) \right\}.
    \]

    With the Macaulay matrix interpretation the difference between Huang et al.'s and our last fall degree notion becomes clear.
    Our last fall degree is defined via a single Macaulay matrix in degree $d$, i.e.\ it is a \emph{last fall degree of the first order}.
    Huang et al.'s last fall degree is defined via an iteration of Macaulay matrices in degree $d$, i.e.\ it is a \emph{last fall degree of higher order}.

    Finally, it follows easily from the definitions that
    \begin{align}
        \overline{d}_\mathcal{F} &\leq d_\mathcal{F}, \nonumber \\
        \overline{\solvdeg}_> \left( \mathcal{F} \right) &\leq \solvdeg_> \left( \mathcal{F} \right). \nonumber
    \end{align}
\end{rem}

    \section{Lower Bounds for the Last Fall Degree of Iterated Polynomial Systems}\label{Sec: lower bounds}
In this section we prove lower bounds for the Castelnuovo-Mumford regularity of attacks on \MiMC, Feistel-\MiMC and Feistel-\MiMC-Hash.
Essentially, we will achieve this by constructing S-polynomials with degree falls.

\subsection{Lower Bound for Univariate Keyed Iterated Polynomial Systems With a Field Equation}
Before we present the theorem we first outline our proof strategy for all results in this section.
First we pick a polynomial $f \in (I, g)$, where $I$ is an ideal with known DRL and LEX Gr\"obner bases and $g$ is an additional polynomial, and assume that $f$ does not have a degree fall in some degree $d_f$.
Now we express as sum $f = f_I + f_g \cdot g$, where $f_I \in I$,  that is compatible with $d_f$ and rearrange this equation so that the right-hand side only consists of elements of $I$, i.e.\ $f - f_g \cdot g = f_I$.
Additionally, we reduce $f_g$ modulo $I$ with respect to DRL, so without loss of generality we can assume that no monomial of $f_g$ is an element of $\inid_{DRL} (I)$.
Then we use the LEX Gr\"obner basis of $I$ to transform the left-hand side into a univariate polynomial.
Finally, we compare the degrees of the univariate left-hand side polynomial and the univariate LEX polynomial of $I$.
If $f$ has a degree fall, then we expect that the degree of the left-hand side polynomial is less than the degree of the univariate LEX polynomial, i.e.\ we have constructed a contradiction.
\begin{thm}\label{Th: lower bound for satiety of iterated polynomial systems}
    Let $\Fq$ be a finite field, let $n \geq 2$ be an integer, and let $f_1, \dots, f_n \in \Fq [x_1, \dots, x_{n - 1}, y]$ be a univariate keyed iterated polynomial system such that
    \begin{enumerate}[label=(\roman*)]
        \item $d_i = \deg \left( f_i \right) \geq 2$ for all $1 \leq i \leq n$ and $d_1 \leq q$, and

        \item $f_i$ has the monomial $x_{i - 1}^{d_i}$ for all $2 \leq i \leq n$.
    \end{enumerate}
    Let $f_{n + 1} = y^q - y$ be the field equation for $y$, let $\mathcal{F} = \{ f_1, \dots, f_{n + 1} \}$, and let $\tilde{f}_n \in \Fq [y]$ be the univariate polynomial in the LEX Gr\"obner basis of $(f_1, \dots, f_n)$.
    Further, assume that
    \begin{enumerate}[label=(\roman*),resume]
        \item\label{Item: field equation not in ideal} $f_{n + 1} \notin (f_1, \dots, f_n)$, and

        \item\label{Item: roots assumption} $\tilde{f}_n$ has less than $d_1$ many roots in $\Fq$.
    \end{enumerate}
    Then
    \[
    d_\mathcal{F} \geq q + \sum_{i = 2}^{n} (d_i - 1).
    \]
    Moreover, if $d_i \geq d$ for all $1 \leq i \leq n$, then
    \[
    d_\mathcal{F} \geq q + (n - 1) \cdot (d - 1).
    \]
\end{thm}
\begin{proof}
    Without loss of generality we can assume that $\LC_{DRL} (f_1) = 1$.
    Let $I = (f_1, \dots, f_n)$, and let $x^\gamma = \prod_{i = 1}^{n - 1} x_{i}^{d_{i + 1} - 1}$.
    We consider the S-polynomial
    \[
    s = x^\gamma \cdot S_{DRL} \left( f_1, f_{n + 1} \right),
    \]
    and the degree $d_s = q + \sum_{i = 2}^{n} (d_i - 1)$.
    By Assumption \ref{Item: field equation not in ideal} $S_{DRL} \left( f_1, f_{n + 1} \right) = y^{q - d_1} \cdot f_1 - f_{n + 1}$ has a degree fall in $q$, so we also have that $\deg \left( s \right) < d_s$.
    For a contradiction let us assume that $s$ does not have a degree fall in $d_s$, i.e.,
    \[
    s = \sum_{i = 1}^{n} s_i \cdot f_i + s_{n + 1} \cdot f_{n + 1},
    \]
    where $\deg \left( s_i \cdot f_i \right) < d_s$ for all $1 \leq i \leq n + 1$.
    Expanding the definition for $s$ and by rearranging we yield that
    \[
    \left( x^\gamma + s_{n + 1} \right) \cdot f_{n + 1} = \sum_{i = 1}^{n} \tilde{s}_i \cdot f_i \in I.
    \]
    Via division by remainder we can split $s_{n + 1} = s_I + s_r$, where $s_I \in I$ and no term of $s_r$ lies in $\inid_{DRL} (I)$.
    Note that for a degree compatible term order the degree of polynomials involved in the division by remainder algorithm can never reach $d_s$.
    Now we move $s_I \cdot f_{n + 1}$ to the right-hand side of the equation, so without loss of generality we can assume that no term of $s_{n + 1}$ lies in $\inid_{DRL} (I)$.
    Via the LEX Gr\"obner basis of $I$, see \Cref{Lem: keyed iterated shape lemma I}, we can transform any polynomial in $g \in \Fq [x_1, \dots, x_{n - 1}, y]$ into a univariate polynomial $\hat{g} \in K [y]$ such that
    \[
    g \equiv \hat{g} \mod (f_1, \dots, f_n)
    \]
    by simply substituting $x_i \mapsto \tilde{f}_i$.
    Via the substitution we now obtain univariate polynomials $\hat{f}_\gamma, \hat{f}_{s_{n + 1}} \in \Fq [y]$ such that
    \begin{equation}\label{Equ: substitution}
        \left( x^\gamma + s_{n + 1} \right) \cdot f_{n + 1} \equiv \left( \hat{f}_\gamma + \hat{f}_{s_{n + 1}} \right) \cdot f_{n + 1} \equiv 0 \mod I.
    \end{equation}
    By our assumption that $s$ does not have a degree fall, we have that $\deg \left( s_{n + 1} \right) < \deg \left( x^\gamma \right)$.
    So by \Cref{Prop: lex substitution properties} \ref{Item: degree after substitution} we also have that $\deg \left( \hat{f}_{s_{n + 1}} \right) < \deg \left( \hat{f}_\gamma \right)$.
    Recall that the degree of $\hat{f}_\gamma$ is given by \Cref{Prop: lex substitution properties} \ref{Item: degree},
    \begin{equation}\label{Equ: degree 1}
        \deg \left( \hat{f}_\gamma \right) = \prod_{i = 1}^{n} d_i - d_1.
    \end{equation}
    Combining \Cref{Lem: lex shape lemma} \ref{Item: univariate LEX ideal membership} and \Cref{Equ: substitution} we now conclude that
    \begin{equation}\label{Equ: contained in ideal}
        \left( x^\gamma + s_{n + 1} \right) \cdot f_{n + 1} \in (f_1, \dots, f_n) \Leftrightarrow \left( \hat{f}_\gamma + \hat{f}_{s_{n + 1}} \right) \cdot f_{n + 1} \in \left( \tilde{f}_n \right),
    \end{equation}
    where $\tilde{f}_n$ is the univariate polynomial in the LEX Gr\"obner basis of $I$.
    I.e., $\left( \hat{f}_\gamma + \hat{f}_{s_{n + 1}} \right) \cdot f_{n + 1}$ must be a multiple of $\tilde{f}_n$.
    By Assumption \ref{Item: roots assumption} $\tilde{f}_n$ has less than $d_1$ many roots in $\Fq$ and $f_{n + 1} = \prod_{a \in \Fq} (y - a)$ is a square-free polynomial, so
    \begin{equation} \label{Equ: degree 2}
        \deg \left( \gcd \left( \hat{f}_n, f_{n + 1} \right) \right) < d_1.
    \end{equation}
    Before our final step we recall the following property of univariate polynomial greatest common divisors: If $p, q, r \in K [x]$, $K$ a field, and $\gcd \left( p, r \right) = 1$, then $\gcd \left( p, q \cdot r \right) = \gcd \left( p, q \right)$.
    Combining this property with \Cref{Equ: contained in ideal} we conclude that the following equation must be true
    \begin{align}
        \tilde{f}_n
        &= \gcd \left( \tilde{f}_n, \left( \hat{f}_\gamma + \hat{f}_{s_{n + 1}} \right) \cdot f_{n + 1} \right) \nonumber \\
        &= \gcd \left( \tilde{f}_n, \left( \hat{f}_\gamma + \hat{f}_{s_{n + 1}} \right) \cdot \gcd \left( \tilde{f}_n, f_{n + 1} \right) \right). \nonumber
    \end{align}
    On the other hand, by \Cref{Equ: degree 1,Equ: degree 2} we have that
    \begin{align}
        \degree{\tilde{f}_n }
        = \prod_{i = 1}^{n} d_i
        &\leq \deg \bigg( \!\! \left( \hat{f}_\gamma + \hat{f}_{s_{n + 1}} \right) \cdot \gcd \left( \tilde{f}_n, f_{n + 1} \right) \!\! \bigg) \nonumber \\
        &= \degree{\hat{f}_\gamma + \hat{f}_{s_{n + 1}}} + \deg \Big( \! \gcd \left( \tilde{f}_n, f_{n + 1} \right) \! \Big) \nonumber \\
        &< \prod_{i = 1}^{n} d_i - d_1 + d_1 = \prod_{i = 1}^{n} d_i. \nonumber
    \end{align}
    A contradiction.
\end{proof}
\begin{rem}
    \begin{enumerate}
        \item If the number of roots of $\hat{f}_n$ in $\Fq$ is greater than or equal to $d_1$, then one can still apply the strategy in the proof to obtain a weaker upper bound.
        It suffices to choose $x^\gamma = \prod_{i = j}^{n - 1} x_i^{d_{i + 1}}$ for a suitable $j > 1$ such that the degrees of the polynomials in the final $\gcd$ equation yield a contradiction.

        \item We note that small non-trivial bounds can also be proven without Assumption \ref{Item: roots assumption}.
        In particular, one can prove that
        \begin{enumerate}[label=(\roman*)]
            \item If $d_1 + q < d_n$, then $d_\mathcal{F} \geq q + 1$.

            \item If $q + \prod_{i = 1}^{n - 1} d_i < d_n$, then $d_\mathcal{F} \geq q + 2$.
        \end{enumerate}
        One considers the polynomials $x_1 \cdot S_{DRL} (f_1, f_{n + 1})$ and $x_{1}^2 \cdot S_{DRL} (f_1, f_{n + 1})$ respectively, and then applies the same strategy as in the proof of \Cref{Th: lower bound for satiety of iterated polynomial systems} to deduce that these polynomials have degree falls.
    \end{enumerate}
\end{rem}
Let us now apply the lower bound to \MiMC.
\begin{ex}[\MiMC and one field equation II]\label{Ex: MiMC solving degree II}
    Let $\MiMC$ be defined over $\Fq$, and let $r$ be the number of rounds.
    The first two conditions of \Cref{Th: lower bound for satiety of iterated polynomial systems} are trivially satisfied by \MiMC.
    For the third assumption, if we consider $\hat{f}_n$, the univariate polynomial in the LEX Gr\"obner basis, as random polynomial, then for $q$ large enough it has on average only one root in $\Fq$ (cf.\ \cite{Leontev-Roots}).
    Thus, with high probability we can assume that \MiMC has only one root in $\Fq$.
    Now we pick a random $k \in \Fq$ and evaluate whether $\MiMC (p, k) = c$ or not.
    If the equality is true we can return $K$ as proper key guess, otherwise it implies that $f_{n + 1} \notin I_\MiMC$.
    So we can combine \Cref{Ex: MiMC solving degree I} and \Cref{Th: lower bound for satiety of iterated polynomial systems} to obtain the following range for the solving degree of \MiMC and one field equation
    \[
    \MiMCfe.
    \]
    Small scale experiments indicate that the solving degree of this attack is always equal to $q + 2 r - 1$.
\end{ex}

For any polynomial system $\mathcal{F} \subset P$ such that $\mathcal{F}^\homog$ is in generic coordinates we have by \Cref{Cor: degree of regularity finite} and \cite[Theorem~5.3]{Caminata-Degrees} that
\begin{equation}\label{Equ: degree of regularity and regularity}
    d_{\reg} \left( \mathcal{F} \right) \leq \reg \left( \mathcal{F}^\homog \right).
\end{equation}
Obviously, this bound also applies to the scenario of \Cref{Th: lower bound for satiety of iterated polynomial systems}, though under the assumptions of the theorem $\mathcal{F}^\topcomp \subset P$ is a DRL Gr\"obner basis since $f_1^\topcomp = y^{d_1} \mid y^q$.
Therefore, $\inid_{DRL}  \left( \mathcal{F}^\topcomp \right) = \left( y^{d_1}, x_1^{d_2}, \ldots, x_{n - 1}^{d_n} \right)$.
Note that a homogeneous ideal and its DRL initial ideal have to have the same degree of regularity, and it is easy to see that
\begin{equation}\label{Equ: Macaulay bound highest degree components}
    d_{\reg} \left( \mathcal{F}^\topcomp \right) = \sum_{i = 1}^{n} d_i - n + 1,
\end{equation}
i.e.\ the degree of regularity is equal to the Macaulay bound of the keyed iterated polynomial system.

Recall from \Cref{Sec: field equation attack} that we can always replace $y^q - y$ by its remainder $r_y$ modulo $I_\text{MiMC}$ with respect to DRL.
For \MiMC experimentally we observed that the highest degree component of $r_y$ is always a monomial and $y^{d_1 - 1} \mid r_y$.
To compute the degree of regularity one then computes the DRL Gr\"obner basis of $\left( \mathcal{F}^\topcomp \right)$ and utilizes it to compute the Hilbert series $h$ of $\inid_{DRL} \left( \mathcal{F}^\topcomp \right)$.
The degree of regularity is then given by $\degree{h} + 1$.

Under some additional assumptions on $r_y$ we can adapt the proof of \Cref{Th: lower bound for satiety of iterated polynomial systems}.
Suppose that the highest degree component of $r_y$ is of the form $y^{d_1 - 1} \cdot \prod_{i = 1}^{k} x_i^{d_{i + 1} - 1}$ for some $k \leq n - 2$.
We set $x^\gamma = \prod_{i = j}^{n - 1} x_i^{d_{i + 1} - 1}$, where $j \geq k + 1$, and consider the S-polynomial
\begin{equation}
    s = x^\gamma \cdot S_{DRL} (f_1, r_y) = x^\gamma \cdot \left( f_1 \cdot \prod_{i = 1}^{k} x_i^{d_{i + 1} - 1} - y \cdot r_y \right).
\end{equation}
Again we assume that $s$ does not have a degree fall in $d_s = \degree{r_y} + \sum_{i = j + 1}^{n} (d_i - 1) + 1$.
By rearranging we then yield that
\begin{equation}
    \left( y \cdot x^\gamma + s_y \right) \cdot r_y \in I = (f_1, \dots, f_n),
\end{equation}
where $\degree{s_y} \leq \degree{x^\gamma}$.
Now we transform again to univariate polynomials in $y$ via the LEX Gr\"obner basis.
Obviously, $r_y \equiv y^q - y \mod I$, and the univariate degree of $y \cdot x^\gamma + s_y$ can again be computed by \Cref{Prop: lex substitution properties} $\degree{\hat{f}_\gamma} = \prod_{i = 1}^{n} d_i - \prod_{i = 1}^{j} d_i + 1$.
Provided that
\begin{equation}
    \begin{split}
        \prod_{i = 1}^{n} d_i - \prod_{i = 1}^{j} d_i + 1 + \deg \Big( \gcd \big( y^q - y, \tilde{f}_{n + 1} \Big) \Big) &< \prod_{i = 1}^{n} d_i \\
        \Leftrightarrow \deg \Big( \gcd \big( y^q - y, \tilde{f}_{n + 1} \Big) \Big) &< \prod_{i = 1}^{j} d_i - 1
    \end{split}
\end{equation}
we can then again construct a contradiction via the greatest common divisor.
Under these additional assumptions one then has the lower bound
\begin{equation}
    d_\mathcal{F} \geq \degree{r_y} + \sum_{i = j + 1}^{n} (d_i - 1) + 1.
\end{equation}
In case of \MiMC, if there is a unique solution for the key variable, then we obtain the lower bound
\begin{equation}
    d_\mathcal{F} \geq \degree{r_y} + 2 \cdot r - 1,
\end{equation}
and if there is less than $8$ solutions for the key variable, then we obtain the lower bound
\begin{equation}
    d_\mathcal{F} \geq \degree{r_y} + 2 \cdot r - 3.
\end{equation}

\subsection{Lower Bound for the Two Plain/Ciphertext Attack of Univariate Keyed Iterated Polynomial Systems}
Next we turn to the attack with two plain/ciphertexts.
For this lower bound we work with the degree of regularity and \cite[Theorem~5.3]{Caminata-Degrees}.
\begin{thm}\label{Th: variety intersection iterated polynomial system}
    Let $\Fq$ be a finite field, let $n \geq 1$ be an integer, and let
    \[
    \begin{split}
        f_1, \dots, f_n &\in \Fq [u_1, \dots, u_n, y], \text{ and} \\
        h_1, \dots, h_n &\in \Fq [v_1, \dots, v_n, y]
    \end{split}
    \]
    be two univariate keyed iterated polynomial systems which are constructed with the same $g_1, \dots, \allowbreak g_n \in \Fq [x, y]$ but have different plain/ciphertext pairs $(p_1, c_1), (p_2, c_2) \in \Fq^2$.
    Assume that
    \begin{enumerate}[label*=(\roman*)]
        \item\label{Item: degrees greater than 2} $d_i = \degree{g_i} \geq 2$ for all $1 \leq i \leq n$,

        \item\label{Item: leading monomials} $g_i$ has the monomial $x_{i - 1}^{d_i}$ for all $2 \leq i \leq n$, and
    \end{enumerate}
    Then for the polynomial system $\mathcal{F} = \left\{ f_1, \dots, f_n, h_1, \dots, h_n \right\} \subset \Fq [u_1, \dots, u_n, v_1, \allowbreak \dots, v_n, y]$ we have that
    \[
    \reg \left( \mathcal{F}^\homog \right) \geq 2 \cdot \left( \sum_{i = 1}^{n} (d_i - 1) \right) - d_1.
    \]
    Moreover, if $\deg \left( g_i \right) \geq d$ for all $1 \leq i \leq n$, then
    \[
    \reg \left( \mathcal{F}^\homog \right) \geq 2 \cdot n \cdot \left( d - 1 \right) - d.
    \]
\end{thm}
\begin{proof}
    We have that $\left( \mathcal{F}^\topcomp \right)$ is a DRL Gr\"obner basis since $f_1^\topcomp = h_1^\topcomp$ and that $\inid_{DRL} \left( \mathcal{F}^\topcomp \right) \allowbreak = \left( y^{d_1}, u_1^{d_2}, \ldots, u_{n - 1}^{d_n}, v_1^{d_2}, \dots, v_{n - 1}^{d_n} \right)$, therefore
    \begin{equation}
        d_{\reg} \left( \mathcal{F}^\topcomp \right) = d_1 - 1 + 2 \cdot \sum_{i = 2}^{n} (d_i - 1) + 1 = 2 \cdot \sum_{i = 1}^{n} (d_i - 1) - d_1,
    \end{equation}
    and the claim follows from \cite[Theorem~5.3]{Caminata-Degrees}.
\end{proof}
\begin{rem}
    Under the additional assumption that $\deg \big( S_{DRL} (f_1, h_1) \big) \geq 2$ it can also be proven that
    \[
        d_\mathcal{F} \geq 2 \cdot \left( \sum_{i = 1}^{n} (d_i - 1) \right) - d_1.
    \]
    Though, it requires slightly more effort
\end{rem}

Let us apply this theorem to \MiMC.
\begin{ex}[\MiMC and two plain/ciphertext pairs II]\label{Ex: MiMC two plaintext solving degree II}
    Let $\Fq$ be a finite field of odd characteristic, let $\MiMC$ be defined over $\Fq$, and let $r$ be the number of rounds.
    Let $(p_1, c_1), (p_2, c_2) \in \Fq^2$ be two different plain/ciphertext pairs given by \MiMC.
    By \Cref{Ex: MiMC two plaintext attack I} and \Cref{Th: variety intersection iterated polynomial system} we have the following range for the Castelnuovo-Mumford regularity of this attack
    \[
    \MiMCtwo.
    \]
    Moreover, small scale experiments indicate that the solving degree of this attack is always equal to $4 \cdot r$.
\end{ex}

\subsection{Lower Bound for Feistel-\texorpdfstring{$2n/n$}{2n/n}}\label{Sec: Feistel-MiMC lower bound}
Recall that the DRL Gr\"obner basis of Feistel-\MiMC, see \Cref{Prop: Feistel Groebner bases} \ref{Item: Feistel Groebner basis}, is almost a univariate keyed iterated polynomial system.
Therefore, we can utilize the same strategy as for \MiMC and a field equation to prove the lower bound for Feistel-$2n/n$.
\begin{thm}\label{Th: Feistel-MiMC lower bound}
    Let $\Fq$ be a finite field, let $n \geq 2$ be an integer, and let $\mathcal{F} = \{ f_{L, 1}, f_{R, 1}, \dots, \allowbreak f_{L, n}, \allowbreak f_{R, n} \} \subset \Fq [x_{L, 1}, x_{R, 1}, \dots, x_{L, n - 1}, x_{R, n - 1}, y]$ be a keyed iterated polynomial system for Feistel-$2n/n$ such that
    \begin{enumerate}[label=(\roman*)]
        \item\label{Item: degree assumption Feistel} $d_i = \deg \left( f_{L, i} \right) \geq 2$ for all $1 \leq i \leq n$,

        \item\label{Item: leading monomials Feistel} $f_{i, L}$ has the monomial $x_{L, i - 1}^{d_i}$ for all $2 \leq i \leq n$,

        \item\label{Item: monomial assumption last polynomial} $d_1 \leq d_n$ and $f_{L, n}$ has the monomial $y^{d_n}$, and

        \item\label{Item: gcd assumption} the greatest common divisor of the univariate polynomials in $y$ that represent the left and the right branch have degree less than $d_1$.
    \end{enumerate}
    Then
    \[
    d_\mathcal{F} \geq d_n + \sum_{i = 2}^{n - 1} \left( d_{i} - 1 \right).
    \]
    Moreover, if $\deg \left( f_{L, i} \right) \geq d$ for all $2 \leq i \leq n$, then
    \[
    d_\mathcal{F} \geq d + \left( n - 2 \right) \cdot \left( d - 1 \right).
    \]
\end{thm}
\begin{proof}
    By Assumption \ref{Item: degree assumption Feistel} and \ref{Item: leading monomials Feistel} we can efficiently compute the DRL Gr\"obner basis of $\mathcal{F} \setminus \{ f_{R, n} \}$ with \Cref{Prop: Feistel Groebner bases} \ref{Item: Feistel Groebner basis}.
    Next we remove the linear polynomials from the Gr\"obner basis, we denote this downsized base with $\mathcal{G} = \{ \tilde{f}_{L, 1}, \dots, \tilde{f}_{L, n} \} \subset P = \Fq [ x_{R, 2}, \dots, x_{R, n - 1}, x_{L, n - 1}, y]$.
    Let $x^\gamma = \prod_{i = 2}^{n - 1} x_{R, i}^{d_i - 1}$, and let $t \in P$ be the polynomial which is obtained by substituting $x_{L, n - 1} \mapsto c_R$ into $\tilde{f}_{L, n} = f_{L, n}$.
    Note that this substitution can be constructed via
    \[
    t = \tilde{f}_{L, n} + \tilde{t} \cdot f_{R, n},
    \]
    where $\tilde{t} \in \Fq [x_{L, n - 1}, y]$ and $\LM_{DRL} (\tilde{t}) = x_{L, n - 1}^{d_n - 1}$, and by  Assumption \ref{Item: monomial assumption last polynomial} $\degree{t} = d_n$.
    Now we consider the polynomial
    \[
    s = x^\gamma \cdot S_{DRL} \left( t, \tilde{f}_{L, 1} \right).
    \]
    By Assumption \ref{Item: monomial assumption last polynomial}  $S_{DRL} (f_{L, 1}, t)$ has a degree fall in $d_n$.
    For a contradiction we now assume that $s$ does not have a degree fall in $d_s = d_n + \sum_{i = 2}^{n - 1} \left( d_i - 1 \right)$, i.e.\
    \[
    s = \sum_{i = 1}^{n} s_i \cdot \tilde{f}_{L, i} + s_{n + 1} \cdot f_{R, n} \quad\Longleftrightarrow\quad \left( x^\gamma \cdot t - s_{n + 1} \right) \cdot f_{R, n} = \sum_{i = 1}^{n} \tilde{s}_i \cdot \tilde{f}_{L, i} \in (\mathcal{G}),
    \]
    where $\deg \left( s_i \cdot \tilde{f}_{L, i} \right) < d_s$ for all $1 \leq i \leq n$ and $\deg \left( s_{n + 1} \right) < d_s - 1$.
    By expanding $t$ we can further rewrite the last equation as
    \[
    \left( x^\gamma \cdot \tilde{t} - s_{n + 1} \right) \cdot f_{R, n} \in (\mathcal{G}).
    \]
    Without loss of generality we can assume that no monomial present in $x^\gamma \cdot \tilde{t}$ and $s_{n + 1}$ is an element of $\inid_{DRL} (\mathcal{G})$.
    Note that by construction
    \begin{align}
        \LM_{DRL} \left( x^\gamma \cdot \tilde{t} \right) &= x_{L, n - 1}^{d_n - 1} \cdot \prod_{i = 2}^{n - 1} x_{R, i}^{d_i - 1}, \nonumber \\
        \degree{s_{n + 1}} &< d_s - 1 = \degree{x^\gamma \cdot \tilde{t}}. \nonumber
    \end{align}
    With the LEX Gr\"obner basis of $(\mathcal{G})$, see \Cref{Prop: Feistel Groebner bases} \ref{Item: Feistel LEX Groebner Basis}, we now construct univariate polynomials $\hat{f}_\gamma, \hat{f}_{s_{n - 1}}, \hat{f}_R, \hat{t} \in \Fq [y]$ such that
    \[
    x^\gamma \equiv \hat{f}_\gamma,\
    s_{n + 1} \equiv \hat{f}_{s_{n + 1}},\
    f_{R, n} \equiv \hat{f}_R,\
    \tilde{t} \equiv \hat{t} \mod \left( \mathcal{G} \right).
    \]
    By \Cref{Prop: Feistel Groebner bases} \ref{Item: Feistel degree inequality} the leading monomial of $x^\gamma \cdot \tilde{t}$ has the largest univariate degree among all monomials in $m \in P \setminus \inid_{DRL} (\mathcal{G})$ with $\degree{m} \leq \degree{x^\gamma \cdot \tilde{t}}$, therefore by \Cref{Prop: Feistel Groebner bases} \ref{Item: Feistel degree}
    \[
    \degree{\hat{f}_\gamma \cdot \hat{t} - \hat{f}_{s_{n + 1}}} = \prod_{i = 1}^{n} d_i - d_1.
    \]
    Denote with $\hat{f}_L$ the univariate polynomial in the LEX Gr\"obner basis of $(\mathcal{G})$, this is exactly the polynomial that describes encoding in the left branch of Feistel-$2n/n$.
    Similar the univariate polynomial $\hat{f}_R \in \Fq [y]$ equivalent to $f_{R, n}$ represents encoding in the right branch of Feistel-$2n/n$.
    By \Cref{Lem: lex shape lemma} and elementary properties of the polynomial greatest common divisor the following equality must be true
    \begin{align}
        f_L
        &= \gcd \bigg( \hat{f}_L, \Big( \hat{f}_\gamma \cdot \hat{t} - \hat{f}_{s_{n + 1}} \Big) \cdot \hat{f}_R \bigg) \nonumber \\
        &= \gcd \bigg( \hat{f}_L, \Big( \hat{f}_\gamma \cdot \hat{t} - \hat{f}_{s_{n + 1}} \Big) \cdot \gcd \Big( \hat{f}_L, \hat{f}_R \Big) \bigg). \nonumber
    \end{align}
    On the other hand, by Assumption \ref{Item: gcd assumption} $\gcd \left( \hat{f}_L, \hat{f}_R \right)$ has degree less than $d_1$.
    So with \Cref{Prop: Feistel Groebner bases} \ref{Item: Feistel degrees in LEX Groebner basis} we have the following inequality
    \begin{align}
        \degree{\hat{f}_L}
        = \prod_{i = 1}^{n} d_i
        &\leq \deg \bigg( \Big( \hat{f}_\gamma \cdot \hat{t} - \hat{f}_{s_{n + 1}} \Big) \cdot \gcd \Big( \hat{f}_L,  \hat{f}_R \Big) \bigg) \nonumber \\
        &= \degree{\hat{f}_\gamma \cdot \hat{t} - \hat{f}_{s_{n + 1}}} + \deg \bigg( \gcd \Big( \hat{f}_L,  \hat{f}_R \Big) \bigg) \nonumber \\
        &< \prod_{i = 1}^{n} d_i - d_1 + d_1
        = \prod_{i = 1}^{n} d_i. \nonumber
    \end{align}
    A contradiction.
\end{proof}

Applying the theorem to \MiMC-$2n/n$ we obtain the following range on the regularity.
\begin{ex}[\MiMC-2n/n II]\label{Ex: Feistel MiMC solving degree II}
    Let $\MiMC$-$2n/n$ be defined over $\Fq$, and let $r$ be the number of rounds.
    We construct the downsized DRL polynomial system from \Cref{Prop: Feistel Groebner bases} $\tilde{\mathcal{F}} \cup \{ f_{L, r} \}$ and embed it into the polynomial ring which has only the variables present in the system.
    Let $f_L, f_R \in \Fq [y]$ be the univariate polynomials that represent encryption in the left and the right branch.
    If we consider them as random polynomials and divide them with $y - k$, where $k \in \Fq$ is the key, then with high probability they are coprime.
    Combining \Cref{Ex: Feistel MiMC I} and \Cref{Th: Feistel-MiMC lower bound} we now obtain the following range for the Castelnuovo-Mumford regularity of \MiMC-2n/n
    \[
    \FeistelMiMC.
    \]
    Small scale experiments indicate that the solving degree of this attack is always equal to $2 \cdot r$.
\end{ex}

Let us again compare \Cref{Th: Feistel-MiMC lower bound} to \Cref{Equ: degree of regularity and regularity}, since $\inid_{DRL} \left( \mathcal{F}^\topcomp \right) = \Big( y_1^{d_1}, x_{R, 2}^{d_2}, \allowbreak \dots, x_{R, n - 1}^{d_{n - 1}}, y^{d_n}, x_{L, n - 1} \Big)$ we have
\begin{equation}
    d_{\reg} \left( \mathcal{F}^\topcomp \right) = \sum_{i = 1}^{n - 1} (d_i - 1) + 1.
\end{equation}
So if $d_n > d_1$, then the bound from the theorem is an improvement.

\subsection{Lower Bound for Feistel-Hash}
We have seen in \Cref{Prop: Feistel Groebner bases} and \Cref{Sec: Feistel-MiMC-Hash upper bounds} that the LEX Gr\"obner basis of the preimage attack of Feistel-\MiMC-Hash has the shape of \Cref{Lem: lex shape lemma}.
Further, we had to include the field equation for the variable $x_2$ to remove the parasitic solutions from the algebraic closure of $\Fq$.
Consequently, to prove a lower bound on the last fall degree we have a mix of the situations in \Cref{Th: lower bound for satiety of iterated polynomial systems,Th: Feistel-MiMC lower bound}.
At this point we expect the reader to be familiar with our techniques, therefore we just mention the polynomials for which it can be proven that they have a degree fall.
\begin{thm}\label{Th: Feistel-MiMC-Hash lower bound}
    Let $\Fq$ be a finite field, let $n \geq 3$ be an integer, and let $\{ f_{L, 1}, f_{R, 1}, \dots, \allowbreak f_{L, n}, \allowbreak f_{R, n} \} \subset K [x_{L, 1}, x_{R, 1}, \dots, x_{L, n - 1}, \allowbreak x_{R, n - 1}, x_1, x_2]$ denote the keyed iterated polynomial system for the Feistel-$2n/n$-Hash preimage attack
    \[
    \text{Feistel-}2n/n
    \begin{pmatrix}
        x_1 \\ 0
    \end{pmatrix}
    =
    \begin{pmatrix}
        \alpha \\ x_2
    \end{pmatrix}
    ,
    \]
    where $\alpha \in \Fq$.
    Assume that the keyed iterated polynomial system of Feistel-$2n/n$-Hash is such that
    \begin{enumerate}[label=(\roman*)]
        \item $d_i = \deg \left( f_{L, i} \right) \geq 2$ for all $2 \leq i \leq n$,

        \item $f_i$ has the monomial $x_{L, i - 1}^{d_i}$ for all $2 \leq i \leq n$, and

        \item the univariate polynomial $\tilde{f} \in \Fq [x_2]$ of the LEX Gr\"obner basis of Feistel-$2n/n$ has less than $d_2$ many roots in $\Fq$.
    \end{enumerate}
    Then for the polynomial system $\mathcal{F} = \left\{ f_{L, 1}, f_{R, 1}, \dots, f_{L, n}, f_{R, n}, x_2^q - x_2 \right\}$ we have that
    \[
    d_\mathcal{F} \geq q + \sum_{i = 3}^{n} \left( d_i - 1 \right).
    \]
    Moreover, if $\degree{f_{L, i }} \geq d$ for all $3 \leq i \leq n$, then
    \[
    d_\mathcal{F} \geq q + \left( n - 3 \right) \cdot \left( d - 1 \right).
    \]
\end{thm}
\begin{proof}[Sketch of proof]
    As a preparation one has to extend \Cref{Prop: Feistel Groebner bases} to Feistel-Hash.
    To do so one sets $y = 0$ and introduces two variables $x_1, x_2$ and sets $p_L = x_1$, $p_R = 0$, $c_L = \alpha$, where $\alpha$ is the hash value, and $c_R = x_2$.
    Now one orders the variables as $x_{R, n - 1} > x_{L, n - 1} > \ldots > x_{R, 1} > x_{L, 1} > x_1 > x_2$ for the DRL and LEX term order.
    Now one can extend \Cref{Prop: Feistel Groebner bases} \ref{Item: Feistel Groebner basis}-\ref{Item: Feistel degree inequality} to Feistel-Hash.

    We denote with $g \in \mathcal{G}$ the polynomial in the DRL Gr\"obner basis with leading monomial $y_1^{d_2}$.
    Let
    \[
    x^\gamma = x_1^{d_n - 1} \cdot \prod_{i = 3}^{n - 1} x_{R, i}^{d_i - 1},
    \]
    then the polynomial
    \[
    s = x^\gamma \cdot S_{DRL} \left( g, y_1^q - y_1 \right)
    \]
    has a degree fall in $q + (d_n - 1) + \sum_{i = 3}^{n - 1} \left( d_i - 1 \right)$.
\end{proof}

For the attacks on Feistel-\MiMC-Hash we now obtain the following regularity ranges.
\begin{ex}[Feistel-\MiMC-Hash preimage attack II]\label{Ex: Feistel-MiMC-Hash preimage II}
    Let Feistel-\MiMC-Hash be defined over $\Fq$, and let $r$ be the number of round.
    Under the assumptions of \Cref{Th: Feistel-MiMC-Hash lower bound} we obtain with \Cref{Ex: Feistel-MiMC-Hash preimage I} the following range for the Castelnuovo-Mumford regularity of the Feistel-\MiMC-Hash preimage attack together with a field equation
    \[
    \Hashpre.
    \]
    Small scale experiments indicate that the solving degree of the preimage attack is always equal to $q + 2 r - 3$.
\end{ex}

Like for \MiMC and the field equation we can replace $x_2^q - x_2$ by its remainder and obtain a lower bound on $d_\mathcal{F}$ via \Cref{Equ: degree of regularity and regularity}.

    \section{Discussion}\label{Sec: discussion}
In this paper we utilized a rigorous mathematical framework to prove Gr\"obner basis complexity estimates for various AO designs.
For \Hades and the \GMiMC family we proved that the Gr\"obner basis cryptanalysis of these designs is indeed mathematically sound.
Our analysis of the \MiMC family revealed that for mildly overdetermined systems we can compute small ranges for the Castelnuovo-Mumford regularity, hence putting a limit on the capabilities of regularity-based solving degree estimates.
Arguably, since our regularity/solving degree estimates for \MiMC polynomial systems that involve field equations exceed the size of the underlying field, these bounds do not have direct cryptographic implications.
Instead, they should be viewed as showcase that for well-behaved cryptographic polynomial systems provable upper as well as lower bounds for the regularity are achievable.
Moreover, as we discussed below \Cref{Ex: MiMC solving degree I,Ex: MiMC solving degree II} these bounds can be significantly improved via an auxiliary division by remainder computation.
The reason why we did not work with the remainder directly is quite simple: For every possible \MiMC instantiation and plain/ciphertext sample the remainder polynomial is different.
So unless one can reveal structural properties of the remainder polynomial one has to do an individual analysis for every possible instantiation.
On the other hand, by working with the field equation itself we could keep our analysis generic.

To the best of our knowledge this paper is the first time that AO Gr\"obner basis analysis has been performed without evasion to assumptions and hypotheses that could fail in practice.
Of course, from an AO designer's point of view this raises whether more advanced AO primitives are also provable in generic coordinates.
We point out that recent designs like \ReinforcedConcrete \cite{CCS:GKLRSW22}, \Anemoi \cite{C:BBCPSVW23}, \Griffin \cite{C:GHRSWW23} and \Arion \cite{Arion} have deviated heavily from classical design strategies, and these deviations seem to be in conflict with elementary applications of the Caminata-Gorla technique.
For example one of the \ReinforcedConcrete permutations over $\Fp$ is of the form
\begin{equation}
    \begin{pmatrix}
        x_1 \\ x_2 \\ x_3
    \end{pmatrix}
    \mapsto
    \begin{pmatrix}
        x_1^d \\[2pt]
        x_2 \cdot \left( x_1^2 + \alpha_1 \cdot x_1 + \beta_1 \right) \\[2pt]
        x_3 \cdot \left( x_2^2 + \alpha_2 \cdot x_2 + \beta_2 \right)
    \end{pmatrix}
    ,
\end{equation}
where $d \in \mathbb{Z}_{> 1}$ such that $\gcd \left( d, p - 1 \right) = 1$ and $\alpha_i, \beta_i \in \Fp$ are such that $\alpha_i^2 - 4 \cdot \beta_i$ are non-squares in $\Fp$.
Let us naively apply the Caminata-Gorla technique for this permutation.
After homogenizing it and substituting $x_0 = 0$ we yield that $x_1^d = x_2 \cdot x_1^2 = x_3 \cdot x_2^2 = 0$, but it is not true that $x_1 = x_2 = x_3 = 0$ is the only solution over $\overline{\Fp}$ to these equations.
Hence, our proving technique for generic coordinates fails.
We also want to point out that we face a similar situation for \Griffin and \Arion.

For all our regularity lower bounds we were given a DRL Gr\"obner basis together with an additional polynomial.
Via careful analysis of the arithmetic of the polynomial systems we could then discover polynomials with degree falls.
Of course, we would like to provide lower bounds in the presence of two or more additional equations.
Our readers might also recall that the attack on \MiMC with all field equations was missing in \Cref{Sec: lower bounds}.
From small scale experiments we raise the following conjecture for this attack.
\begin{conj}
    Let $\Fq$ a finite field.
    Let $\mathcal{F} = \left\{ f_1, \dots, f_n \right\} \subset P = \Fq [x_1, \dots, \allowbreak x_{n - 1}, y]$ be a keyed iterated system of polynomials such that
    \begin{enumerate}[label=(\roman*)]
        \item $d_i = \deg \left( f_i \right) \geq 2$ for all $1 \leq i \leq n$, and

        \item $f_i$ has the monomial $x_{i - 1}^{d_i}$ for all $2 \leq i \leq n$.
    \end{enumerate}
    Let $F \subset P$ be the ideal of all field equations.
    Further, assume that
    \begin{enumerate}[label=(\roman*),resume]
        \item $(F) \not\subset (f_1, \dots, f_n)$, and

        \item the univariate LEX polynomial has less than $d_1$ many roots in $\Fq$.
    \end{enumerate}
    Then the polynomial
    \[
    \left( \prod_{i = 1}^{n - 1} x_i \right) \cdot S_{DRL} \left( f_1, y^q - y\right)
    \]
    has a degree fall for the polynomial system $\mathcal{F} + F$.
\end{conj}
We expect that a resolution to this \MiMC problem will also reveal insight into the more general cryptographic polynomial systems.

    \section*{Acknowledgments}
    The author would like to thank the anonymous reviewers at ToSC for their valuable comments and helpful suggestions which improved both the quality and presentation of this paper.
    The author would like to thank Arnab Roy for his suggestion to study Gr\"obner basis attacks on Arithmetization-Oriented designs.
    Matthias Steiner has been supported in part by the European Research Council (ERC) under the European Union's Horizon 2020 research and innovation program (grant agreement No.\ 725042).

    \bibliographystyle{alphaurl}
    \bibliography{abbrev3.bib,crypto.bib,literature.bib}

\end{document}